\documentclass[11pt]{article}
\usepackage[utf8]{inputenc}
\usepackage[T1]{fontenc}
\pdfoutput=1
\usepackage{amsmath, amssymb, amsthm, dsfont, mathtools, nicefrac, mleftright}
\usepackage[dvipsnames]{xcolor}
\usepackage{hyperref, cleveref, comment}
\usepackage[shortlabels]{enumitem}
\usepackage[ruled,vlined,linesnumbered]{algorithm2e}
\usepackage[left=1.00in, right=1.00in, top=1.00in, bottom=1.00in]{geometry}
\usepackage{subcaption}
\usepackage{natbib}
\setcitestyle{authoryear,open={(},close={)}} 
\usepackage{titlesec}
\titlespacing*{\paragraph}{0pt}{0.5\baselineskip}{0.5\baselineskip}
\usepackage{etoolbox}
\makeatletter
\patchcmd{\@algocf@start}
  {-1.5em}
  {0pt}
  {}{}
\makeatother

\allowdisplaybreaks

\title{Fixed-Parameter Tractable Submodular Maximization over a Matroid}
\author{
  Shamisa Nematollahi\thanks{CNRS, IRIF, Université Paris Cit\'e. \texttt{\{shamisa, vladu, junyao\}@irif.fr}. SN and AV were partially supported by the French Agence Nationale de la Recherche (ANR) under grant ANR-21-CE48-0016 (project COMCOPT). JZ is supported by a postdoctoral fellowship from the Fondation Sciences Math\'ematiques de Paris.} \and Adrian Vladu\footnotemark[1] \and Junyao Zhao\footnotemark[1]
}
\date{}

\newtheorem{theorem}{Theorem}[section]
\newtheorem{lemma}[theorem]{Lemma}

\newtheorem{claim}[theorem]{Claim}
\newtheorem{definition}[theorem]{Definition}

\newcommand{\acc}{\textnormal{\sffamily acc}}
\newcommand{\alg}{\textnormal{\sffamily ALG}}
\newcommand{\eps}{\varepsilon}
\newcommand{\final}{\textnormal{\sffamily final}}
\newcommand{\poly}{\textnormal{\sffamily poly}}

\def \E {\mathbb{E}}

\newcommand{\M}{{\mathcal M}}
\newcommand{\N}{{\mathbb N}}
\newcommand{\cO}{{\mathcal O}}
\newcommand{\R}{{\mathbb R}}
\newcommand{\cR}{{\mathcal R}}

\newcommand{\Z}{{\mathbb Z}}

\DeclareMathOperator*{\argmax}{arg\,max}
\DeclareMathOperator*{\argmin}{arg\,min}

\newcommand{\ceil}[1]{\left\lceil #1 \right\rceil}
\newcommand{\cupdot}{\mathbin{\mathaccent\cdot\cup}}
\newcommand{\floor}[1]{\left\lfloor #1 \right\rfloor}
\newcommand{\smallfloor}[1]{\big\lfloor #1 \big\rfloor}
\newcommand\0{\kern-1.2pt\vec{\kern1.2pt 0}}

\begin{document}

\maketitle

\begin{abstract}
In this paper, we design fixed-parameter tractable (FPT) algorithms for (non-monotone) submodular maximization subject to a matroid constraint, where the matroid rank $r$ is treated as a fixed parameter that is independent of the total number of elements $n$. We provide two FPT algorithms: one for the offline setting and another for the random-order streaming setting. Our streaming algorithm achieves a $\nicefrac{1}{2}-\eps$ approximation using $\widetilde{\cO}\left(\frac{r}{\poly(\eps)}\right)$ memory, while our offline algorithm obtains a $1-\nicefrac{1}{e}-\eps$ approximation with $n\cdot 2^{\widetilde{\cO}\left(\frac{r}{\poly(\eps)}\right)}$ runtime and $\widetilde{\cO}\left(\frac{r}{\poly(\eps)}\right)$ memory. Both approximation factors are near-optimal in their respective settings, given existing hardness results. In particular, our offline algorithm demonstrates that––unlike in the polynomial-time regime––there is essentially no separation between monotone and non-monotone submodular maximization under a matroid constraint in the FPT framework.
\end{abstract}

\section{Introduction}
Submodular maximization under a matroid constraint is a fundamental problem in combinatorial optimization, where we are given\footnote{We assume that the submodular function $f$ is given by a value oracle that outputs $f(X)$ for input $X\subseteq[n]$, and the matroid $\M$ is given by a membership oracle that answers whether $X\in\M$ for input $X\subseteq[n]$.} a submodular function $f:2^{[n]}\to\R_{\ge0}$ and a matroid $\M\subseteq2^{[n]}$ with rank $r\in\N^+$, and our task is to find a set of elements $X\in\M$ that maximizes the function value $f(X)$. For this problem, there is a well-known gap between monotone (i.e., non-decreasing) and non-monotone submodular functions: For monotone submodular functions, the optimal polynomial-time algorithm achieves a $1-\nicefrac{1}{e}$ approximation to the optimal solution~\citep{CCPV11}, whereas for non-monotone submodular functions, no polynomial-time algorithm can obtain a $0.478$ approximation~\citep{GV11,DV12}, even for the special case of the cardinality constraint~\citep{Qi24}, where the matroid $\M$ simply consists of all subsets of $[n]$ with size at most $r$.

One of the most important paradigms in theoretical computer science for overcoming the limitations of polynomial-time algorithms is parameterized complexity, which treats a specific component of the problem as a parameter that is fixed at a small value. In the context of matroid-constrained submodular maximization, it is natural to assume that the matroid rank $r$ is a small parameter\footnote{For example, in data summarization~\citep{MJK18}, given a large image dataset, an algorithm needs to choose a representative subset of images that is small enough to fit the screen of our smartphones.} and design \emph{fixed-parameter tractable} (FPT) algorithms, i.e., algorithms with runtime $h(r)\cdot \poly(n)$, where $h$ can be any finite function.

For the special case of the cardinality constraint,~\citet{RZ22b} designed an FPT algorithm that guarantees a $0.539$ approximation for non-monotone submodular functions, showing a substantial advantage of FPT algorithms over polynomial-time algorithms. On the hardness side, essentially the only known query complexity lower bound that applies to FPT algorithms is from the classic work of~\citet{NW78}. They proved that any algorithm that achieves a better-than-$(1-\nicefrac{1}{e})$ approximation requires $n^{\Omega(r)}$ queries to the value oracle, a result that holds for both monotone and non-monotone submodular maximization with the cardinality constraint. These results seem to suggest that the gap between monotone and non-monotone submodular maximization (at least for the cardinality constraint) could be narrowed further, or perhaps even closed, in the FPT regime. This brings us to the question:
\begin{quote}
    For matroid/cardinality-constrained non-monotone submodular maximization, what is the best possible approximation ratio that can be achieved by FPT algorithms\footnote{The special case of this question for the cardinality constraint was posed by~\citet{RZ22b}.}? Is there a separation between monotone and non-monotone submodular maximization under a matroid/cardinality constraint in the FPT framework?
\end{quote}

In this paper, we reveal a surprising answer to this question: There is essentially no separation between monotone and non-monotone submodular maximization under a matroid constraint in the FPT regime. Specifically, we design a nearly $(1-\nicefrac{1}{e})$-approximation FPT algorithm (Algorithm~\ref{alg:recursive_continuous_greedy_filtering}) for general (non-monotone) submodular maximization subject to a matroid constraint.
\begin{theorem}[Restatement of Theorem~\ref{thm:offline}]
There is an FPT algorithm that achieves a $1-\nicefrac{1}{e}-\eps$ approximation for maximizing any non-negative submodular function under a matroid constraint, with $n\cdot 2^{\widetilde{\cO}\left(\frac{r}{\poly(\eps)}\right)}$ runtime and $\widetilde{\cO}\left(\frac{r}{\poly(\eps)}\right)$ memory for any $\eps>0$.
\end{theorem}

The bedrock of our $(1-\nicefrac{1}{e}-\eps)$-approximation FPT algorithm is a subroutine (Subroutine~\ref{sub:continuous_greedy_filtering}) that can be simplified into a random-order semi-streaming algorithm (Algorithm~\ref{alg:greedy_filtering}), and it has significant implications for the streaming setting as well. Briefly, a random-order streaming algorithm makes a single pass over elements of $[n]$ in a uniformly random order. During the pass, the algorithm maintains a subset of visited elements in a memory buffer of bounded size, and it is allowed to make unlimited queries to the value oracle of the function $f$ and the membership oracle of the matroid $\M$ for any subset of elements stored in its memory. In the literature, there is particular interest in streaming algorithms that use $\widetilde{\cO}(r)$ memory, which are known as semi-streaming algorithms.

Our random-order semi-streaming algorithm (Algorithm~\ref{alg:greedy_filtering}) achieves a nearly $\nicefrac{1}{2}$ approximation. (We note that a nearly $\nicefrac{1}{2}$ approximation is known to be achievable by a streaming algorithm with $2^{\cO(r)}$ memory, for the more general setting of matroid intersection constraints and adversarial-order streams~\citep{FNSZ22}.) This matches the lower bound established by~\citet{RZ22b}, which asserts that any random-order streaming algorithm exceeding $\nicefrac{1}{2}$ approximation must use $\Omega\left(\frac{n}{r^2}\right)$ memory, even for the special case of the cardinality constraint.
\begin{theorem}[Restatement of Theorem~\ref{thm:streaming}]
There is a random-order semi-streaming algorithm that achieves a $\nicefrac{1}{2}-\eps$ approximation for maximizing any non-negative submodular function under a matroid constraint, using $\widetilde{\cO}\left(\frac{r}{\poly(\eps)}\right)$ memory for any $\eps>0$.
\end{theorem}
In comparison, the previous best random-order semi-streaming algorithm for maximizing monotone submodular functions under a matroid constraint achieves a $\left(\frac{1-e^{-2}}{2}\approx0.432\right)$ approximation~\citep{Shadravan20}. For non-monotone functions, the best known semi-streaming algorithm obtains a $0.1921$ approximation~\citep{FLNSZ25}, though this algorithm applies to the more restrictive adversarial-order streaming setting. We highlight that the algorithms of~\citet{Shadravan20} and~\citet{FLNSZ25} run in polynomial time, whereas our algorithm runs in FPT time.

\subsection{Overview of our approach}
Deferring a more detailed exposition to the main sections, we provide an overview of our algorithms and techniques. As mentioned previously, the core of our $(1-\nicefrac{1}{e}-\eps)$-approximation algorithm (Algorithm~\ref{alg:recursive_continuous_greedy_filtering}) is Subroutine~\ref{sub:continuous_greedy_filtering}, which can be simplified into a $(\nicefrac{1}{2}-\eps)$-approximation semi-streaming algorithm (Algorithm~\ref{alg:greedy_filtering}). We call this subroutine continuous greedy filtering. It is inspired by a filtering technique developed by~\citet{RZ22b} for cardinality-constrained non-monotone submodular maximization, and a fast continuous greedy algorithm designed by~\citet{BV14} for matroid-constrained monotone submodular maximization.

\paragraph{Continuous greedy filtering.} The initial phase of this subroutine maintains a fractional solution over $\frac{1}{\eps}$ epochs. In each epoch $t\in\left[\frac{1}{\eps}\right]$, the subroutine constructs an incremental sequence of integral solutions $S^t_1,\dots,S^t_r\in\M$ through $r$ selection steps. At each step $i\in[r]$, it samples a small subset of elements. Among the elements that can be added to $S^t_{i-1}$, it selects the element $s_i^t$ that maximizes the marginal value relative to the current fractional solution $x^t_{i-1}$ (a ``weighted average'' of the integral solutions $S^1_r,\dots,S^{t-1}_r,S^{t}_{i-1}$). The integral solution $S^t_i$ is then updated to be $S^t_{i-1}\cup \{s_i^t\}$.

In the second phase, the subroutine uses the integral solutions constructed in the first phase as a filter to select a subset $H$ of elements from $[n]$. Specifically, for each element $e\in[n]$, $e$ is included in $H$ if it can be added to $S_{i-1}^t$ and its marginal value with respect to $x^t_{i-1}$ exceeds that of element $s_i^t$, for some $t\in\left[\frac{1}{\eps}\right]$ and $i\in[r]$. A minor detail is that the subroutine rounds the marginal values of $e$ and $s_i^t$ to the closest power of $1+\eps$ before comparing them. This detail will keep the subroutine's memory usage nearly linear in the matroid rank $r$.

This subroutine achieves a key property: We partition the optimal solution $O$ into two sets $O_L:=O\setminus H$ and $O_H:=O\cap H$. Then, the subroutine guarantees that there is a feasible subset of $\bigcup_{t=1}^{1/\eps} S_r^t$ that is a $1-\nicefrac{1}{e}-\eps$ approximation to the partial optimal solution $O_L$.

\paragraph{$(1-\nicefrac{1}{e}-\eps)$-approximation: a recursive approach.} Our $(1-\nicefrac{1}{e}-\eps)$-approximation algorithm applies the continuous greedy filtering subroutine recursively on refined instances. Specifically, after the first invocation of the subroutine on the input instance $(f,\M)$, if the value of $O_H$ is negligible, then by submodularity, $O_L$ must account for almost all the optimal value. In this case, because of the aforementioned property, our algorithm can find a $(1-\nicefrac{1}{e}-\eps)$-approximate solution through an exhaustive search over all subsets of $\bigcup_{t=1}^{1/\eps} S_r^t$. If otherwise, our algorithm enumerates all subsets of $H$ to identify $O_H$ (the algorithm will encounter $O_H$ in some iteration of the enumeration, even though it does not know what $O_H$ is), and then, it recursively applies the above process to a refined instance $(f',\M')$, where the matroid $\M'$ is obtained by contracting the original matroid $\M$ by the set $O_H$ (see Definition~\ref{def:contraction}), and the function $f'$ is defined by $f'(X):=f(X\cup O_H)$. The high-level intuition is that, after sufficiently many recursions, we will eventually reach the first case, where the value of $O_H'$ (the analogue of $O_H$ in the refined instance) is negligible.

\paragraph{$(\nicefrac{1}{2}-\eps)$-approximation semi-streaming algorithm: single-epoch filtering.} Our random-order semi-streaming algorithm simplifies continuous greedy filtering by performing a single-epoch greedy procedure to construct an integral solution in the first $\eps$ fraction of the random stream, followed by a similar filtering phase in the remaining stream. The final solution is obtained by enumerating all subsets of elements selected during these two phases and choosing the best one. In the main body of the paper, we will first present this simpler algorithm and its analysis.

\begin{paragraph}{Comparison to~\citet{RZ22b}.}
Briefly,~\citet{RZ22b} first designed a $(\nicefrac{1}{2}-\eps)$-approximation random-order streaming algorithm with quadratic memory for cardinality-constrained non-monotone submodular maximization, based on the filtering technique. They then identified cases where their streaming algorithm only achieves a $\nicefrac{1}{2}$ approximation, and introduced an offline postprocessing procedure to handle those cases and improve the approximation ratio. Their analysis of this procedure relies on a factor-revealing program to balance different cases. To better leverage this factor-revealing program, they extended their postprocessing procedure to a recursive one, which led to a $0.539$-approximation FPT algorithm.

Our semi-streaming algorithm is a generalization of their streaming algorithm to the matroid constraint, with improved memory usage. Compared to their $0.539$-approximation FPT algorithm, the two main novelties of our $(1-\nicefrac{1}{e}-\eps)$-approximation algorithm, as highlighted above, are (i) the combination of the filtering technique and the continuous greedy algorithm, which enables our subroutine to achieve the aforementioned key property under the matroid constraint, and (ii) a natural recursive procedure that applies the subroutine repeatedly to refine the problem instance. This differs significantly from the recursive procedure of~\citet{RZ22b}, which is used only to postprocess an output produced once by their streaming algorithm.
\end{paragraph}

\subsection{Additional related work}
\paragraph{FPT submodular maximization.}
The parameterized complexity of submodular maximization was first studied by~\citet{Skowron17}, who provided FPT approximation schemes for maximizing $p$-separable monotone submodular functions under the cardinality constraint, assuming that both the cardinality constraint and $p$ are fixed parameters. If the cardinality constraint is the only parameter, without further assumptions, even FPT algorithms cannot obtain a better-than-$(1-\nicefrac{1}{e})$ approximation for monotone submodular functions, both in terms of query complexity~\citep{NW78} and computational complexity under the Gap-Exponential Time Hypothesis ~\citep{CGKLL19,Manurangsi20}. For non-monotone submodular maximization under the cardinality constraint,~\citet{RZ22b} gave a $0.539$-approximation FPT algorithm, while the current best polynomial-time algorithm achieves a $0.401$ approximation~\citep{BF24}, which also applies to the matroid constraint.

\paragraph{Streaming submodular maximization.}
In the random-order streaming setting,~\citet{ASS19} designed a nearly $(1-\nicefrac{1}{e})$-approximation semi-streaming algorithm for monotone submodular maximization subject to a cardinality constraint, which was later simplified by~\citet{LRVZ21} to reduce the memory size. Moreover, under the cardinality constraint,~\citet{RZ22b} designed a quadratic-memory streaming algorithm that achieves a $0.5029$ approximation and a $\nicefrac{1}{2}$ approximation for symmetric and asymmetric submodular maximization respectively. Under the matroid constraint,~\citet{Shadravan20} gave a $\frac{1-e^{-2}}{2}$-approximation semi-streaming algorithm for monotone submodular maximization. All these random-order streaming algorithms (and ours) also apply to the (similar but incomparable) secretary with shortlists model~\citep{ASS19}.

In the adversarial-order streaming setting,~\citet{BMKK14} designed a nearly $\nicefrac{1}{2}$-approximation semi-streaming algorithm for cardinality-constrained monotone submodular maximization, which was later improved by~\citet{KMZLK19}  to achieve a linear memory size. \citet{FNSZ23} proved a matching hardness result. More recently,~\citet{AEFNS22} designed a nearly $\nicefrac{1}{2}$-approximation semi-streaming algorithm for cardinality-constrained non-monotone submodular maximization, which is also an FPT algorithm (although not explicitly stated in their paper). Furthermore, matroid-constrained submodular maximization in the adversarial-order streaming setting was first studied by~\citet{CK15}, who provided a $\nicefrac{1}{4}$-approximation semi-streaming algorithm for monotone functions. This was recently improved to a $0.3178$ approximation by~\citet{FLNSZ25}, who also gave a $0.1921$ approximation for non-monotone functions. Finally, we mention that~\citet{FLNSZ25} developed a multi-pass nearly $(1-\nicefrac{1}{e})$-approximation semi-streaming algorithm for matroid-constrained monotone submodular maximization, which also builds on the continuous greedy algorithm of~\citet{BV14}. However, their approach is very different from ours.

\section{Preliminaries}\label{section:preliminaries}
\subsection{Submodular functions and matroids}
We first introduce submodular functions, matroids, and their associated concepts and properties. Given a ground set $[n]$ and a non-negative discrete function $f:2^{[n]}\to\R_{\ge0}$, for any $X,Y\subseteq [n]$, we let $f(X|Y)$ denote the marginal value of $X$ with respect to $Y$, i.e., $f(X|Y):=f(X\cup Y)-f(Y)$.
\begin{definition}[submodular function]\label{def:submodular}
A non-negative function $f:2^{[n]}\to\R_{\ge0}$ is submodular, if it satisfies that $f(X|Y)\ge f(X|Z)$ for any $X\subseteq[n]$ and any $Y\subseteq Z\subseteq[n]$.
\end{definition}
For any vector $x\in[0,1]^n$, we use the notation $\cR(x)$ to refer to a random set such that each element of $[n]$ appears in $\cR(x)$ independently with probability $x_i$, and we define the multi-linear extension of a submodular function as follows.
\begin{definition}[multi-linear extension]\label{def:multi-linear}
The multi-linear extension $F:[0,1]^n\to\R_{\ge0}$ of a submodular function $f:2^{[n]}\to\R_{\ge0}$ is given by
\[
F(x):=\E[f(\cR(x))]=\sum_{S\subseteq[n]}\prod_{i\in S} x_i\cdot\prod_{i\in[n]\setminus S}(1-x_i)\cdot f(S) \textrm{ for all } x\in[0,1]^n.
\]
Moreover, for any $x,y\in[0,1]^n$ such that $x+y\in[0,1]^n$, we let $F(x|y)$ denote the marginal value of $x$ with respect to $y$, i.e., $F(x|y):=F(x+y)-F(y)$.
\end{definition}

We are interested in the problem of maximizing a non-negative submodular function subject to a matroid constraint. A matroid is a set system that satisfies certain independence structure.
\begin{definition}[matroid]\label{def:matroid}
A set system $\M\subseteq 2^{[n]}$ is a matroid\footnote{In the literature, a matroid is often represented as a pair $([n], \M)$, where $[n]$ is the ground set and $\M\subseteq 2^{[n]}$ is the family of independent sets. For simplicity, we refer to the matroid simply as $\M$ throughout this paper.} if the following conditions hold:
\begin{enumerate}[i.]
    \item $\emptyset\in\M$.
    \item If $X\in\M$, then $Y\in\M$ for all $Y\subseteq X$.
    \item If $X,Y\in\M$ and $|Y|<|X|$, then there exists $i\in X\setminus Y$ such that $Y\cup\{i\}\in\M$.
\end{enumerate}
\end{definition}
Given a matroid, we define its rank, bases, restriction and contraction as follows.
\begin{definition}[rank]
Given a matroid $\M\subseteq 2^{[n]}$, the rank of $\M$ is $r_{\M}:=\max_{X\in\M}|X|$.
\end{definition}
\begin{definition}[base]
Given a matroid $\M\subseteq 2^{[n]}$, we say that a set $B\in\M$ is a base of $\M$, if $B\cup\{i\}\notin\M$ for all $i\in [n]$.
\end{definition}
\begin{definition}[restriction]
Given a matroid $\M\subseteq 2^{[n]}$ and a set $S\subseteq [n]$, we define the restriction of $\M$ to $S$ as $\M_S:=\{X\subseteq S\mid X\in \M\}$. It is well-known that the restriction $\M_S$ is a matroid.
\end{definition}
\begin{definition}[contraction]\label{def:contraction}
Given a matroid $\M\subseteq 2^{[n]}$ and a set $S\subseteq[n]$, we let $\M/S$ denote the contraction of $\M$ by $S$, which is given by $\M/S:=\{X\subseteq[n]\setminus S \mid X\cup S'\in\M \textrm{ for all } S'\subseteq S \textrm{ such that } S'\in\M\}$. It is well-known that the contraction $\M/S$ is a matroid, and its membership oracle can be efficiently evaluated using the membership oracle of $\M$.
\end{definition}

In matroid-constrained submodular maximization, an algorithm $\alg$ is given a non-negative submodular function $f:2^{[n]}\to\R_{\ge0}$ and a matroid $\M\subseteq2^{[n]}$, and its task is to find a feasible set $X\in\M$ that maximizes $f(X)$. Throughout the paper, we denote the optimal solution by $O$, i.e., $O:=\argmax_{X\in\M} f(X)$ (and we reserve $\cO$ for the Big-O notation to avoid confusion). We say that $\alg$ achieves an $\alpha$ approximation with $\alpha\in[0,1]$, if its solution set $X_{\alg}$ satisfies that $\E[f(X_{\alg})]\ge \alpha\cdot f(O)$, where the expectation is taken over the randomness of $\alg$ (and the arrival order of the elements, in case $\alg$ is a random-order streaming algorithm, which we will introduce shortly). Moreover, we assume that the matroid rank $r_{\M}$ is a fixed parameter that does not depend on the total number of elements $n$ (and hence, $r_{\M}=o(n)$\footnote{If $r_{\M}=\Omega(n)$, then an FPT algorithm with $2^{\cO(r_{\M})}$ runtime can simply enumerate all subsets of $[n]$, and a streaming algorithm with $\cO(r_{\M})$ memory can simply store all elements of $[n]$, rendering the problem trivial.}), and we are interested in \emph{fixed-parameter tractable} (FPT) algorithms, i.e., algorithms with runtime $h(r_{\M})\cdot \poly(n)$, where $h$ can be any finite function.

\subsection{Random-order streaming model}
Besides the standard offline setting (where all elements in the ground set $[n]$ are presented to the algorithm from the outset), we are also interested in the random-order streaming setting, where elements in $[n]$ arrive in a uniformly random order as a data stream.

In this setting, an algorithm has a memory buffer of a limited size (which is $\widetilde{O}(r_{\M})$ for semi-streaming algorithms). Each time when a new element arrives, the algorithm decides how to update its memory buffer, i.e., whether to store the new element, remove some previously stored elements, or modify other stored information. At any point, the algorithm can make any number of queries to the value oracle of the submodular function and the membership oracle of the matroid, provided that the input for the query is a subset of elements stored in its memory. At the end of the stream, the algorithm outputs a subset of elements from its memory that satisfies the matroid constraint.

\subsection{Other useful notions and lemmata}
We use the notations $\mathbf{1}_S\in\{0,1\}^n$ and $\mathbf{1}_{e}\in\{0,1\}^n$ to represent the indicator vectors for any set $S\subseteq[n]$ and for any element $e\in[n]$ respectively, and let $\mathbf{0}$ denote the all-zero vector. Moreover, we define a rounding operator that will be the key to reducing the memory usage of our algorithms.
\begin{definition}[rounding operator]\label{def:rounding_op}
Given any finite non-empty set of real numbers $I$, we use the notation $\floor{\cdot}_I$ to denote the rounding operator that maps any real number $a\in\R$ to its greatest lower bound in $I$, or to the smallest number in $I$ if no lower bound exists, i.e., $\floor{a}_I:=\max\{i\in I\mid i\le a\}$ if $a\ge\min I$, and $\floor{a}_I:=\min I$ otherwise.
\end{definition}
Then, we present two base exchange lemmas for matroids and two sampling lemmas for non-negative submodular functions.
\begin{lemma}[{\citet{Greene73,Woodall74}}]\label{lem:base_exchange}
Given two bases $B_1$ and $B_2$ of a matroid $\M\subseteq2^{[n]}$, and a partition $B_1=X_1\cupdot Y_1$, there exists a partition $B_2=X_2\cupdot Y_2$ such that $X_1\cupdot Y_2$ and $X_2\cupdot Y_1$ are both bases of $\M$.
\end{lemma}
\begin{lemma}[{\citet{Brualdi69}}]\label{lem:base_matching}
Given two bases $B_1$ and $B_2$ of a matroid $\M\subseteq2^{[n]}$, there exists a bijection $h:B_1\setminus B_2\to B_2 \setminus B_1$ such that for
all $e\in B_1\setminus B_2$, $(B_1\setminus\{e\})\cup\{h(e)\}$ is a base of $\M$.
\end{lemma}
\begin{lemma}[{\citet[Lemma 2.2]{FMV11}}]\label{lem:subsample_exactly}
Given a non-negative submodular function $f:2^{[n]} \to \mathbb{R}_{\ge 0}$ and a set $X\subseteq[n]$, if $R$ is a random subset of $X$ such that every element of $X$ appears in $R$ with probability \textbf{exactly} $p$ (not necessarily independently), then we have that $\E[f(R)] \ge (1-p)\cdot f(\emptyset) + p\cdot f(X)$.
\end{lemma}
\begin{lemma}[{\citet[Lemma 2.2]{BFNS14}}]\label{lem:subsample_at_most}
Given a non-negative submodular function $f:2^{[n]} \to \mathbb{R}_{\ge 0}$ and a set $X\subseteq[n]$, if $R$ is a random subset of $X$ such that every element of $X$ appears in $R$ with probability \textbf{at most} $p$ (not necessarily independently), then we have that $\E[f(R)] \ge (1-p)\cdot f(\emptyset)$.
\end{lemma}
Finally, we state a lemma that follows from standard rounding schemes~\citep{CCPV11,CVZ10} for matroid-constrained submodular maximization.
\begin{lemma}\label{lem:pipage}
Given a matroid $\M\in2^{[n]}$, a submodular function $f:2^{[n]}\to\R_{\ge0}$, and its multi-linear extension $F:[0,1]^n\to\R_{\ge0}$, for any $x=\sum_{i=1}^{N} \frac{1}{N}\cdot\mathbf{1}_{S_i}$ such that $S_i\in\M$ for all $i\in[N]$, there exists $X\subseteq\bigcup_{i=1}^{N} S_i$ such that $X\in\M$ and $f(X)\ge F(x)$.
\end{lemma}

\section{A \texorpdfstring{$(\nicefrac{1}{2}-\eps)$-approximation}{(1/2-eps)-approximation} random-order streaming algorithm}\label{section:streaming}
In this section, we present our semi-streaming algorithm, \textsc{Greedy-Filtering} (Algorithm~\ref{alg:greedy_filtering}), which achieves a $\nicefrac{1}{2}-\eps$ approximation using $\widetilde{\cO}\left(\frac{r}{\poly(\eps)}\right)$ memory for any submodular function $f:2^{[n]}\to\R_{\ge0}$ and any matroid with rank $r=o(n)$. Algorithm~\ref{alg:greedy_filtering} operates in three phases.

\begin{algorithm}[ht]
\SetAlgoLined
\SetKwInOut{Input}{Input}
\SetKwInOut{Output}{Output}
\Input{Matroid $\M\subseteq 2^{[n]}$ with rank $r\in\N^+$, submodular function $f:2^{[n]}\to\R_{\ge0}$, uniformly random permutation $\pi:[n]\to[n]$, and parameter $\eps\in\left(0,\frac{1}{2}\right)$.}
\SetAlgorithmName{Algorithm}~~
\SetKw{Break}{break}
 Throughout the stream $\pi$, the algorithm maintains a set $T$ consisting of the top $\ceil{\frac{\ln(1/\eps)}{\eps}}$ elements with the highest singleton values among the elements that have arrived (in case of ties, w.l.o.g., keep the element that arrived earliest)\;
 \tcp{Phase 1}
 Denote $V_0:=\{\pi(1),\dots,\pi(\ceil{\eps n})\}$\;
 $w\gets\max_{e\in V_0} f(\{e\})\cdot \frac{r}{\eps}$\;
 \tcp{Phase 2}
 Denote $V_i:=\{\pi(j)\mid j\in\{\ceil{\eps n}+(i-1)\cdot\ceil{\frac{\eps n}{r}}+1,\dots,\ceil{\eps n}+i\cdot\ceil{\frac{\eps n}{r}}\}\}$ for all $i\in[r]$\;
 $s_1\gets\argmax_{e\in V_1} f(\{e\})$\;
 $S_0\gets\emptyset$ and $S_1\gets\{s_1\}$\;
 \For{$i=2,\dots,r$}{
    $s_i\gets\argmax_{e\in V_i\cup\{s_1\} \textnormal{ s.t.}\,S_{i-1}\cup\{e\}\in\M} f(\{e\}|S_{i-1})$\label{algline:streaming_greedy_selection}\;
    $S_i\gets S_{i-1}\cup\{s_i\}$\;
 }
 \tcp{Phase 3}
 $I\gets\{(1+\eps)^i\cdot\frac{\eps^2 w}{r^2}\mid i\in\{0,...,\ceil{2\log_{1+\eps}(r/\eps)}\}\}$\;
 $H\gets\emptyset$\;
 \For{$j=\ceil{\eps n}+r\cdot\ceil{\frac{\eps n}{r}}+1,\dots,n$}{
    \If{$\exists\,i\in[r]$ s.t.~$S_{i-1}\cup\{\pi(j)\}\in\M$ and $\floor{f(\{\pi(j)\}|S_{i-1})}_I>\floor{f(\{s_i\}|S_{i-1})}_I$\label{algline:streaming_filter_condition}}{
        $H\gets H\cup\{\pi(j)\}$\;
    }
    \If{$|H|>\frac{r\ln(r/\eps)\cdot|I|}{\eps}$\label{algline:streaming_break_condition}}{
        \Break;
    }
 }
 $X_{\alg}\gets\argmax_{X\subseteq T\cup S_r\cup H \textnormal{ s.t.}\,X\in\M} f(X)$\;
 \Return $X_{\alg}$\;
 \caption{\textsc{Greedy-Filtering}$(f,\M,r,\pi,\eps)$}
 \label{alg:greedy_filtering}
\end{algorithm}

\begin{paragraph}{Phase 1: Estimating the highest singleton value.}
Algorithm~\ref{alg:greedy_filtering} first identifies the element with the highest singleton value in the initial $\eps$ fraction of the random stream $\pi$, and sets $w$ to be $\frac{r}{\eps}$ times the value of this element. It treats $w$ as an upper bound on the value of the most valuable element in $[n]$. If $w$ fails to upper bound this value, we can show that w.h.p., only a few of the most valuable elements contribute non-negligibly to the optimal solution. Hence, to address this, throughout the stream, Algorithm~\ref{alg:greedy_filtering} also maintains a set $T$ of the top $\ceil{\frac{\ln(1/\eps)}{\eps}}$ elements with highest singleton values among the elements that have appeared.
\end{paragraph}

\begin{paragraph}{Phase 2: Greedily constructing a solution set.}
Then, Algorithm~\ref{alg:greedy_filtering} divides the second $\eps$ fraction of the stream into $r$ windows of equal sizes, and greedily constructs a feasible solution set $S_r$ as follows: For each $i\in[r]$, among all elements in the $i$-th window $V_i$ that can be added to the current solution set $S_{i-1}$ without violating the matroid constraint, Algorithm~\ref{alg:greedy_filtering} identifies the element $s_i$ with the highest marginal value with respect to $S_{i-1}$, and updates $S_i=S_{i-1}\cup\{s_i\}$. If no element in $V_i$ can be added to $S_{i-1}$ without violating the matroid constraint (or all elements in $V_i$ that can be added have negative marginal values with respect to $S_{i-1}$), the algorithm sets $s_i$ to be $s_1$, which acts as a dummy element with zero marginal value with respect to $S_{i-1}$.
\end{paragraph}

\begin{paragraph}{Phase 3: Filtering the remaining stream.}
Algorithm~\ref{alg:greedy_filtering} filters elements in the rest of the stream, keeping only a subset $H$ of them: For each remaining element $e$, Algorithm~\ref{alg:greedy_filtering} adds element $e$ to $H$ if, for some $i\in[r]$, element $e$ has a strictly higher marginal value with respect to solution set $S_{i-1}$ than element $s_i$. Notably, when comparing the marginal values of $e$ and $s_i$, Algorithm~\ref{alg:greedy_filtering} first applies the rounding operator defined in Definition~\ref{def:rounding_op} with respect to a set $I$ of geometrically increasing numbers, upper bounded by $w$. As we will show in the analysis, this coarse comparison of marginal values, based on the rounding operator, is crucial for the algorithm to achieve $\nicefrac{1}{2}-\eps$ approximation while keeping the memory cost nearly linear in the matroid rank w.h.p. In the algorithm, we also set a hard memory limit at Line~\ref{algline:streaming_break_condition}.
\vspace{0.5\baselineskip}
\end{paragraph}

Finally, Algorithm~\ref{alg:greedy_filtering} finds the most valuable subset of $T\cup S_r\cup H$ that satisfies the matroid constraint through an exhaustive search and returns it as the final solution. The approximation guarantee of Algorithm~\ref{alg:greedy_filtering}, along with its memory cost, is formally stated in Theorem~\ref{thm:streaming}.

\begin{theorem}\label{thm:streaming}
Given any input submodular function $f:2^{[n]}\to\R_{\ge0}$, matroid $\M\subseteq2^{[n]}$ with rank $r\in\N^+$, and parameter $\eps\in\left(0,\frac{1}{2}\right)$, Algorithm~\ref{alg:greedy_filtering} achieves a $\nicefrac{1}{2}-\eps'$ approximation using $\cO\left(\frac{r\ln(r/\eps)^2}{\eps^2}\right)$ memory, where $\eps'=8\cdot\sqrt{2\eps+\frac{2r}{n}}$.
\end{theorem}

The memory bound follows from a simple calculation: Throughout the data stream, Algorithm~\ref{alg:greedy_filtering} only needs to maintain three sets of elements $T,S_r,H$ (note that it can infer $S_1,\dots,S_{r-1}$ from $S_r$). By the construction of $T$ and $S_r$, we have that $|T|=\ceil{\frac{\ln(1/\eps)}{\eps}}=\cO\left(\frac{1}{\eps^2}\right)$ and $|S_r|\le r$. Moreover, the break condition at Line~\ref{algline:streaming_break_condition} of Algorithm~\ref{alg:greedy_filtering} ensures that $|H|= \cO\left(\frac{r\ln(r/\eps)\cdot|I|}{\eps}\right)=\cO\left(\frac{r\ln(r/\eps)^2}{\eps^2}\right)$, since $|I|=\ceil{2\log_{1+\eps}(\frac{r}{\eps})}+1=\cO\left(\frac{\ln(r/\eps)}{\eps}\right)$. Hence, the total memory cost is $\cO\left(\frac{r\ln(r/\eps)^2}{\eps^2}\right)$.

Next, we prove the approximation guarantee in three main steps: First, we show that unless $\max_{e\in [n]}f(\{e\})\ge \frac{r}{\eps}\cdot\min_{e\in T}f(\{e\})$ (in which case we prove that the most valuable and feasible subset of $T$ is a $1-\eps$ approximation of the optimal solution), w.h.p., the value $w$ in Algorithm~\ref{alg:greedy_filtering} is an upper bound on the value of the most valuable element in $[n]$. Then, we show that not many elements pass the filter in the third phase of Algorithm~\ref{alg:greedy_filtering}, meaning that w.h.p., the break condition at Line~\ref{algline:streaming_break_condition} is never met. Finally, in the regular case where $w$ caps the value of the most valuable element and only a few elements pass the filter, we prove that w.h.p., the most valuable and feasible subset of $S_r\cup H$ is almost a $\nicefrac{1}{2}$ approximation of the optimal solution. We implement these three steps in the following three subsections.

\subsection{The simple case: \texorpdfstring{$\max_{e\in [n]}f(\{e\})\ge \frac{r}{\eps}\cdot\min_{e\in T}f(\{e\})$}{r square times lowest singleton value is below highest singleton value in T}}
We first analyze the simple case where $\max_{e\in [n]}f(\{e\})\ge \frac{r}{\eps}\cdot\min_{e\in T}f(\{e\})$. We note that the set $T$ in Algorithm~\ref{alg:greedy_filtering} eventually consists of the top $\ceil{\frac{\ln(1/\eps)}{\eps}}$ elements in $[n]$ with the highest singleton values (and we will only consider this final version of $T$ in the analysis). In Lemma~\ref{lem:streaming_easy_case}, we show that Algorithm~\ref{alg:greedy_filtering} achieves a $1-\eps$ approximation in this case.

\begin{lemma}\label{lem:streaming_easy_case}
If $\max_{e\in [n]}f(\{e\})\ge \frac{r}{\eps}\cdot\min_{e\in T}f(\{e\})$, Algorithm~\ref{alg:greedy_filtering} obtains a $1-\eps$ approximation.
\end{lemma}
\begin{proof}
Since $T$ contains the most valuable element in $[n]$, we have $\max_{e\in [n]}f(\{e\})=\max_{e\in T}f(\{e\})$. Hence, the assumption that $\max_{e\in [n]}f(\{e\})\ge \frac{r}{\eps}\cdot\min_{e\in T}f(\{e\})$ implies that
\begin{equation}\label{eq:max_higher_than_r_over_eps_times_min_in_T}
\max_{e\in T}f(\{e\})\ge \frac{r}{\eps}\cdot\min_{e\in T}f(\{e\}).
\end{equation}
Moreover, because $T$ consists of the top $\ceil{\frac{\ln(1/\eps)}{\eps}}$ most valuable elements in $[n]$, any element in $T$ has a value at least as large as any element outside $T$, i.e., $\min_{e\in T}f(\{e\}) \ge \max_{e\in[n]\setminus T}f(\{e\})$. Combining this with Ineq.~\eqref{eq:max_higher_than_r_over_eps_times_min_in_T}, we obtain that
\begin{equation}\label{eq:max_in_T_higher_than_r_over_eps_times_max_outside_T}
\max_{e\in T}f(\{e\})\ge \frac{r}{\eps}\cdot\max_{e\in[n]\setminus T}f(\{e\}).
\end{equation}

Now we upper bound the value of the optimal solution $O$ as follows,
\begin{align*}
    f(O)&\le f(O\cap T)+\sum_{e\in O\cap ([n]\setminus T)}f(\{e\}) &&\text{(By submodularity)}\\
    &\le f(O\cap T)+r\cdot\max_{e\in[n]\setminus T}f(\{e\}) &&\text{(Since $|O|\le r$)}\\
    &\le f(O\cap T)+\eps\cdot\max_{e\in T}f(\{e\}) &&\text{(By Ineq.~\eqref{eq:max_in_T_higher_than_r_over_eps_times_max_outside_T})}\\
    &\le (1+\eps)\cdot\max_{X\subseteq T \textnormal{ s.t.}\,X\in\M} f(X).
\end{align*}
Because each subset of $T$ is considered in the final exhaustive search of Algorithm~\ref{alg:greedy_filtering}, the returned solution $X_{\alg}$ is at least a $\frac{1}{1+\eps}$ approximation of the optimal solution $O$. Notice that $\frac{1}{1+\eps}\ge1-\eps$.
\end{proof}

Given Lemma~\ref{lem:streaming_easy_case}, we can now focus on the cases where $\max_{e\in [n]}f(\{e\})<\frac{r}{\eps}\cdot\min_{e\in T}f(\{e\})$. Under this assumption, we show that w.h.p., the value $w$ in Algorithm~\ref{alg:greedy_filtering} is an upper bound on $\max_{e\in [n]}f(\{e\})$. Specifically, we define the set $T^+:=\{e'\in[n]\mid f(\{e'\})\ge \min_{e\in T} f(\{e\})\}$ (note that if there are multiple elements with the same singleton value, $T$ could be a random set because of the random stream, but $T^+$ is always deterministic). In Lemma~\ref{lem:streaming_w_upper_bound_max_singleton_value}, we prove that w.h.p., at least one element of $T^+$ appears in the first $\eps$ fraction of the random stream $\pi$, which implies that $w\ge\max_{e\in [n]}f(\{e\})$, assuming $\max_{e\in [n]}f(\{e\})< \frac{r}{\eps}\cdot\min_{e\in T}f(\{e\})$.

\begin{lemma}\label{lem:streaming_w_upper_bound_max_singleton_value}
Let $E_1$ denote the event that there exists some element of $T^+$ that appears in $V_0=\{\pi(1),\dots,\pi(\ceil{\eps n})\}$, i.e., the first $\eps$ fraction of the random stream $\pi$. Then, we have that $\Pr[E_1]\ge 1-\eps$. Moreover, if $\max_{e\in [n]}f(\{e\})< \frac{r}{\eps}\cdot\min_{e\in T}f(\{e\})$, then conditioned on event $E_1$, it holds that $w\ge\max_{e\in [n]}f(\{e\})$.
\end{lemma}

\begin{proof}
Let $t=|T^+|$, and note that $t\ge|T|=\ceil{\frac{\ln(1/\eps)}{\eps}}$ since $T\subseteq T^+$. Let $e_1,\dots,e_t$ denote the elements of $T^+$. For each $i\in[t]$, we let $A_i$ be the event that $e_i\in V_0$, and let $\bar{A}_i$ be its complement. First, we upper bound the probability that none of elements in $T^+$ appear in $V_0$ as follows,
\begin{align*}
    \Pr\Bigl[\bigwedge_{i\in[t]}\bar{A}_i\Bigr]&=\prod_{i\in[t]}\Pr\Bigl[\bar{A}_i \mathrel{\Big|} \bigwedge_{j\in[i-1]}\bar{A}_j\Bigr]\le\prod_{i\in[t]}\Pr[\bar{A}_i]\\
    &\quad\text{(Since the probability that $e_i\in V_0$ increases when we condition on $e_1,\dots,e_{i-1}\notin V_0$)}\\
    &\le(1-\eps)^t\le(1-\eps)^{\frac{\ln(1/\eps)}{\eps}}\le\eps.
\end{align*}
It follows that $\Pr[E_1]=1-\Pr\bigl[\bigwedge_{i\in[t]}\bar{A}_i\bigr]\ge1-\eps$, which establishes the first part of the lemma. To prove the second part, we notice that conditioned on event $E_1$, the set $V_0$ contains some element in $T^+$, and therefore, we have that $w=\frac{r}{\eps}\cdot\max_{e\in V_0} f(\{e\})\ge \frac{r}{\eps}\cdot\min_{e\in T^+} f(\{e\})=\frac{r}{\eps}\cdot\min_{e\in T} f(\{e\})$. This implies that $w\ge\max_{e\in [n]}f(\{e\})$, assuming that $\max_{e\in [n]}f(\{e\})< \frac{r}{\eps}\cdot\min_{e\in T}f(\{e\})$.
\end{proof}

\subsection{The unlikely case: \texorpdfstring{$|H|>\frac{r\ln(r/\eps)\cdot|I|}{\eps}$}{too many elements passing the filter}}
Recall that the set $H$ in Algorithm~\ref{alg:greedy_filtering} consists of elements that pass the filter in the third phase (throughout the analysis, we will only consider the final version of $H$). In Lemma~\ref{lem:streaming_unlikely_case}, we show that w.h.p., $|H|\le\frac{r\ln(r/\eps)\cdot|I|}{\eps}$, which implies that the break condition at Line~\ref{algline:streaming_break_condition} of Algorithm~\ref{alg:greedy_filtering} is unlikely to be reached.

\begin{lemma}\label{lem:streaming_unlikely_case}
Let $E_2$ denote the event that $|H|\le \frac{r\ln(r/\eps)\cdot|I|}{\eps}$. Then, we have that $\Pr[E_2]\ge1-\eps$.
\end{lemma}

Now we introduce several key concepts that will play important roles in the proof of Lemma~\ref{lem:streaming_unlikely_case}. First, we notice that in Algorithm~\ref{alg:greedy_filtering}, each element $s_i$ chosen in the second phase acts as a \emph{selector} in the third phase: for each element $e$ in the third phase of Algorithm~\ref{alg:greedy_filtering}, the algorithm will include $e$ in set $H$ if it satisfies the selection condition posed by $s_i$, namely, if $S_{i-1}\cup\{e\}\in\M$ and $\floor{f(\{e\}|S_{i-1})}_I>\floor{f(\{s_i\}|S_{i-1})}_I$ (we say that an element $e\in[n]$ is \emph{selected} by the selector $s_i$ if it satisfies this condition). For each $i\in[r]$, we let $G_i$ denote the set of elements selected by $s_i$ among all elements appearing after the $i$-th window in the second phase of Algorithm~\ref{alg:greedy_filtering}, i.e.,
\begin{equation}\label{eq:G_i}
\textstyle G_i:=\{e\in[n]\setminus(\bigcup_{j=0}^i V_j)\mid S_{i-1}\cup\{e\}\in\M \textrm{ and } \floor{f(\{e\}|S_{i-1})}_I>\floor{f(\{s_i\}|S_{i-1})}_I\}
\end{equation}
(recall that $V_0$ is the set of elements in the first $\eps$ fraction of the stream, and for each $j\in[r]$, $V_j$ is the set of elements in the $j$-th window during the second phase of Algorithm~\ref{alg:greedy_filtering}).

We order the selectors according to their indices: $s_1,\dots,s_r$, and for each $i\in[r]$, we define
\begin{equation}\label{eq:H_i}\textstyle
H_i:=\bigcup_{j=1}^{i} G_j \textrm{ (and we let $H_0:=\emptyset$ for completeness)}.
\end{equation}
We say that a selector $s_i$ is \emph{ineffective} if there is an earlier selector $s_j$ for some $j<i$ such that $\floor{f(\{s_i\}|S_{i-1})}_I\ge\floor{f(\{s_j\}|S_{j-1})}_I$ (otherwise we call $s_i$ an \emph{effective} selector). In Claim~\ref{claim:ineffective_selector}, we show that, as the name suggests, any element that is selected by an ineffective selector $s_i$ would have already been selected by an earlier selector $s_j$ for some $j<i$, which implies that $H_i=H_{i-1}$.

\begin{claim}\label{claim:ineffective_selector}
Suppose that $s_i$ is an ineffective selector, i.e., there exists some $j<i$ such that $\floor{f(\{s_i\}|S_{i-1})}_I\ge\floor{f(\{s_j\}|S_{j-1})}_I$. Then, any element $e\in[n]$ that satisfies the selection condition of $s_i$, i.e., $S_{i-1}\cup\{e\}\in\M$ and $\floor{f(\{e\}|S_{i-1})}_I>\floor{f(\{s_i\}|S_{i-1})}_I$, also satisfies the selection condition of $s_j$, i.e., $S_{j-1}\cup\{e\}\in\M$ and $\floor{f(\{e\}|S_{j-1})}_I>\floor{f(\{s_j\}|S_{j-1})}_I$. In particular, this implies that $H_i=H_{i-1}$.
\end{claim}
\begin{proof}
We establish the first part of the claim by showing that $S_{i-1}\cup\{e\}\in\M$ implies $S_{j-1}\cup\{e\}\in\M$, and $\floor{f(\{e\}|S_{i-1})}_I>\floor{f(\{s_i\}|S_{i-1})}_I$ implies $\floor{f(\{e\}|S_{j-1})}_I>\floor{f(\{s_j\}|S_{j-1})}_I$.

First, since $\M$ is a matroid and $S_{j-1}\cup\{e\}\subseteq S_{i-1}\cup\{e\}$, the condition that $S_{i-1}\cup\{e\}\in\M$ implies the condition that $S_{j-1}\cup\{e\}\in\M$.

Moreover, because $f$ is a submodular function and $S_{j-1}\subseteq S_{i-1}$, it holds that $f(\{e\}|S_{j-1})\ge f(\{e\}|S_{i-1})$, which implies that $\floor{f(\{e\}|S_{j-1})}_I\ge\floor{f(\{e\}|S_{i-1})}_I$, by monotonicity of the rounding operator $\floor{\cdot}_I$ (Definition~\ref{def:rounding_op}). Hence, the condition that $\floor{f(\{e\}|S_{i-1})}_I>\floor{f(\{s_i\}|S_{i-1})}_I$ implies that $\floor{f(\{e\}|S_{j-1})}_I>\floor{f(\{s_i\}|S_{i-1})}_I$. Since we assume that $\floor{f(\{s_i\}|S_{i-1})}_I\ge\floor{f(\{s_j\}|S_{j-1})}_I$, it follows that $\floor{f(\{e\}|S_{j-1})}_I>\floor{f(\{s_j\}|S_{j-1})}_I$.

Next, we prove the second part of the claim: $H_i=H_{i-1}$. Recall that every element $e\in G_i$ appears after the $i$-th window and satisfies the selection condition of $s_i$, which implies that element $e$ also appears after the $j$-th window (because $j<i$) and satisfies the selection condition of $s_j$ (because of the first part of the claim). It follows that every element $e\in G_i$ belongs to $G_j$. Hence, we have that $G_i\subseteq G_j\subseteq H_{i-1}$, since $H_{i-1}=\bigcup_{\ell=1}^{i-1} G_{\ell}$ and $j\le i-1$. Finally, $H_i=H_{i-1}$ follows from $H_i=G_i\cup H_{i-1}$ and $G_i\subseteq H_{i-1}$.
\end{proof}

Now observe that the set $H$ in Algorithm~\ref{alg:greedy_filtering} is a subset of $H_r$. Hence, to upper bound $|H|$, it suffices to upper bound $|H_r|$. If we could prove that for each $i\in[r]$, $|H_i\setminus H_{i-1}|=\cO\left(\frac{r\ln(r/\eps)}{\eps}\right)$ holds w.h.p., conditioned on $s_i$ being an effective selector, then Lemma~\ref{lem:streaming_unlikely_case} would follow immediately. Indeed, Claim~\ref{claim:ineffective_selector} states that $H_i\setminus H_{i-1}=\emptyset$ for every ineffective selector $s_i$, and thus, we have that $H_r=\bigcup_{i\in[r] \textrm{ s.t.}\,s_i\textrm{ is an effective selector}} (H_i\setminus H_{i-1})$. Notice that there can be at most $|I|$ effective selectors because of the rounding operator $\floor{\cdot}_I$ (specifically, any two distinct effective selectors $s_i$ and $s_j$ must satisfy that $\floor{f(\{s_i\}|S_{i-1})}_I\neq\floor{f(\{s_j\}|S_{j-1})}_I$, and there are only $|I|$ distinct values in the range of $\floor{\cdot}_I$). Therefore, if $|H_i\setminus H_{i-1}|=\cO\left(\frac{r\ln(r/\eps)}{\eps}\right)$ holds for every effective selector $s_i$, then it would follow that $|H_r|=\cO\left(\frac{r\ln(r/\eps)\cdot|I|}{\eps}\right)$.

However, there is a caveat to the above argument: $|H_i\setminus H_{i-1}|=\cO\left(\frac{r\ln(r/\eps)}{\eps}\right)$ does not always hold w.h.p., conditioned on $s_i$ being an effective selector. It is possible to construct instances where $|H_i\setminus H_{i-1}|=\mathcal{\omega}\left(\frac{r\ln(r/\eps)}{\eps}\right)$ always holds, conditioned on $s_i$ being an effective selector. Fortunately, if this is the case, we can show that the probability that $s_i$ is an effective selector is negligible. We formalize this intuition in Lemma~\ref{lem:H_i_setminus_H_i-1}. The technical proof of Lemma~\ref{lem:H_i_setminus_H_i-1} is provided in Section~\ref{sec:proof_of_lem_H_i_setminus_H_i-1} of the appendix.

\begin{lemma}\label{lem:H_i_setminus_H_i-1}
For each $i\in[r]$, let $k_i$ denote the number of effective selectors among the first $i$ selectors $s_1,\dots,s_i$, and let $k_0=0$. Then, for all $i\in[r]$, we have that
\begin{equation}\label{eq:H_i_setminus_H_i-1}
\Pr\left[|H_i| > k_i\cdot \frac{r\ln(r/\eps)}{\eps} \,\middle\vert\, \forall\,j\in\{0,\dots,i-1\},\,|H_j| \le k_j\cdot \frac{r\ln(r/\eps)}{\eps}\right]\le\frac{\eps}{r}.
\end{equation}
\end{lemma}

We are ready to complete the proof of Lemma~\ref{lem:streaming_unlikely_case}.
\begin{proof}[Proof of Lemma~\ref{lem:streaming_unlikely_case}]
First, following the notations in Lemma~\ref{lem:H_i_setminus_H_i-1}, we upper bound the probability $\Pr\left[|H_r|\le k_r\cdot\frac{r\ln(r/\eps)}{\eps}\right]$ as follows,
\begin{align}\label{eq:Pr_H_r_le_k_r}
    &\Pr\left[|H_r|\le k_r\cdot\frac{r\ln(r/\eps)}{\eps}\right]\nonumber\\
    \ge& \Pr\left[\forall\,i\in[r],\,|H_i| \le k_i\cdot\frac{r\ln(r/\eps)}{\eps}\right]\nonumber\\
    =&\prod_{i\in[r]}\Pr\left[|H_i| \le k_i\cdot \frac{r\ln(r/\eps)}{\eps} \,\middle\vert\, \forall\,j\in\{0,\dots,i-1\},\,|H_j| \le k_j\cdot \frac{r\ln(r/\eps)}{\eps}\right]\nonumber\\
    \ge&\left(1-\frac{\eps}{r}\right)^r\qquad\qquad\qquad\qquad\qquad\qquad\qquad\qquad\qquad\qquad\qquad\qquad\qquad\qquad\text{(By Lemma~\ref{lem:H_i_setminus_H_i-1})}\nonumber\\
    \ge&\,1-\eps.
\end{align}
Moreover, we notice that the total number of effective selectors $k_r$ is at most $|I|$, because any two distinct effective selectors $s_i$ and $s_j$ must satisfy that $\floor{f(\{s_i\}|S_{i-1})}_I\neq\floor{f(\{s_j\}|S_{j-1})}_I$ (and there are only $|I|$ distinct values in the range of $\floor{\cdot}_I$). Thus, Ineq.~\eqref{eq:Pr_H_r_le_k_r} implies that $\Pr\left[|H_r|\le \frac{r\ln(r/\eps)\cdot|I|}{\eps}\right]\ge1-\eps$. The proof finishes by noticing that the set $H$ in Algorithm~\ref{alg:greedy_filtering} is a subset of $H_r$.
\end{proof}

\subsection{The regular case: \texorpdfstring{$\max_{e\in [n]}f(\{e\})< \frac{r}{\eps}\cdot\min_{e\in T}f(\{e\})$ and $|H|\le\frac{r\ln(r/\eps)\cdot|I|}{\eps}$}{when neither of the above two cases occurs}}
Now we consider the regular case where $\max_{e\in [n]}f(\{e\})< \frac{r}{\eps}\cdot\min_{e\in T}f(\{e\})$ and $|H|\le\frac{r\ln(r/\eps)\cdot|I|}{\eps}$. To analyze this case, we first establish two simple claims. Claim~\ref{claim:eps_fraction_does_not_hurt} shows that w.h.p., the value of the optimal solution $O$ does not decrease significantly, when combined with any subset of elements in a small fraction of the random stream $\pi$.
\begin{claim}\label{claim:eps_fraction_does_not_hurt}
For any $\eta,\delta\in[0,1]$, let $V$ be a random subset of $[n]$ such that every element of $[n]$ appears in $V$ with probability at most $\delta$ (not necessarily independently). Then, with probability at least $1-\frac{\delta}{\eta}$, it holds for all $X\subseteq V$ that $f(X|O)\ge-\eta\cdot f(O)$.
\end{claim}
\begin{proof}
We define $g:2^{[n]}\to\R_{\ge0}$ such that $g(Y)=f(Y\cup O)$ for all $Y\subseteq[n]$. Because $f$ is a non-negative submodular function, it follows that $g$ is also non-negative and submodular. Moreover, we define $X_{\min}:=\argmin_{X\subseteq V} f(X|O)$. Since $X_{\min}$ is a subset of $V$ and each element of $[n]$ appears in $V$ with probability at most $\delta$, it follows that each element of $[n]$ appears in $X_{\min}$ with probability at most $\delta$. Hence, by Lemma~\ref{lem:subsample_at_most}, we have that $\E[g(X_{\min})]\ge(1-\delta)\cdot g(\emptyset)$, which is equivalent to $\E[f(X_{\min}|O)+f(O)]\ge(1-\delta)\cdot f(O)$. By rearranging, we obtain that $\E[-f(X_{\min}|O)]\le\delta\cdot f(O)$. Then, by applying Markov's inequality to the random variable $-f(X_{\min}|O)$ (note that this random variable is non-negative because $f(X_{\min}|O)\le f(\emptyset|O)=0$), we have that
\[
    \Pr[-f(X_{\min}|O)\ge\eta\cdot f(O)]\le\frac{\E[-f(X_{\min}|O)]}{\eta\cdot f(O)}\le\frac{\delta\cdot f(O)}{\eta\cdot f(O)}=\frac{\delta}{\eta}.
\]
Thus, $f(X_{\min}|O)\ge-\eta\cdot f(O)$ holds with probability at least $1-\frac{\delta}{\eta}$, which implies the claim.
\end{proof}

Claim~\ref{claim:streaming_rounding_error} bounds the error incurred by the rounding operator $\floor{\cdot}_I$.
\begin{claim}\label{claim:streaming_rounding_error}
Given $I=\{(1+\eps)^i\cdot\gamma w\mid i\in\{0,...,\ceil{\log_{1+\eps}(1/\gamma)}\}\}$ with $w\ge0$ and $\eps,\gamma\in(0,1)$, for any two real numbers $a\ge0$ and $b\le w$, if $\floor{a}_I\ge\floor{b}_I$, then we have that $a\ge(1-\eps)b-\gamma w$.
\end{claim}
\begin{proof}
We consider two cases ($a\ge\gamma w$ and $a<\gamma w$) and prove the statement by case analysis.

\paragraph{Case 1: $a\ge\gamma w$.} In this case, if $b<\gamma w$, then $b<a$, which implies that $a\ge(1-\eps)b$. If $b\ge\gamma w$, then we have that $\min I\le b\le w\le\max I$. Since the rounding operator $\floor{\cdot}_I$ maps $b$ to its greatest lower bound in $I$, it follows that $\floor{b}_I\le b\le(1+\eps)\cdot\floor{b}_I$. By our assumption that $\floor{a}_I\ge\floor{b}_I$, it follows that $b\le(1+\eps)\cdot\floor{a}_I\le(1+\eps)\cdot a$, which implies that $a\ge(1-\eps)b$.

\paragraph{Case 2: $a<\gamma w$.} In this case, by definition of $I$ and Definition~\ref{def:rounding_op}, we have that $\floor{a}_I=\min I$. Since we assume that $\floor{a}_I\ge\floor{b}_I$, it follows that $\floor{b}_I=\min I$, which implies that $b<(1+\eps)\gamma w$. Hence, it follows that $(1-\eps)b-\gamma w<0$, which implies the statement since $a\ge0$.
\end{proof}

Next, we denote by $V_{\final}$ the set of elements that appear in the third phase of Algorithm~\ref{alg:greedy_filtering}, i.e., $V_{\final}:=\{\pi(i) \mid i=\ceil{\eps n}+r\cdot\ceil{\frac{\eps n}{r}}+1,\dots,n\}$. Recall that $X_{\alg}$ denotes the solution set returned by Algorithm~\ref{alg:greedy_filtering}, and $O$ denotes the optimal solution. In Lemma~\ref{lem:streaming_regular_case}, we show that $X_{\alg}$ is almost a $\nicefrac{1}{2}$ approximation of $O\cap V_{\final}$, conditioned on three high-probability events in the regular case.
\begin{lemma}\label{lem:streaming_regular_case}
Let $E_1$ and $E_2$ be the events defined in Lemma~\ref{lem:streaming_w_upper_bound_max_singleton_value} and Lemma~\ref{lem:streaming_unlikely_case} respectively. Moreover, let $E_3$ denote the event that $f(X|O)\ge-\eta\cdot f(O)$ for all $X\subseteq\bigcup_{i=1}^r V_i$, where $\eta:=\sqrt{\frac{2\eps n+2r}{n}}$. Suppose that the input parameter $\eps$ for Algorithm~\ref{alg:greedy_filtering} is in the range $\left(0,\frac{1}{2}\right)$, and that $\max_{e\in [n]}f(\{e\})< \frac{r}{\eps}\cdot\min_{e\in T}f(\{e\})$. Then, conditioned on events $E_1,E_2,E_3$, Algorithm~\ref{alg:greedy_filtering} guarantees that
$\textstyle f(X_{\alg})\ge\left(\frac{1}{2}-\eps\right)\cdot (f(O\cap V_{\final})-5\eta\cdot f(O))$.
\end{lemma}

We postpone the proofs of Lemma~\ref{lem:streaming_regular_case} and Theorem~\ref{thm:streaming} to Section~\ref{sec:proof_of_lem_streaming_regular_case} and Section~\ref{sec:proof_of_thm_streaming} of the appendix respectively. Here, we provide a simplified proof sketch to illustrate the main idea.

\begin{proof}[Proof sketch]\renewcommand{\qedsymbol}{}
To highlight the main idea, we make the simplifying assumption that both the set $S_r$ constructed in the second phase of Algorithm~\ref{alg:greedy_filtering} and the optimal solution $O$ are bases of matroid $\M$. We further assume that all elements in the optimal solution $O$ arrive in the third phase of Algorithm~\ref{alg:greedy_filtering} (in reality, we expect a random $1-2\eps$ fraction of them to appear in the third phase, which, by Lemma~\ref{lem:subsample_exactly}, accounts for $1-2\eps$ fraction of the optimal value in expectation).

First, we partition the optimal solution into two sets: $O=O_L\cupdot O_H$, where $O_L:=O\setminus H$ and $O_H:=O\cap H$. That is, $O_H$ contains the optimal elements that are included in $H$ by Algorithm~\ref{alg:greedy_filtering}, and $O_L$ consists of those that are not included in $H$ (note that conditioned on the high-probability event $E_2$ defined in Lemma~\ref{lem:streaming_unlikely_case}, these elements were not added to $H$ because they did not satisfy the filter condition at Line~\ref{algline:streaming_filter_condition} of Algorithm~\ref{alg:greedy_filtering}, not because of the break condition at Line~\ref{algline:streaming_break_condition}). By Lemma~\ref{lem:base_exchange}, we can find a partition $S_r=S_L\cupdot S_H$ such that $S_L\cupdot O_H$ and $O_L\cupdot S_H$ are also bases. Moreover, because both $S_r=S_L\cupdot S_H$ and $O_L\cupdot S_H$ are bases, by Lemma~\ref{lem:base_matching}, there exists a bijection $h:S_L\to O_L$ such that for all $s_i\in S_L$, substituting element $s_i$ in $S_r$ with element $h(s_i)$ yields a new base. In particular, this implies that element $h(s_i)$ could be added to set $S_{i-1}$ without violating the matroid constraint. Despite this, element $h(s_i)\in O_L$ was not selected by the selector $s_i$ (since otherwise it would have been included in $H$). Thus, it must hold that $f(\{h(s_i)\}|S_{i-1})\le f(\{s_i\}|S_{i-1})$ (barring the rounding error incurred by the rounding operator $\floor{\cdot}_I$, which can be handled by Claim~\ref{claim:streaming_rounding_error}). Hence, we can derive that
\begin{align}\label{eq:S_L_vs_O_L_simplifed}
f(S_L)&=\sum_{i\in[r]\textrm{ s.t.}\,s_i\in S_L} f(\{s_i\}|S_{i-1}\cap S_L) &&\text{(By telescoping sum)}\nonumber\\
&\ge\sum_{i\in[r]\textrm{ s.t.}\,s_i\in S_L} f(\{s_i\}|S_{i-1}) &&\text{(By submodularity)}\nonumber\\
&\ge\sum_{i\in[r]\textrm{ s.t.}\,s_i\in S_L} f(\{h(s_i)\}|S_{i-1})\nonumber\\
&=\sum_{o\in O_L} f(\{o\}|S_{i-1}) &&\text{(Since $h:S_L\to O_L$ is a bijection)}\nonumber\\
&\ge f(O_L|S_{i-1})\ge f(O_L|S_r) &&\text{(By submodularity)}.
\end{align}

Then, we consider the following two solution sets: $S_r$ and $S_L\cupdot O_H$. Recall that both $S_r$ and $S_L\cupdot O_H$ are bases, and thus, they are feasible solutions. Moreover, both $S_r$ and $S_L\cupdot O_H$ are subsets of $S_r\cup H$, and hence, they are candidate solutions in the final exhaustive search of Algorithm~\ref{alg:greedy_filtering}. Now we lower bound the sum of the solution values of $S_r$ and $S_L\cupdot O_H$, as follows,
\begin{align*}
f(S_r)+f(S_L\cupdot O_H)&=f(S_r)+f(S_L)+f(O_H|S_L)\\
&\ge f(S_r)+f(O_L|S_r)+f(O_H|S_L) &&\text{(By Ineq.~\eqref{eq:S_L_vs_O_L_simplifed})}\\
&\ge f(S_r)+f(O_L|S_r)+f(O_H|O_L\cup S_r)&&\text{(By submodularity and $S_L\subseteq S_r$)}\\
&= f(O\cup S_r) &&\text{(Since $O_L\cupdot O_H=O$)}\\
&\ge (1-\eta)\cdot f(O),
\end{align*}
where the last inequality holds because of the event $E_3$ in the statement of Lemma~\ref{lem:streaming_regular_case} (which is a high-probability event by Claim~\ref{claim:eps_fraction_does_not_hurt}) and the fact that $S_r\subseteq\bigcup_{i=1}^r V_i$. It follows that one of $S_r$ and $S_L\cupdot O_H$ must be a nearly $\nicefrac{1}{2}$ approximation of the optimal solution $O$. This establishes Lemma~\ref{lem:streaming_regular_case}.

To complete the proof of Theorem~\ref{thm:streaming}, we can assume w.l.o.g.~that $\max_{e\in [n]}f(\{e\})< \frac{r}{\eps}\cdot\min_{e\in T}f(\{e\})$ (because otherwise Theorem~\ref{thm:streaming} follows immediately from Lemma~\ref{lem:streaming_easy_case}), and then combine Lemma~\ref{lem:streaming_regular_case} with the high-probability bounds we have established for events $E_1,E_2,E_3$.
\end{proof}

\section{Interlude: continuous greedy filtering}\label{section:interlude}
In this section, we present \textsc{Continuous-Greedy-Filtering} (Subroutine~\ref{sub:continuous_greedy_filtering}), a subroutine that will serve as a building block of our final $(1-\nicefrac{1}{e}-\eps)$-approximation offline algorithm. Subroutine~\ref{sub:continuous_greedy_filtering} combines the techniques used in our streaming algorithm (Algorithm~\ref{alg:greedy_filtering}) with a fast continuous greedy algorithm introduced by~\citet{BV14}. It operates in three phases.

\begin{figure}[!t]
\noindent\begin{minipage}{\textwidth}
\renewcommand\footnoterule{}
\begin{algorithm}[H]
\KwIn{Matroid $\M\subseteq 2^{[n]}$ with rank $r\in\Z_{\ge0}$, submodular function $f:2^{[n]}\to\R_{\ge0}$ (and its multi-linear extension\footnote{We note that throughout the subroutine, the multi-linear extension $F$ is evaluated only on fractional solutions with a support of size at most $\frac{r}{\eps}$, and hence can be computed exactly in $2^{\cO(r/\eps)}$ time.} $F:[0,1]^n\to\R_{\ge0}$), and parameter $\eps$ such that $\frac{1}{\eps}\in \N^+$.}
\SetAlgorithmName{Subroutine}~~
\SetKw{Break}{break}
 \If{$r=0$}{\tcp{In case Subroutine~\ref{sub:continuous_greedy_filtering} is invoked on an input matroid with rank zero}
    $S\gets\emptyset$ and $H\gets\emptyset$\;
    \Return $S$ and $H$\;
 }
 \tcp{Phase 1}
 $v\gets\max_{e\in [n]} f(\{e\})$\;
 \tcp{Phase 2}
 $x^0\gets\mathbf{0}$\;
 \For{$t=1,\dots,\frac{1}{\eps}$}{
     $x^t\gets x^{t-1}$ and $S_0^t\gets \emptyset$\;
     \For{$i=1,\dots,r$}{
        Sample a set\footnote{We define $V_i^t$ explicitly only for the purpose of the analysis. Instead of explicitly sampling and storing $V_i^t$, the algorithm can iterate through all elements in $[n]$ and skip each element with probability $1-\frac{\eps^3}{r}$.} $V_i^t$ by including each element of $[n]$ independently with probability $\frac{\eps^3}{r}$\;
        \If{$\exists\,e\in V_i^t$ s.t.~$S_{i-1}^t\cup\{e\}\in\M$ and $F(\eps\cdot\mathbf{1}_{e}|x^{t-1}+\eps\cdot\mathbf{1}_{S_{i-1}^t})\ge0$\label{algline:offline_greedy_selection_if_condition}}{
            $s_i^t\gets\argmax_{e\in V_i^t \textnormal{ s.t.}\,S_{i-1}^t\cup\{e\}\in\M} F(\eps\cdot\mathbf{1}_{e}|x^{t-1}+\eps\cdot\mathbf{1}_{S_{i-1}^t})$\;\label{algline:offline_greedy_selection_arg_max}
            $S_i^t\gets S_{i-1}^t\cup\{s_i^t\}$\;
        }\Else{
            $S_i^t\gets S_{i-1}^t$\;
        }
        \tcp{If $s_i^t$ exists, $\Delta_{i}^t=\mathbf{1}_{s_i^t}$; otherwise, $\Delta_{i}^t=\mathbf{0}$ acts as a dummy element}
        $\Delta_{i}^t\gets \mathbf{1}_{S_{i}^t}-\mathbf{1}_{S_{i-1}^t}$\;
     }
     $x^t\gets x^{t-1}+\eps\cdot\mathbf{1}_{S_r^t}$\;
 }
 \tcp{Phase 3}
 $I\gets\{(1+\eps)^i\cdot\frac{\eps^2 v}{r}\mid i\in\{0,...,\ceil{\log_{1+\eps}(r/\eps)}\}\}$\;
 $H\gets\emptyset$\;
 \For{$j=1,\dots,n$}{
    \If{$\exists\,t\in\left[\frac{1}{\eps}\right],\,i\in[r]$ s.t.~$S_{i-1}^t\cup\{j\}\in\M$ and $\smallfloor{F(\eps\cdot\mathbf{1}_j|x^{t-1}+\eps\cdot\mathbf{1}_{S_{i-1}^t})}_I>\smallfloor{F(\eps\cdot\Delta_{i}^t|x^{t-1}+\eps\cdot\mathbf{1}_{S_{i-1}^t})}_I$\label{algline:offline_filter_condition}}{
        $H\gets H\cup\{j\}$\;
    }
    \If{$|H|>\frac{r\ln(r/\eps^2)\cdot|I|}{\eps^4}$\label{algline:offline_break_condition}}{
        \Break;
    }
 }
 $S\gets\bigcup_{t=1}^{1/\eps} S_r^t$\;
 \Return $S$ and $H$\;
 \caption{\textsc{Continuous-Greedy-Filtering}$(f,\M,r,\eps)$}
 \label{sub:continuous_greedy_filtering} 
\end{algorithm}
\end{minipage}
\end{figure}

\begin{paragraph}{Phase 1: Computing the highest singleton value.}
Subroutine~\ref{sub:continuous_greedy_filtering} first identifies the element in $[n]$ with the highest singleton value and assigns this value to $v$. Since Subroutine~\ref{sub:continuous_greedy_filtering} is designed for the offline setting, $v$ can be computed exactly by iterating through all elements in $[n]$.
\end{paragraph}

\begin{paragraph}{Phase 2: Constructing a fractional solution using the continuous greedy algorithm.}
Then, Subroutine~\ref{sub:continuous_greedy_filtering} uses the continuous greedy algorithm to construct a fractional solution from $\frac{r}{\eps}$ subsampled sets: It runs for $\frac{1}{\eps}$ epochs. In each epoch $t\in\left[\frac{1}{\eps}\right]$, it builds a new integral solution $S_r^t$ from scratch in $r$ steps. In each step $i\in[r]$, it samples a small subset of elements $V_i^t$. Among the elements in $V_i^t$ that can be added to the current integral solution without violating the matroid constraint, Subroutine~\ref{sub:continuous_greedy_filtering} identifies the element with the highest ``marginal value'' (as specified at Line~\ref{algline:offline_greedy_selection_arg_max}) with respect to the current fractional solution. If such an element exists and its marginal value is non-negative, Subroutine~\ref{sub:continuous_greedy_filtering} includes it in the current integral solution and adds an $\eps$ fraction of it to the current fractional solution.
\end{paragraph}

\begin{paragraph}{Phase 3: Filtering all elements in $[n]$.}
Finally, Subroutine~\ref{sub:continuous_greedy_filtering} filters all elements in $[n]$, selecting a subset $H$ from them: For each element $e\in[n]$, it adds element $e$ to $H$ if, for some step $i\in[r]$ of some epoch $t\in\left[\frac{1}{\eps}\right]$ in the second phase, (i) element $e$ can be added to the integral solution at that step without violating the matroid constraint, and (ii) it has a strictly higher marginal value than element $s_i^t$ with respect to the fractional solution at that step (or has a strictly positive marginal value in case $s_i^t$ does not exist). Here, the marginal values are also compared after applying the rounding operator $\floor{\cdot}_I$ with some well-chosen discretization $I$, in order to keep the memory usage nearly linear in the matroid rank w.h.p. In the subroutine, we also impose a hard memory limit at 
Line~\ref{algline:offline_break_condition}. At the end of Subroutine~\ref{sub:continuous_greedy_filtering}, it returns the set $H$, along with the union $\bigcup_{t=1}^{1/\eps} S_r^t$ of all the integral solutions constructed during the second phase.
\vspace{0.5\baselineskip}
\end{paragraph}

Now we let $O\in\M$ denote \emph{any} optimal solution, and define $O_H:=O\cap H$ and $O_L:=O\setminus H$ (i.e., $O_H$ contains the optimal elements that are selected in the third phase of Subroutine~\ref{sub:continuous_greedy_filtering}, and $O_L$ consists of those that are not selected). In Theorem~\ref{thm:offline_subroutine}, we show that Subroutine~\ref{sub:continuous_greedy_filtering} satisfies a property that is crucial for proving the $1-\nicefrac{1}{e}-\eps$ approximation guarantee of our final algorithm.

\begin{theorem}\label{thm:offline_subroutine}
Given any input submodular function $f:2^{[n]}\to\R_{\ge0}$, matroid $\M\subseteq2^{[n]}$ with rank $r\in\Z_{\ge0}$, and parameter $\eps<\frac{1}{4}$ such that $\frac{1}{\eps}\in\N^+$, Subroutine~\ref{sub:continuous_greedy_filtering} guarantees that with probability at least $1-2\eps$, there exists a solution set $X_{\alg}\subseteq S$ such that $X_{\alg}\in\M$ and
\[
f(X_{\alg})\ge \left(1-\frac{1}{e}\right)\cdot (f(O_L)-8\eps\cdot f(O)),
\]
where the set $S$ is the first output of Subroutine~\ref{sub:continuous_greedy_filtering}.
\end{theorem}

To prove Theorem~\ref{thm:offline_subroutine}, we assume w.l.o.g.~that the matroid rank $r$ is non-zero, as the statement holds trivially if $r=0$ (we include this edge case in Theorem~\ref{thm:offline_subroutine} because our final algorithm might invoke Subroutine~\ref{sub:continuous_greedy_filtering} on a matroid with rank zero). We first establish Lemma~\ref{lem:offline_break_condition}, which shows that w.h.p., the break condition at Line~\ref{algline:offline_break_condition} of Subroutine~\ref{sub:continuous_greedy_filtering} is never reached. The proof of Lemma~\ref{lem:offline_break_condition} relies on techniques similar to those developed in the proof of Lemma~\ref{lem:streaming_unlikely_case}, which we defer to Section~\ref{sec:proof_of_lem_offline_break_condition} of the appendix.
\begin{lemma}\label{lem:offline_break_condition}
Let $E_1$ denote the event that $|H|\le\frac{r\ln(r/\eps^2)\cdot|I|}{\eps^4}$. Then, we have that $\Pr[E_1]\ge1-\eps$.
\end{lemma}

Then, in Lemma~\ref{lem:offline_regular_case}, we show that conditioned on event $E_1$ and another high-probability event $E_2$, the fractional solution $x^{1/\eps}$ constructed in Subroutine~\ref{sub:continuous_greedy_filtering} is a nearly $1-\frac{1}{e}$ approximation of $O_L$.
\begin{lemma}\label{lem:offline_regular_case}
Suppose that the input parameter $\eps$ to Subroutine~\ref{sub:continuous_greedy_filtering} is in the range $\left(0,\frac{1}{4}\right)$. Let $E_1$ be the event defined in Lemma~\ref{lem:offline_break_condition}, and let $E_2$ denote the event that $f(X|O)\ge-\eps\cdot f(O)$ for all $X\subseteq\bigcup_{t\in\left[\frac{1}{\eps}\right],\,i\in[r]} V_i^t$, where the sets $V_i^t$ are the subsampled sets in Subroutine~\ref{sub:continuous_greedy_filtering}. Then, conditioned on events $E_1$ and $E_2$, Subroutine~\ref{sub:continuous_greedy_filtering} guarantees that
$F(x^{1/\eps})\ge\left(1-\frac{1}{e}\right)\cdot (f(O_L)-8\eps\cdot f(O))$.
\end{lemma}

We postpone the proof of Lemma~\ref{lem:offline_regular_case} to Section~\ref{sec:proof_of_lem_offline_regular_case} and provide a simplified proof sketch.

\begin{proof}[Proof sketch]\renewcommand{\qedsymbol}{}
To illustrate the main idea, we make the simplifying assumption that the integral solution sets $S_r^t$ constructed in Subroutine~\ref{sub:continuous_greedy_filtering} for all $t\in\left[\frac{1}{\eps}\right]$, and the optimal solution $O$ are all bases of matroid $\M$. 

Recall that the optimal solution $O$ is split into $O_L$ and $O_H$, where $O_H$ contains the optimal elements that are included in $H$ by Subroutine~\ref{sub:continuous_greedy_filtering}, and $O_L$ consists of those that are not included in $H$ (note that conditioned on the high-probability event $E_1$ defined in Lemma~\ref{lem:offline_break_condition}, these elements were not added to $H$ because they did not satisfy the filter condition at Line~\ref{algline:offline_filter_condition} of Subroutine~\ref{sub:continuous_greedy_filtering}, not because of the break condition at Line~\ref{algline:offline_break_condition}).
Given the partition $O=O_L\cupdot O_H$, for each $t\in\left[\frac{1}{\eps}\right]$, we can find a partition $S_r^t=S_L^t\cupdot S_H^t$ such that $S_L^t\cupdot O_H$ and $O_L\cupdot S_H^t$ are also bases. Moreover, because both $S_r^t=S_L^t\cupdot S_H^t$ and $O_L\cupdot S_H^t$ are bases, by Lemma~\ref{lem:base_matching}, there exists a bijection $h_t:S_L^t\to O_L$ such that for all $s_i^t\in S_L^t$, replacing element $s_i^t$ in $S_r^t$ with element $h_t(s_i^t)$ results in a new base. In particular, this implies that element $h_t(s_i^t)$ could be added to set $S_{i-1}^t$ without violating the matroid constraint. Despite this, element $h_t(s_i^t)\in O_L$ did not satisfy the selection condition at Line~\ref{algline:offline_filter_condition} corresponding to $\Delta_i^t$  (since otherwise it would have been included in $H$). This implies that 
\begin{equation}\label{eq:did_not_pass_Delta_i_t}
F(\eps\cdot\mathbf{1}_{h_t(s_i^t)}|x^{t-1}+\eps\cdot\mathbf{1}_{S_{i-1}^t})\le F(\eps\cdot\Delta_{i}^t|x^{t-1}+\eps\cdot\mathbf{1}_{S_{i-1}^t})=F(\eps\cdot\mathbf{1}_{s_i^t}|x^{t-1}+\eps\cdot\mathbf{1}_{S_{i-1}^t})
\end{equation}
(ignoring the error incurred by the rounding operator $\floor{\cdot}_I$, which can be addressed by Claim~\ref{claim:streaming_rounding_error}). Now we derive that
\begin{align}\label{eq:one_epoch_of_continuous_greedy_illustration_1}
F(x^t)-F(x^{t-1})&=\sum_{i\in[r]} F(\eps\cdot\mathbf{1}_{s_i^t}|x^{t-1}+\eps\cdot\mathbf{1}_{S_{i-1}^t}) &&\text{(By telescoping sum)}\nonumber\\
&\ge\sum_{i\in[r] \textrm{ s.t.}\,s_i^t\in S_L^t} F(\eps\cdot\mathbf{1}_{s_i^t}|x^{t-1}+\eps\cdot\mathbf{1}_{S_{i-1}^t}) &&\text{(By the if condition at Line~\ref{algline:offline_greedy_selection_if_condition})}\nonumber\\
&\ge \sum_{i\in[r] \textrm{ s.t.}\,s_i^t\in S_L^t} F(\eps\cdot\mathbf{1}_{h_t(s_i^t)}|x^{t-1}+\eps\cdot\mathbf{1}_{S_{i-1}^t}) &&\text{(By Ineq.~\eqref{eq:did_not_pass_Delta_i_t})}\nonumber\\
&\ge \sum_{i\in[r] \textrm{ s.t.}\,s_i^t\in S_L^t} F(\eps\cdot\mathbf{1}_{h_t(s_i^t)}|x^t) &&\text{(By submodularity)}\nonumber\\
&= \sum_{o\in O_L} F(\eps\cdot\mathbf{1}_{o}|x^t) &&\text{(Since $h_t:S_L^t\to O_L$ is a bijection)}.
\end{align}
By Definition~\ref{def:multi-linear}, $F(\eps\cdot\mathbf{1}_{o}|x^t)$ represents the additional value gained by increasing the probability of element $o$ appearing in the random set $\cR(x^t)$ by $\eps$, and thus, $F(\eps\cdot\mathbf{1}_{o}|x^t)=\eps\cdot\E[f(\{o\}|\cR(x^t)\setminus\{o\})]$. Combining this with Ineq.~\eqref{eq:one_epoch_of_continuous_greedy_illustration_1}, we obtain that
\begin{align}\label{eq:one_epoch_of_continuous_greedy_illustration_2}
F(x^t)-F(x^{t-1})&\ge \sum_{o\in O_L} \eps\cdot\E[f(\{o\}|\cR(x^t)\setminus\{o\})]\nonumber\\
&\ge \sum_{o\in O_L} \eps\cdot\E[f(\{o\}|\cR(x^t))] &&\text{(By submodularity)}\nonumber\\
&\ge \eps\cdot\E[f(O_L|\cR(x^t))] &&\text{(By submodularity)}\nonumber\\
&= \eps\cdot(\E[f(O_L\cup \cR(x^t))]-F(x^t)) &&\text{(By Definition~\ref{def:multi-linear})}\nonumber\\
&\ge \eps\cdot(f(O_L)-\eps\cdot f(O)-F(x^t)),
\end{align}
where the last inequality follows because conditioned on event $E_2$ in the statement of Lemma~\ref{lem:offline_regular_case}, the random set $\cR(x^t)$ (which is always a subset of $\bigcup_{t\in\left[\frac{1}{\eps}\right],\,i\in[r]} V_i^t$ in Subroutine~\ref{sub:continuous_greedy_filtering}) satisfies that $f(\cR(x^t)|O)\ge-\eps\cdot f(O)$.

Then, by iteratively applying Ineq.~\eqref{eq:one_epoch_of_continuous_greedy_illustration_2} from $t=1$ to $t=\frac{1}{\eps}$ (as in the analysis of the continuous greedy algorithm, with the difference being that the approximation guarantee here is with respect to $f(O_L)-\eps\cdot f(O)$ rather than $f(O)$), we can show that $F(x^{1/\eps})\ge\left(1-\frac{1}{e}\right)\cdot (f(O_L)-c\cdot\eps\cdot f(O))$ for some constant $c$, which establishes Lemma~\ref{lem:offline_regular_case}.
\end{proof}

Finally, Theorem~\ref{thm:offline_subroutine} follows by combining Lemma~\ref{lem:offline_regular_case} with Lemma~\ref{lem:offline_break_condition} and Claim~\ref{claim:eps_fraction_does_not_hurt}.
\begin{proof}[Proof of Theorem~\ref{thm:offline_subroutine}]
We assume w.l.o.g.~that the matroid rank $r$ is non-zero.
First, we show that the event $E_2$ defined in Lemma~\ref{lem:offline_regular_case} occurs w.h.p. Notice that in Subroutine~\ref{sub:continuous_greedy_filtering}, for any $t\in\left[\frac{1}{\eps}\right]$ and $i\in[r]$, the probability that an element $e\in[n]$ appears in the subsampled set $V_i^t$ is $\frac{\eps^3}{r}$. By a union bound, the probability that element $e$ appears in $\bigcup_{t\in\left[\frac{1}{\eps}\right],\,i\in[r]} V_i^t$ is at most $\frac{r}{\eps}\cdot\frac{\eps^3}{r}=\eps^2$. Therefore, it follows from Claim~\ref{claim:eps_fraction_does_not_hurt} that event $E_2$ occurs with probability at least $1-\eps$. Then, combining this with Lemma~\ref{lem:offline_break_condition} through a union bound, the probability that both events $E_1$ and $E_2$ occur is at least $1-2\eps$. Thus, by Lemma~\ref{lem:offline_regular_case}, it holds with probability at least $1-2\eps$ that $F(x^{1/\eps})\ge\left(1-\frac{1}{e}\right)\cdot (f(O_L)-8\eps\cdot f(O))$. This implies Theorem~\ref{thm:offline_subroutine} by Lemma~\ref{lem:pipage}, because $x^{1/\eps}=\sum_{t=1}^{1/\eps}\eps\cdot\mathbf{1}_{S_r^t}$ and $S=\bigcup_{t=1}^{1/\eps}S_r^t$.
\end{proof}

On a separate note, Subroutine~\ref{sub:continuous_greedy_filtering} itself does not guarantee a better-than-$\nicefrac{1}{2}$ approximation. This is because Subroutine~\ref{sub:continuous_greedy_filtering} can be interpreted as an $\widetilde{\cO}\left(\frac{r}{\poly(\eps)}\right)$-memory random-order streaming algorithm (by estimating the highest singleton value as in Algorithm~\ref{alg:greedy_filtering} and replacing the subsampled sets with small segments of the random stream), and the hardness result by~\citet{RZ22b} asserts that such random-order streaming algorithms cannot exceed $\nicefrac{1}{2}$ approximation.

\section{A \texorpdfstring{$(1-\nicefrac{1}{e}-\eps)$-approximation}{(1-1/e-eps)-approximation} offline algorithm}\label{section:offline}
In this section, we present our final algorithm (Algorithm~\ref{alg:recursive_continuous_greedy_filtering}) that achieves a nearly $1-\nicefrac{1}{e}$ approximation by recursively applying \textsc{Continuous-Greedy-Filtering} (Subroutine~\ref{sub:continuous_greedy_filtering}) to progressively refined instances. Specifically, Algorithm~\ref{alg:recursive_continuous_greedy_filtering} performs a recursive process consisting of $\frac{1}{\alpha}$ levels.

\begin{algorithm}[ht]
\KwIn{Matroid $\M\subseteq 2^{[n]}$ with rank $r\in\N^+$, submodular function $f:2^{[n]}\to\R_{\ge0}$, sets $S^{\acc},O_H^{\acc}\subseteq[n]$, parameter $\alpha$ such that $\frac{1}{\alpha}\in \N^+$, and depth counter $k\in\left[\frac{1}{\alpha}\right]$.}
\SetAlgorithmName{Algorithm}~~
Define $f^{(k)}:2^{[n]\setminus O_H^{\acc}}\to\R_{\ge0}$ such that $f^{(k)}(X):=f(X\cup O_H^{\acc})$ for all $X\subseteq [n]\setminus O_H^{\acc}$\;
Define $\M^{(k)}\subseteq2^{[n]\setminus O_H^{\acc}}$ as $\M^{(k)}:=\M/O_H^{\acc}$\;
$S^{(k)},H^{(k)}\gets\textsc{Continuous-Greedy-Filtering}(f^{(k)},\M^{(k)},r_{\M^{(k)}},\alpha^2)$; \tcp{Subroutine~\ref{sub:continuous_greedy_filtering}}
$S^{\acc}\gets S^{\acc}\cup S^{(k)}$\;
\If{$k=\frac{1}{\alpha}$}{
    $X_{\alg}\gets\argmax_{X\subseteq S^{\acc}\cup O_H^{\acc}\textrm{ s.t.}\,X\in\M}f(X)$\label{algline:offline_exhaustive_search}\;
    \Return $X_{\alg}$\;
}
$X_{\alg}\gets\emptyset$\;
\For{$O_H^{(k)}\subseteq H^{(k)}$\label{algline:recursion_for_loop}}{
    $O_H^{\acc'}\gets O_H^{\acc}\cup O_H^{(k)}$\;
    $X_{\alg}'\gets\textsc{Recursive-Continuous-Greedy-Filtering}(f,\M,r,S^{\acc},O_H^{\acc'},\alpha,k+1)$\;
    \If{$f(X_{\alg}')>f(X_{\alg})$}{
        $X_{\alg}\gets X_{\alg}'$\;
    }        
}
\Return $X_{\alg}$\;

 \caption{\textsc{Recursive-Continuous-Greedy-Filtering}$(f,\M,r,S^{\acc},O_H^{\acc},\alpha,k)$}
 \label{alg:recursive_continuous_greedy_filtering}
\end{algorithm}

\begin{paragraph}{Initial phase.}
We initiate the first call to Algorithm~\ref{alg:recursive_continuous_greedy_filtering} by setting $k=1$, $S^{\acc}=\emptyset$, and $O_H^{\acc}=\emptyset$. We denote the optimal solution to the input instance $(f,\M)$ by $O^{(1)}$.  At the root level $k=1$, Algorithm~\ref{alg:recursive_continuous_greedy_filtering} first invokes Subroutine~\ref{sub:continuous_greedy_filtering} with parameter $\eps=\alpha^2$ on the input instance $(f,\M)$, which returns two sets $S^{(1)}$ and $H^{(1)}$. Then, Algorithm~\ref{alg:recursive_continuous_greedy_filtering} enumerate all subsets of $H^{(1)}$ in the for loop at Line~\ref{algline:recursion_for_loop}. Crucially, there will be an iteration of this loop where $O_H^{(1)}=O^{(1)}\cap H^{(1)}$ (all other iterations will be irrelevant to the analysis). Within this particular iteration (as well as in all other iterations), Algorithm~\ref{alg:recursive_continuous_greedy_filtering} makes a recursive call to itself by setting $S^{\acc}=S^{(1)}$ and $O_H^{\acc}=O_H^{(1)}$. These sets will be accumulated throughout the recursion.
\end{paragraph}

\begin{paragraph}{Recursion phase.}
At each internal recursion level $k\in\{2,\dots,\frac{1}{\alpha}-1\}$, Algorithm~\ref{alg:recursive_continuous_greedy_filtering} receives the accumulated sets $S^{\acc}=\bigcup_{j=1}^{k-1} S^{(j)}$ and $O_H^{\acc}=\bigcup_{j=1}^{k-1} O_H^{(j)}$ from the previous $k-1$ levels. It first constructs a refined instance $(f^{(k)},\M^{(k)})$, where the matroid $\M^{(k)}$ is the contraction of the original matroid $\M$ by the accumulated set $\bigcup_{j=1}^{k-1} O_H^{(j)}$, and the non-negative submodular function $f^{(k)}$ is constructed by evaluating the original function $f$ on the union of the accumulated set $\bigcup_{j=1}^{k-1} O_H^{(j)}$ and the input set to $f^{(k)}$. We define $O^{(k)}:=O^{(1)}\setminus\left(\bigcup_{j=1}^{k-1}O_H^{(j)}\right)$. Next, Algorithm~\ref{alg:recursive_continuous_greedy_filtering} calls Subroutine~\ref{sub:continuous_greedy_filtering} with parameter $\eps=\alpha^2$ on the instance $(f^{(k)},\M^{(k)})$, which returns two sets $S^{(k)}$ and $H^{(k)}$. Then, Algorithm~\ref{alg:recursive_continuous_greedy_filtering} enumerates all subsets of $H^{(k)}$. Crucially, there will be an iteration of this loop where $O_H^{(k)}=O^{(k)}\cap H^{(k)}$ (all other iterations will be irrelevant to the analysis). In this particular iteration (as well as in all other iterations), Algorithm~\ref{alg:recursive_continuous_greedy_filtering} makes a recursive call to itself, passing the updated accumulated sets $S^{\acc}=\bigcup_{j=1}^k S^{(j)}$ and $O_H^{\acc}=\bigcup_{j=1}^k O_H^{(j)}$ as arguments.

It is important to keep in mind that, in the recursion tree of Algorithm~\ref{alg:recursive_continuous_greedy_filtering} (illustrated in Figure~\ref{fig:illustration_recursive}), there is a unique root-to-leaf path such that for each $k\in\left[\frac{1}{\alpha}-1\right]$, the $(k+1)$-th call to Algorithm~\ref{alg:recursive_continuous_greedy_filtering} along the path (corresponding to the $(k+1)$-th node in the path) is invoked by the iterate $O_H^{(k)}=O^{(k)}\cap H^{(k)}$ (from the for loop at Line~\ref{algline:recursion_for_loop}) during the $k$-th call, where $O^{(k)}=O^{(1)}\setminus\left(\bigcup_{j=1}^{k-1}O_H^{(j)}\right)$. We will focus on this particular path when we prove the approximation guarantee of Algorithm~\ref{alg:recursive_continuous_greedy_filtering}.
\end{paragraph}

\begin{figure}[ht]
    \centering
    \vspace*{0.1in}
    \includegraphics[scale=0.25]{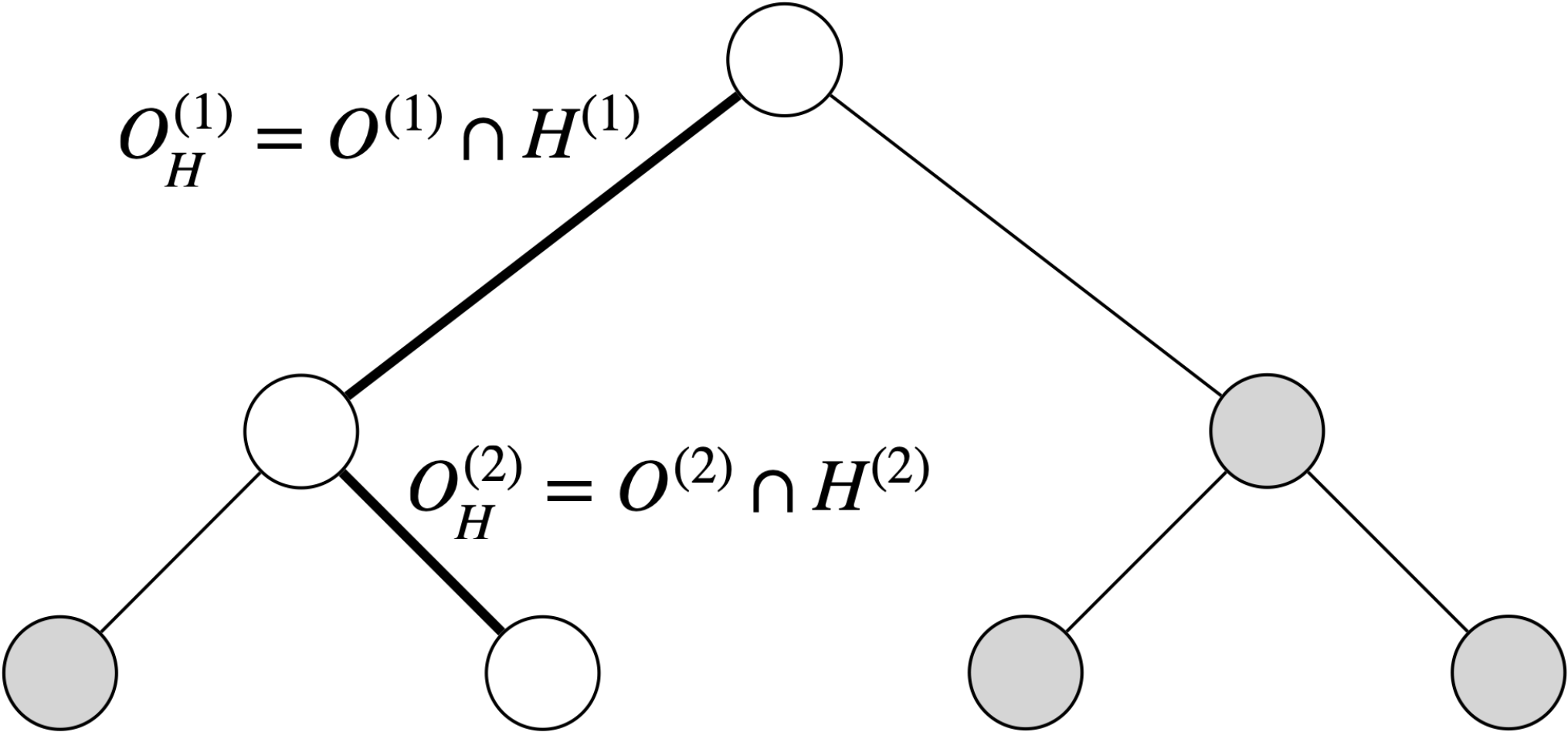}
    \vspace*{0.15in}
    \caption{This figure shows an example depth-$2$ recursion tree of Algorithm~\ref{alg:recursive_continuous_greedy_filtering}. The empty nodes form a unique path (made bold in the figure) in which, for each $k\in[2]$, the $(k+1)$-th call is invoked by the iterate $O_H^{(k)}=O^{(k)}\cap H^{(k)}$ during the $k$-th call, where $O^{(k)}=O^{(1)}\setminus\left(\bigcup_{j=1}^{k-1}O_H^{(j)}\right)$.}
    \label{fig:illustration_recursive}
\end{figure}

\begin{paragraph}{Return phase.}
At the leaf level $k=\frac{1}{\alpha}$, Algorithm~\ref{alg:recursive_continuous_greedy_filtering} first calls Subroutine~\ref{sub:continuous_greedy_filtering} on a refined instance $(f^{(1/\alpha)},\M^{(1/\alpha)})$ (constructed similarly to the instances in the previous levels), which returns two sets $S^{(1/\alpha)}$ and $H^{(1/\alpha)}$. Then, it conducts an exhaustive search to find the most valuable subset $X_{\alg}$ of $S^{\acc}\cup O_H^{\acc}=\left(\bigcup_{j=1}^{1/\alpha} S^{(j)}\right)\cup\left(\bigcup_{j=1}^{1/\alpha-1} O_H^{(j)}\right)$ that belongs to matroid $\M$, and returns $X_{\alg}$ to the previous level. Then, at each level $k\in\left[\frac{1}{\alpha}-1\right]$, among all solution sets returned by the recursive calls, Algorithm~\ref{alg:recursive_continuous_greedy_filtering} identifies and returns the most valuable one to the previous level.
\vspace{0.5\baselineskip}
\end{paragraph}

\begin{paragraph}{Time and space complexity.} Now we analyze the time and space complexity of Algorithm~\ref{alg:recursive_continuous_greedy_filtering} (assuming that we set $k=1$, $S^{\acc}=\emptyset$, and $O_H^{\acc}=\emptyset$ in the initial call).
We first observe that at each level $k\in\left[\frac{1}{\alpha}\right]$, the sets $S^{(k)}$ and $H^{(k)}$, which are outputs of Subroutine~\ref{sub:continuous_greedy_filtering} with parameter $\alpha^2$, have sizes $|S^{(k)}|\le\frac{r}{\alpha^2}$ (because $S^{(k)}$ is either the union of $\frac{1}{\alpha^2}$ solution sets constructed in the second phase of Subroutine~\ref{sub:continuous_greedy_filtering} or an empty set) and $|H^{(k)}|=\cO\left(\frac{r\ln(r)^2}{\poly(\alpha)}\right)$ (because of the break condition at Line~\ref{algline:offline_break_condition} of Subroutine~\ref{sub:continuous_greedy_filtering}). From this, we observe the following:
\begin{enumerate}
\item[\textit{i.}] At each leaf node of the recursion tree of Algorithm~\ref{alg:recursive_continuous_greedy_filtering}, the accumulated set $S^{\acc}=\bigcup_{k=1}^{1/\alpha} S^{(k)}$ has size $|S^{\acc}|\le\frac{r}{\alpha^3}$ (because each set $S^{(k)}$ has size at most $\frac{r}{\alpha^2}$), and the accumulated set $O^{\acc}_H=\bigcup_{k=1}^{1/\alpha} O_H^{(k)}$ has size $|O^{\acc}_H|=\cO\left(\frac{r\ln(r)^2}{\poly(\alpha)}\right)$ (because $|O_H^{(k)}|\le|H^{(k)}|=\cO\left(\frac{r\ln(r)^2}{\poly(\alpha)}\right)$).
\item[\textit{ii.}] Each non-leaf node of the recursion tree has at most $2^{\cO\big(\frac{r\ln(r)^2}{\poly(\alpha)}\big)}$ child nodes (because each set $H^{(k)}$ has size $\cO\left(\frac{r\ln(r)^2}{\poly(\alpha)}\right)$), which implies that the total number of leaves is $2^{\cO\big(\frac{r\ln(r)^2}{\poly(\alpha)}\big)}$, since the recursion tree has $\frac{1}{\alpha}$ levels in total.
\end{enumerate}
Then, we notice that the runtime of each call to Subroutine~\ref{sub:continuous_greedy_filtering} at each node is $n\cdot 2^{\cO\left(r/\alpha^2\right)}$, where the $2^{\cO\left(r/\alpha^2\right)}$ factor accounts for the exact computation of the multi-linear extension (which, if desired, can be replaced by polynomial-time estimation using standard Monte-Carlo sampling). Moreover, the exhaustive search step at each leaf node takes $2^{\cO\big(\frac{r\ln(r)^2}{\poly(\alpha)}\big)}$ time, since $|S^{\acc}|+|O^{\acc}_H|=\cO\left(\frac{r\ln(r)^2}{\poly(\alpha)}\right)$. Hence, summing the runtime over all nodes, we obtain that the total runtime is $n\cdot 2^{\cO\big(\frac{r\ln(r)^2}{\poly(\alpha)}\big)}$. For the space complexity, we assume that Algorithm~\ref{alg:recursive_continuous_greedy_filtering} traverses the recursion tree in depth-first order. We notice that at any node in any level $\ell\in\left[\frac{1}{\alpha}\right]$ of the recursion tree, Algorithm~\ref{alg:recursive_continuous_greedy_filtering} only needs to store the sets $S^{(k)},H^{(k)},O_H^{(k)}$ for all $k\in[\ell]$ from its ancestor nodes and itself, and among $S^{(k)},H^{(k)},O_H^{(k)}$, the set $H^{(k)}$ has the largest size, which is $\cO\left(\frac{r\ln(r)^2}{\poly(\alpha)}\right)$. Therefore, the space complexity is $\cO\left(\frac{r\ln(r)^2}{\poly(\alpha)}\right)$.
\vspace{0.5\baselineskip}
\end{paragraph}

Next, in Theorem~\ref{thm:offline}, we show that Algorithm~\ref{alg:recursive_continuous_greedy_filtering} achieves a nearly $1-\nicefrac{1}{e}$ approximation.
At a high level, the argument is as follows: Consider the specific path in the recursion tree of Algorithm~\ref{alg:recursive_continuous_greedy_filtering}, as illustrated in Figure~\ref{fig:illustration_recursive}, but with depth $\frac{1}{\alpha}$ for arbitrarily small $\alpha$. Suppose that for all $k\in\left[\frac{1}{\alpha}\right]$, the marginal value of the set $O^{(k)}_H$ with respect to $\bigcup_{j=1}^{k-1} O^{(j)}_H$ is at least $\alpha$ fraction of the optimal value, i.e., at least $\alpha\cdot f(O^{(1)})$. Then, the union $\bigcup_{j=1}^{1/\alpha} O^{(j)}_H$, which is a candidate solution in the exhaustive search step at the leaf node, must capture the entire optimal value. On the other hand, if for any $k\in\left[\frac{1}{\alpha}\right]$, the marginal value of the set $O^{(k)}_H$ with respect to $\bigcup_{j=1}^{k-1} O^{(j)}_H$ is less than $\alpha$ fraction of the optimal value, then by submodularity, the set $O^{(1)}\setminus O^{(k)}_H$ should account for the majority of the optimal value. We can show that (using Theorem~\ref{thm:offline}) w.h.p., the union of the set $\bigcup_{j=1}^{k-1} O^{(j)}_H$ and a subset of the first output of Subroutine~\ref{sub:continuous_greedy_filtering} during the $k$-th call along the path, which is a candidate solution in the exhaustive search step, is a nearly $1-\nicefrac{1}{e}$ approximation to $O^{(1)}\setminus O^{(k)}_H$.
\begin{theorem}\label{thm:offline}
Given any input submodular function $f:2^{[n]}\to\R_{\ge0}$, matroid $\M\subseteq2^{[n]}$ with rank $r\in\N^+$, and parameter $\alpha\in\left(0,\frac{1}{2}\right)$ such that $\frac{1}{\alpha}\in\N^+$, by setting $S^{\acc}=\emptyset$, $O_H^{\acc}=\emptyset$, and $k=1$, Algorithm~\ref{alg:recursive_continuous_greedy_filtering} obtains a $1-\frac{1}{e}-7\alpha$ approximation with $n\cdot 2^{\cO\big(\frac{r\ln(r)^2}{\poly(\alpha)}\big)}$ runtime and $\cO\left(\frac{r\ln(r)^2}{\poly(\alpha)}\right)$ memory.
\end{theorem}
\begin{proof}
We let $O^{(1)}$ denote the optimal solution to the input instance $(f,\M)$. Throughout the proof, we focus on the unique root-to-leaf path in the recursion tree of Algorithm~\ref{alg:recursive_continuous_greedy_filtering} (as illustrated in Figure~\ref{fig:illustration_recursive}) such that for each $k\in\left[\frac{1}{\alpha}-1\right]$, the $(k+1)$-th call along the path is invoked by the iterate $O_H^{(k)}=O^{(k)}\cap H^{(k)}$ during the $k$-th call, where $O^{(k)}=O^{(1)}\setminus\left(\bigcup_{j=1}^{k-1}O_H^{(j)}\right)$. For simplicity, we will refer to the $k$-th call along this path simply as the ``$k$-th call''. For completeness, we also let $O^{(1/\alpha)}=O^{(1)}\setminus\left(\bigcup_{j=1}^{1/\alpha-1}O_H^{(j)}\right)$ and $O_H^{(1/\alpha)}=O^{(1/\alpha)}\cap H^{(1/\alpha)}$.

\subsubsection*{Step 1: For all $k\in\left[\frac{1}{\alpha}\right]$, $O^{(k)}$ is an optimal solution to $(f^{(k)},\M^{(k)})$}
First, we show that $O^{(k)}$ is an optimal solution to the instance $(f^{(k)},\M^{(k)})$ constructed in the $k$-th call. Specifically, recall that $\M^{(k)}=\M/O_H^{\acc}=\M/\left(\bigcup_{j=1}^{k-1}O_H^{(j)}\right)$. Because $O^{(1)}\in\M$ and $\bigcup_{j=1}^{k-1}O_H^{(j)}\subseteq O^{(1)}$, it follows that $O^{(k)}=O^{(1)}\setminus\left(\bigcup_{j=1}^{k-1}O_H^{(j)}\right)$ belongs to matroid $\M^{(k)}$, and thus, $O^{(k)}$ is a feasible solution to the instance $(f^{(k)},\M^{(k)})$. Moreover, since $f^{(k)}(X)=f\left(X\cup O_H^{\acc}\right)=f\left(X\cup \left(\bigcup_{j=1}^{k-1}O_H^{(j)}\right)\right)$ for all $X\subseteq[n]\setminus\left(\bigcup_{j=1}^{k-1}O_H^{(j)}\right)$, and $O^{(1)}$ is the optimal solution to the original input instance $(f,\M)$, it follows that $O^{(k)}=O^{(1)}\setminus\left(\bigcup_{j=1}^{k-1}O_H^{(j)}\right)$ is an optimal solution to the instance $(f^{(k)},\M^{(k)})$.

\subsubsection*{Step 2: For all $k\in\left[\frac{1}{\alpha}\right]$, $X_{\alg}^{(k)}$ is a nearly $1-\nicefrac{1}{e}$ approximation for $(f^{(k)},\M^{(k)})$, w.h.p.}
Then, for each $k\in\left[\frac{1}{\alpha}\right]$, we define the set $O_L^{(k)}:=O^{(k)}\setminus O_H^{(k)}$ (and thus, $O_L^{(k)}=O^{(k)}\setminus H^{(k)}$ since $O_H^{(k)}=O^{(k)}\cap H^{(k)}$), and we let $E_k$ be the following event:
\begin{itemize}
    \item Event $E_k$: There exists $X_{\alg}^{(k)}\subseteq S^{(k)}$ such that $X_{\alg}^{(k)}\in\M^{(k)}$ and
    $$f^{(k)}(X_{\alg}^{(k)})\ge \left(1-\frac{1}{e}\right)\cdot (f^{(k)}(O_L^{(k)})-8\alpha^2\cdot f^{(k)}(O^{k})),$$
    where $f^{(k)}$ is the input function to Subroutine~\ref{sub:continuous_greedy_filtering} and $S^{(k)}$ is the first output of Subroutine~\ref{sub:continuous_greedy_filtering} during the $k$-th call.
\end{itemize}
Since Subroutine~\ref{sub:continuous_greedy_filtering} is invoked with parameter $\alpha^2$ in Algorithm~\ref{alg:recursive_continuous_greedy_filtering}, and given that $\alpha\in\left(0,\frac{1}{2}\right)$, it follows by Theorem~\ref{thm:offline_subroutine} that $\Pr[E_1]\ge1-2\alpha^2$. Similarly, for each $k\in\left\{2,\dots,\frac{1}{\alpha}\right\}$, conditioned on the events $E_j$ for all $j\in[k-1]$, we have that $\Pr[E_k\mid \bigwedge\nolimits_{j\in\left[k-1\right]}E_j]\ge1-2\alpha^2$ by Theorem~\ref{thm:offline_subroutine} (because the approximation guarantee of Subroutine~\ref{sub:continuous_greedy_filtering} during the $k$-th call, as specified by event $E_k$, holds with probability at least $1-2\alpha^2$ by Theorem~\ref{thm:offline_subroutine}, regardless of the previous calls). Hence, we have that
\begin{equation}\label{eq:probability_of_joint_event}
\Pr\Bigl[\bigwedge\nolimits_{k\in\left[\frac{1}{\alpha}\right]}E_k\Bigr]=\prod\nolimits_{k\in\left[\frac{1}{\alpha}\right]}\Pr\Bigl[E_k \mathrel{\Big|} \bigwedge\nolimits_{j\in\left[k-1\right]}E_j\Bigr]\ge(1-2\alpha^2)^{\frac{1}{\alpha}}\ge1-2\alpha.
\end{equation}

\subsubsection*{Step 3: $X_{\alg}^{(\ell)}\cup\left(\bigcup_{j=1}^{\ell-1}O_H^{(j)}\right)$ is a nearly $1-\nicefrac{1}{e}$ approximation for $(f,\M)$, w.h.p.}
Next, we consider the minimum $\ell\in\left[\frac{1}{\alpha}\right]$ (note that $\ell$ is a random variable) such that
\begin{equation}\label{eq:minimum_l_with_less_than_alpha_opt}
f\Bigl(O_H^{(\ell)}\Big|\Bigl(\bigcup_{j=1}^{\ell-1}O_H^{(j)}\Bigr)\Bigr)\le\alpha\cdot f(O^{(1)}),
\end{equation}
which always exists since
\[
\sum_{k\in\left[\frac{1}{\alpha}\right]}f\Bigl(O_H^{(k)}\Big|\Bigl(\bigcup_{j=1}^{k-1}O_H^{(j)}\Bigr)\Bigr)=f\Bigl(\bigcup_{k=1}^{1/\alpha}O_H^{(k)}\Bigr)\le f(O^{(1)}),
\]
where the inequality follows because $\bigcup_{k=1}^{1/\alpha}O_H^{(k)}$ is a subset of $O^{(1)}$, and $O^{(1)}$ is the optimal solution to the original input instance $(f,\M)$. Then, we derive that
\begin{align}\label{eq:1-alpha_opt}
f(O^{(1)})&=f\Bigl(O^{(\ell)}\cup\Bigl(\bigcup_{j=1}^{\ell-1}O_H^{(j)}\Bigr)\Bigr) &&\text{(Since $\textstyle O^{(\ell)}=O^{(1)}\setminus\left(\bigcup_{j=1}^{\ell-1}O_H^{(j)}\right)$)}\nonumber\\
&=f\Bigl(O_L^{(\ell)}\cup\Bigl(\bigcup_{j=1}^{\ell-1}O_H^{(j)}\Bigr)\Bigr)+f\Bigl(O_H^{(\ell)}\Big|O_L^{(\ell)}\cup\Bigl(\bigcup_{j=1}^{\ell-1}O_H^{(j)}\Bigr)\Bigr) &&\text{(Since $O_L^{(\ell)}=O^{(\ell)}\setminus O_H^{(\ell)}$)}\nonumber\\
&\le f\Bigl(O_L^{(\ell)}\cup\Bigl(\bigcup_{j=1}^{\ell-1}O_H^{(j)}\Bigr)\Bigr)+f\Bigl(O_H^{(\ell)}\Big|\bigcup_{j=1}^{\ell-1}O_H^{(j)}\Bigr) &&\text{(By submodularity)}\nonumber\\
&\le f\Bigl(O_L^{(\ell)}\cup\Bigl(\bigcup_{j=1}^{\ell-1}O_H^{(j)}\Bigr)\Bigr)+\alpha\cdot f(O^{(1)}) &&\text{(By Ineq.~\eqref{eq:minimum_l_with_less_than_alpha_opt})}.
\end{align}

Now consider the solution set $X_{\alg}^{(\ell)}\cup\left(\bigcup_{j=1}^{\ell-1}O_H^{(j)}\right)$ (recall that $X_{\alg}^{(\ell)}$ is from the definition of the event $E_{\ell}$), which belongs to $\M$ because $X_{\alg}^{(\ell)}\in\M^{(\ell)}$ and $\M^{(\ell)}=\M/\left(\bigcup_{j=1}^{\ell-1}O_H^{(j)}\right)$. Moreover, since $X_{\alg}^{(\ell)}\subseteq S^{(\ell)}$, we have that $X_{\alg}^{(\ell)}\cup\left(\bigcup_{j=1}^{\ell-1}O_H^{(j)}\right)\subseteq S^{(\ell)}\cup\left(\bigcup_{j=1}^{\ell-1}O_H^{(j)}\right)$, and thus, $X_{\alg}^{(\ell)}\cup\left(\bigcup_{j=1}^{\ell-1}O_H^{(j)}\right)$ is a candidate solution in the exhaustive search step at Line~\ref{algline:offline_exhaustive_search} during the final $\frac{1}{\alpha}$-th call. Conditioned on the joint event $\bigwedge\nolimits_{k\in\left[\frac{1}{\alpha}\right]}E_i$, we lower bound its value as follows,
\begin{align*}
&\,f\Bigl(X_{\alg}^{(\ell)}\cup\Bigl(\bigcup_{j=1}^{\ell-1}O_H^{(j)}\Bigr)\Bigr)\\
=&\,f^{(\ell)}(X_{\alg}^{(\ell)}) &&\text{(By definition of $f^{(\ell)}$)}\\
\ge&\,\left(1-\frac{1}{e}\right)\cdot (f^{(\ell)}(O_L^{(\ell)})-8\alpha^2\cdot f^{(\ell)}(O^{(\ell)})) &&\text{(By event $E_{\ell}$)}\\
=&\,\left(1-\frac{1}{e}\right)\cdot \Bigl(f\Bigl(O_L^{(\ell)}\cup\Bigl(\bigcup_{j=1}^{\ell-1}O_H^{(j)}\Bigr)\Bigr)-8\alpha^2\cdot f(O^{(1)})\Bigr) &&\text{(By definition of $f^{(\ell)}$ and $O^{(\ell)}$)}\\
\ge&\,\left(1-\frac{1}{e}\right)\cdot ((1-\alpha)\cdot f(O^{(1)})-8\alpha^2\cdot f(O^{(1)})) &&\text{(By Ineq.~\eqref{eq:1-alpha_opt})}\\
\ge&\,\left(1-\frac{1}{e}\right)\cdot(1-5\alpha)\cdot f(O^{(1)}) &&\text{(Since $\textstyle\alpha<\frac{1}{2}$)}\\
\ge&\,\left(1-\frac{1}{e}-5\alpha\right)\cdot f(O^{(1)}).
\end{align*}
Then, Theorem~\ref{thm:offline} follows by taking into account the probability of the joint event $\bigwedge\nolimits_{k\in\left[\frac{1}{\alpha}\right]}E_k$, which is at least $1-2\alpha$ by Ineq.~\eqref{eq:probability_of_joint_event}
\end{proof}

\bibliography{cite}

\appendix
\section{Supplementary proofs}\label{sec:supplementary_proofs}
\subsection{Proof of Lemma~\ref{lem:H_i_setminus_H_i-1}}\label{sec:proof_of_lem_H_i_setminus_H_i-1}
\begin{proof}[Proof of Lemma~\ref{lem:H_i_setminus_H_i-1}]
\subsubsection*{Case 1: $i=1$}
We first prove the lemma for $i=1$, in which case Ineq.~\eqref{eq:H_i_setminus_H_i-1} reduces to
\begin{equation}
\Pr\left[|H_1| > \frac{r\ln(r/\eps)}{\eps}\right]\le\frac{\eps}{r},
\end{equation}
because, by definition, we have that $H_0=\emptyset$, $k_0=0$ and $k_1=1$ (the first selector is always effective). Recall that $H_1=G_1$ by definition of $H_1$, and $$G_1=\{e\in[n]\setminus (V_0\cup V_1)\mid \floor{f(\{e\})}_I>\floor{f(\{s_1\})}_I\},$$
by definition of $G_1$ in Eq.~\eqref{eq:G_i}. Thus, $H_1$ consists of elements that appear after the first window $V_1$ in the second phase of Algorithm~\ref{alg:greedy_filtering} and have a strictly higher (rounded) singleton value than the most valuable element $s_1$ in $V_1$. Therefore, $|H_1|>\frac{r\ln(r/\eps)}{\eps}$ holds only if none of the top $\ceil{\frac{r\ln(r/\eps)}{\eps}}$ elements in $[n]\setminus V_0$ with the highest singleton values appear in $V_1$. Because $V_1$ contains more than a uniformly random $\frac{\eps}{r}$ fraction of elements in $[n]\setminus V_0$, the probability that this happens is at most $(1-\frac{\eps}{r})^{\ceil{\frac{r\ln(r/\eps)}{\eps}}}\le\frac{\eps}{r}$. Hence, we have that $\Pr\big[|H_1|>\frac{r\ln(r/\eps)}{\eps}\big]\le\frac{\eps}{r}$.
\subsubsection*{Case 2: $i\in\{2,\dots,r\}$}
Now we prove the lemma for any $i\in\{2,\dots,r\}$. Recall that $S_{\ell}=\{s_1,\dots,s_{\ell}\}$ for each $\ell\in[r]$ and $S_0=\emptyset$ in Algorithm~\ref{alg:greedy_filtering}. We let $T_i^{(1)}$ be 
the set of elements that appear after the $(i-1)$-th window and can be added to set $S_{i-1}$ without violating the matroid constraint, whose (rounded) marginal values with respect to $S_{i-1}$ are no less than $\min_{j\in[i-1]}\floor{f(\{s_j\}|S_{j-1})}_{I}$, i.e.,
\[\textstyle
    T_i^{(1)}:=\{e\in[n]\setminus(\bigcup_{j=0}^{i-1} V_j)\mid S_{i-1}\cup\{e\}\in\M \textrm{ and } \floor{f(\{e\}|S_{i-1})}_{I}\ge\min_{j\in[i-1]}\floor{f(\{s_j\}|S_{j-1})}_{I}\}
\]
Then, we let $T_i^{(2)}$ be the set of elements that appear after the $(i-1)$-th window and can be added to set $S_{i-1}$ without violating the matroid constraint, excluding $T_i^{(1)}$, i.e.,
\[\textstyle
T_i^{(2)}:=\{e\in [n]\setminus(\bigcup_{j=0}^{i-1} V_j)\mid S_{i-1}\cup\{e\}\in\M\}\setminus T_i^{(1)}.
\]
Moreover, we define two events $A_i^{(1)}$ and $A_i^{(2)}$ as follows:
\begin{itemize}
    \item Event $A_i^{(1)}$: None of the elements in $T_i^{(1)}$ appear in the $i$-th window $V_i$.
    \item Event $A_i^{(2)}$: None of the top $\ceil{\frac{r\ln(r/\eps)}{\eps}}-|T_i^{(1)}|$ elements in $T_i^{(2)}$, with the highest marginal values with respect to $S_{i-1}$, appear in the $i$-th window $V_i$. (If $\ceil{\frac{r\ln(r/\eps)}{\eps}}-|T_i^{(1)}|\le 0$ or $\ceil{\frac{r\ln(r/\eps)}{\eps}}-|T_i^{(1)}|>|T_i^{(2)}|$, we assume that this event holds trivially.)
\end{itemize}
Next, we show that conditioned on $|H_{i-1}| \le k_{i-1}\cdot \frac{r\ln(r/\eps)}{\eps}$, the event $|H_i| > k_i\cdot \frac{r\ln(r/\eps)}{\eps}$ implies events $A_i^{(1)}$ and $A_i^{(2)}$.

\subsubsection*{Step 1: Conditioned on $|H_{i-1}| \le k_{i-1}\cdot \frac{r\ln(r/\eps)}{\eps}$, the event $|H_i| > k_i\cdot \frac{r\ln(r/\eps)}{\eps}$ implies $A_i^{(1)}$}
We first prove that conditioned on $|H_{i-1}| \le k_{i-1}\cdot \frac{r\ln(r/\eps)}{\eps}$, if event $A_i^{(1)}$ does not occur, then the event $|H_i| > k_i\cdot \frac{r\ln(r/\eps)}{\eps}$ cannot occur either. Specifically, if $A_i^{(1)}$ does not occur, then there is an element $e_1\in T_i^{(1)}$ that appears in the $i$-th window $V_i$. Given how element $s_i$ is chosen at Line~\ref{algline:streaming_greedy_selection} of Algorithm~\ref{alg:greedy_filtering}, $s_i$ must satisfy that $\floor{f(\{s_{i}\}\mid S_{i-1})}_I\ge\floor{f(\{e_1\}\mid S_{i-1})}_I$. Moreover, since $e_1\in T_i^{(1)}$, it follows from the definition of $T_i^{(1)}$ that $\floor{f(\{e_1\}\mid S_{i-1})}_I\ge\min_{j\in[i-1]}\floor{f(\{s_j\}|S_{j-1})}_{I}$. Therefore, we have that $\floor{f(\{s_{i}\}\mid S_{i-1})}_I\ge\min_{j\in[i-1]}\floor{f(\{s_j\}|S_{j-1})}_{I}$, which implies that $s_i$ is an ineffective selector.
Because $s_i$ is an ineffective selector, we have that $k_i=k_{i-1}$ by definition of $k_i$ and $k_{i-1}$, and that $H_i=H_{i-1}$ by Claim~\ref{claim:ineffective_selector}. Conditioned on $|H_{i-1}| \le k_{i-1}\cdot \frac{r\ln(r/\eps)}{\eps}$, this implies that $|H_i| \le k_i\cdot \frac{r\ln(r/\eps)}{\eps}$, and hence, the event $|H_i| > k_i\cdot \frac{r\ln(r/\eps)}{\eps}$ cannot occur.

\subsubsection*{Step 2: Conditioned on $|H_{i-1}| \le k_{i-1}\cdot \frac{r\ln(r/\eps)}{\eps}$, the event $|H_i| > k_i\cdot \frac{r\ln(r/\eps)}{\eps}$ implies $A_i^{(2)}$}
Now we show that conditioned on $|H_{i-1}| \le k_{i-1}\cdot \frac{r\ln(r/\eps)}{\eps}$, if the event $|H_i| > k_i\cdot \frac{r\ln(r/\eps)}{\eps}$ occurs, then event $A_i^{(2)}$ also occurs (we assume w.l.o.g.~that $\ceil{\frac{r\ln(r/\eps)}{\eps}}-|T_i^{(1)}|\ge1$, since otherwise $A_i^{(2)}$ holds trivially by definition). Specifically, conditioned on $|H_{i-1}| \le k_{i-1}\cdot \frac{r\ln(r/\eps)}{\eps}$, if the event $|H_i| > k_i\cdot \frac{r\ln(r/\eps)}{\eps}$ occurs, then it must hold that $k_i=k_{i-1}+1$ (because otherwise $k_i=k_{i-1}$, and $s_i$ would be an ineffective selector, which would imply that $|H_i|=|H_{i-1}|\le k_{i-1}\cdot \frac{r\ln(r/\eps)}{\eps}=k_i\cdot \frac{r\ln(r/\eps)}{\eps}$ by Claim~\ref{claim:ineffective_selector}), and hence, $|H_i\setminus H_{i-1}|=|H_i|-|H_{i-1}|>\frac{r\ln(r/\eps)}{\eps}$. By definition of $H_i$ and $H_{i-1}$ in Eq.~\eqref{eq:H_i}, we have that $H_i\setminus H_{i-1}=G_i\setminus(\bigcup_{j=1}^{i-1} G_j)$, which implies that $|G_i|\ge|H_i\setminus H_{i-1}|>\frac{r\ln(r/\eps)}{\eps}$. This in turn implies that $|G_i\setminus T_i^{(1)}|\ge|G_i|-|T_i^{(1)}|>\frac{r\ln(r/\eps)}{\eps}-|T_i^{(1)}|$.

Note that by definition of $G_i$ (Eq.~\eqref{eq:G_i}) and $T_i^{(1)},T_i^{(2)}$, the set $G_i\setminus T_i^{(1)}$ is a subset of elements in $T_i^{(2)}$ that have higher (rounded) marginal values with respect to $S_{i-1}$ than element $s_i$. Since $|G_i\setminus T_i^{(1)}|>\frac{r\ln(r/\eps)}{\eps}-|T_i^{(1)}|$, there must be \emph{strictly} more than $\frac{r\ln(r/\eps)}{\eps}-|T_i^{(1)}|$ elements in $T_i^{(2)}$ that have higher (rounded) marginal values with respect to $S_{i-1}$ than element $s_i$. Next, we derive a contradiction to this, assuming event $A_i^{(2)}$ does not occur.

Suppose for contradiction that $A_i^{(2)}$ does not occur. That is, one of the top $\ceil{\frac{r\ln(r/\eps)}{\eps}}-|T_i^{(1)}|$ elements in $T_i^{(2)}$ with the highest marginal values with respect to $S_{i-1}$, which we denote by $e_2\in T_i^{(2)}$, appears in the $i$-th window $V_i$. Given how element $s_i$ is chosen at Line~\ref{algline:streaming_greedy_selection} of Algorithm~\ref{alg:greedy_filtering}, $s_i$ must have at least the same (rounded) marginal value with respect to $S_{i-1}$ as element $e_2$. Hence, the (rounded) marginal value of element $s_i$ is no less than that of any element in $T_i^{(2)}$, except for the other top $\ceil{\frac{r\ln(r/\eps)}{\eps}}-|T_i^{(1)}|-1$ elements (excluding $e_2$), which is a contradiction.

\subsubsection*{Step 3: Proving Ineq.~\eqref{eq:H_i_setminus_H_i-1} under the additional condition $\ceil{\frac{r\ln(r/\eps)}{\eps}}>|T_i^{(1)}|+|T_i^{(2)}|$}
Next, we prove that Ineq.~\eqref{eq:H_i_setminus_H_i-1} holds trivially if we additionally condition on the event $\ceil{\frac{r\ln(r/\eps)}{\eps}}>|T_i^{(1)}|+|T_i^{(2)}|$. Specifically, we show that
\begin{equation}\label{eq:H_i_setminus_H_i-1_not_enough_feasible_elements}
\Pr\left[|H_i| \le k_i\cdot \frac{r\ln(r/\eps)}{\eps} \,\middle\vert\, |H_{i-1}| \le k_{i-1}\cdot \frac{r\ln(r/\eps)}{\eps},\, \ceil{\frac{r\ln(r/\eps)}{\eps}}>|T_i^{(1)}|+|T_i^{(2)}|\right]=1.
\end{equation}

To this end, we notice that conditioned on $\ceil{\frac{r\ln(r/\eps)}{\eps}}>|T_i^{(1)}|+|T_i^{(2)}|$, the set $T_i^{(1)}\cup T_i^{(2)}$, which consists of all elements that appear after the $(i-1)$-th window and can be added to $S_{i-1}$ without violating the matroid constraint, has size at most $\ceil{\frac{r\ln(r/\eps)}{\eps}}-1\le\frac{r\ln(r/\eps)}{\eps}$. By definition of $G_i$ in Eq.~\eqref{eq:G_i}, we have that $G_i\subseteq T_i^{(1)}\cup T_i^{(2)}$, and hence, $|G_i|\le\frac{r\ln(r/\eps)}{\eps}$. This implies that $|H_i\setminus H_{i-1}|\le\frac{r\ln(r/\eps)}{\eps}$ since $H_i\setminus H_{i-1}\subseteq G_i$. Conditioned on $|H_{i-1}| \le k_{i-1}\cdot \frac{r\ln(r/\eps)}{\eps}$, it follows that $|H_i|=|H_i\setminus H_{i-1}|+|H_{i-1}|\le (k_{i-1}+1)\cdot \frac{r\ln(r/\eps)}{\eps}$, which implies the event $|H_i|\le k_i\cdot \frac{r\ln(r/\eps)}{\eps}$ if $k_i=k_{i-1}+1$. If instead $k_i=k_{i-1}$, then element $s_i$ is an ineffective selector, and hence, by Claim~\ref{claim:ineffective_selector}, we have that $|H_i|=|H_{i-1}|$, which also implies the event $|H_i|\le k_i\cdot \frac{r\ln(r/\eps)}{\eps}$, conditioned on $|H_{i-1}| \le k_{i-1}\cdot \frac{r\ln(r/\eps)}{\eps}$. Thus, Ineq.~\eqref{eq:H_i_setminus_H_i-1_not_enough_feasible_elements} follows.

\subsubsection*{Step 4: Proving Ineq.~\eqref{eq:H_i_setminus_H_i-1} under the additional condition $\ceil{\frac{r\ln(r/\eps)}{\eps}}\le|T_i^{(1)}|+|T_i^{(2)}|$}
Finally, we prove that Ineq.~\eqref{eq:H_i_setminus_H_i-1} also holds if we additionally condition on $\ceil{\frac{r\ln(r/\eps)}{\eps}}\le|T_i^{(1)}|+|T_i^{(2)}|$.
Since we have shown that conditioned on $|H_{i-1}| \le k_{i-1}\cdot \frac{r\ln(r/\eps)}{\eps}$, the event $|H_i| > k_i\cdot \frac{r\ln(r/\eps)}{\eps}$ implies events $A_i^{(1)}$ and $A_i^{(2)}$, it suffices to show that $$\Pr\left[A_i^{(1)}\wedge A_i^{(2)}\,\middle\vert\, \left(\forall\,j\in\{0,\dots,i-1\},\,|H_j|\le k_j\cdot \frac{r\ln(r/\eps)}{\eps}\right),\,\ceil{\frac{r\ln(r/\eps)}{\eps}}\le|T_i^{(1)}|+|T_i^{(2)}|\right]\le\frac{\eps}{r}.$$
Moreover, observe that the numbers $k_1,\dots,k_{i-1}$, the sets $H_1,\dots,H_{i-1}$, $T_{i}^{(1)}$ and $T_{i}^{(2)}$, and the events $A_{i}^{(1)}$ and $A_{i}^{(2)}$ are fully determined by the first $i-1$ windows $V_0,\dots,V_{i-1}$ in Algorithm~\ref{alg:greedy_filtering}. Hence, it suffices to prove that $\Pr[A_i^{(1)}\wedge A_i^{(2)} \mid V_0,\dots,V_{i-1}]\le\frac{\eps}{r}$, for any fixed $V_0,\dots,V_{i-1}$ that are consistent with the condition $\ceil{\frac{r\ln(r/\eps)}{\eps}}\le|T_i^{(1)}|+|T_i^{(2)}|$.

To this end, we consider any $V_0,\dots,V_{i-1}$ that satisfy the condition $\ceil{\frac{r\ln(r/\eps)}{\eps}}\le|T_i^{(1)}|+|T_i^{(2)}|$. We notice that since the $i$-th window $V_i$ contains more than a uniformly random $\frac{\eps}{r}$ fraction of elements in $[n]\setminus(\bigcup_{j=0}^{i-1} V_j)$, the probability that none of the elements in $T_i^{(1)}$ appear in $V_i$ is at most $(1-\frac{\eps}{r})^{|T_i^{(1)}|}$, namely,
\begin{equation}\label{eq:A_i_1_probability}
\Pr\left[A_i^{(1)} \,\middle\vert\, V_0,\dots,V_{i-1}\right]\le\left(1-\frac{\eps}{r}\right)^{|T_i^{(1)}|}.
\end{equation}
If $V_0,\dots,V_{i-1}$ further satisfy the condition $\ceil{\frac{r\ln(r/\eps)}{\eps}}-|T_i^{(1)}|\ge1$, then it holds that $1\le\ceil{\frac{r\ln(r/\eps)}{\eps}}-|T_i^{(1)}|\le|T_i^{(2)}|$, and hence, the top $\ceil{\frac{r\ln(r/\eps)}{\eps}}-|T_i^{(1)}|$ elements in $T_i^{(2)}$, with the highest marginal values with respect to $S_{i-1}$, are well-defined. The probability that none of these top $\ceil{\frac{r\ln(r/\eps)}{\eps}}-|T_i^{(1)}|$ elements of $T_i^{(2)}$ appear in $V_i$ is at most $(1-\frac{\eps}{r})^{\ceil{\frac{r\ln(r/\eps)}{\eps}}-|T_i^{(1)}|}$, and this probability can only decrease if we additionally condition on event $A_i^{(1)}$. Hence, when $V_0,\dots,V_{i-1}$ further satisfy $\ceil{\frac{r\ln(r/\eps)}{\eps}}-|T_i^{(1)}|\ge1$, we have that 
\begin{equation}\label{eq:A_i_2_probability}
\Pr\left[A_i^{(2)} \,\middle\vert\, V_0,\dots,V_{i-1} \textrm{ and } A_i^{(1)}\right]\le\left(1-\frac{\eps}{r}\right)^{\ceil{\frac{r\ln(r/\eps)}{\eps}}-|T_i^{(1)}|}.
\end{equation}
On the other hand, if $V_0,\dots,V_{i-1}$ do not satisfy $\ceil{\frac{r\ln(r/\eps)}{\eps}}-|T_i^{(1)}|\ge1$, which implies that $\ceil{\frac{r\ln(r/\eps)}{\eps}}-|T_i^{(1)}|\le0$, then Ineq.~\eqref{eq:A_i_2_probability} holds trivially. Thus, Ineq.~\eqref{eq:A_i_2_probability} holds regardless of whether $V_0,\dots,V_{i-1}$ satisfy $\ceil{\frac{r\ln(r/\eps)}{\eps}}-|T_i^{(1)}|\ge1$. Finally, we derive that
\begin{align*}
    &\Pr\left[A_i^{(1)}\wedge A_i^{(2)} \,\middle\vert\, V_0,\dots,V_{i-1}\right]\\
    =&\Pr\left[A_i^{(1)}\,\middle\vert\, V_0,\dots,V_{i-1}\right]\times
    \Pr\left[A_i^{(2)}\,\middle\vert\, V_0,\dots,V_{i-1} \textrm{ and } A_i^{(1)}\right]\\
    \le& \left(1-\frac{\eps}{r}\right)^{|T_i^{(1)}|}\times
    \Pr\left[A_i^{(2)}\,\middle\vert\, V_0,\dots,V_{i-1} \textrm{ and } A_i^{(1)}\right] &&\text{(By Ineq.~\eqref{eq:A_i_1_probability})}\\
    \le& \left(1-\frac{\eps}{r}\right)^{|T_i^{(1)}|}\times\left(1-\frac{\eps}{r}\right)^{\ceil{\frac{r\ln(r/\eps)}{\eps}}-|T_i^{(1)}|} &&\text{(By Ineq.~\eqref{eq:A_i_2_probability})}\\
    =& \left(1-\frac{\eps}{r}\right)^{\ceil{\frac{r\ln(r/\eps)}{\eps}}}\le\frac{\eps}{r},
\end{align*}
which finishes the proof of Lemma~\ref{lem:H_i_setminus_H_i-1}.
\end{proof}

\subsection{Proof of Lemma~\ref{lem:streaming_regular_case}}\label{sec:proof_of_lem_streaming_regular_case}
\begin{proof}[Proof of Lemma~\ref{lem:streaming_regular_case}]
\subsubsection*{The proof setup}
We begin by introducing the notations that will be used in the analysis and deriving their relations using basic structural properties of matroids. First, we denote $O_{\final}:=O\cap V_{\final}$ and consider the restriction $\M_{S_r\cup O_{\final}}$ of matroid $\M$ to set $S_r\cup O_{\final}$. Because the restriction $\M_{S_r\cup O_{\final}}$ is a matroid, we can augment $S_r\in\M_{S_r\cup O_{\final}}$ with some subset $O_S\subseteq O_{\final}$ such that $S_r\cupdot O_S$ is a base of $\M_{S_r\cup O_{\final}}$, and similarly, we can augment $O_{\final}\in\M_{S_r\cup O_{\final}}$ with some subset $S_O\subseteq S_r$ such that $O_{\final}\cupdot S_O$ is a base of $\M_{S_r\cup O_{\final}}$.

Then, we partition the set $O_{\final}$ into two disjoint subsets: $O_L:=O_{\final}\setminus H$ and $O_H:=O_{\final}\cap H$. That is, $O_H$ is the subset of elements from $O_{\final}$ that are included in $H$ by Algorithm~\ref{alg:greedy_filtering}, and $O_L$ is the subset of elements from $O_{\final}$ that are not included in $H$ (importantly, conditioned on event $E_2$ defined in Lemma~\ref{lem:streaming_unlikely_case}, these elements were not added to $H$ because they did not satisfy the filter condition at Line~\ref{algline:streaming_filter_condition} of Algorithm~\ref{alg:greedy_filtering}, not because of the break condition at Line~\ref{algline:streaming_break_condition}). We denote $O_L':=O_L\setminus O_S$ and $O_H':=O_H\setminus O_S$. Consequently, we have the following relation:
\begin{equation*}
O_S\cup O_H'=((O_S\cap O_L)\cup(O_S\cap O_H))\cup O_H'=(O_S\cap O_L)\cup O_H,
\end{equation*}
and thus, $O_{\final}=O_L\cup O_H=O_S\cupdot O_L'\cupdot O_H'$ can be decomposed as follows,
\begin{equation}\label{eq:O_final_decomposition}
O_{\final}=O_L'\cup (O_S\cup O_H')=O_L'\cup (O_S\cap O_L)\cup O_H.
\end{equation}
Moreover, the base $O_{\final}\cupdot S_O$ of matroid $\M_{S_r\cup O_{\final}}$ can now be expressed as $O_S\cupdot O_L'\cupdot O_H'\cupdot S_O$. We note that the other base $S_r\cupdot O_S$ of matroid $\M_{S_r\cup O_{\final}}$ can be written as $O_S\cupdot(S_r\setminus S_O)\cupdot S_O$.

Now we consider the contraction $\M_{S_r\cup O_{\final}}/(O_S\cupdot S_O)$ of matroid $\M_{S_r\cup O_{\final}}$. Notice that both $O_L'\cupdot O_H'$ and $S_r\setminus S_O$ are bases of matroid $\M_{S_r\cup O_{\final}}/(O_S\cupdot S_O)$. Hence, by Lemma~\ref{lem:base_exchange}, there exists a partition $S_r\setminus S_O=S_L\cupdot S_H$ such that both $S_L\cupdot O_H'$ and $O_L'\cupdot S_H$ are bases of $\M_{S_r\cup O_{\final}}/(O_S\cupdot S_O)$, which implies that both $O_S\cupdot S_L\cupdot O_H'\cupdot S_O$ and $O_S\cupdot O_L'\cupdot S_H\cupdot S_O$ are bases of $\M_{S_r\cup O_{\final}}$.

In particular, since both $O_S\cupdot O_L'\cupdot S_H\cupdot S_O$ and $O_S\cupdot S_L\cupdot S_H\cupdot S_O=O_S\cupdot S_r$ are bases of $\M_{S_r\cup O_{\final}}$, by Lemma~\ref{lem:base_matching}, there exists a bijection $h:S_L\to O_L'$ such that for all $s\in S_L$,
\[
\textrm{$((O_S\cupdot S_L\cupdot S_H\cupdot S_O)\setminus\{s\})\cup\{h(s)\}$ is a base of $\M_{S_r\cup O_{\final}}$.}
\]
Since $S_L\cupdot S_H\cupdot S_O=S_r$, it follows that $(S_r\setminus\{s\})\cup\{h(s)\}\in\M_{S_r\cup O_{\final}}$ for all $s\in S_L$, which implies that for all $s\in S_L$,
\begin{equation}\label{eq:h_maps_S_L_to_O_L'}
(S_r\setminus\{s\})\cup\{h(s)\}\in\M.
\end{equation}
(To help the reader keep track of the notations, we illustrate the relations between $O_S\cupdot O_L'\cupdot O_H'\cupdot S_O$ and $O_S\cupdot S_L\cupdot S_H\cupdot S_O$, and between $O_S\cupdot S_L\cupdot S_H\cupdot S_O$ and $O_S\cupdot O_L'\cupdot S_H\cupdot S_O$ in Figure~\ref{fig:illustration_streaming}.)

\begin{figure}[ht]
    \centering
    \includegraphics[scale=0.3]{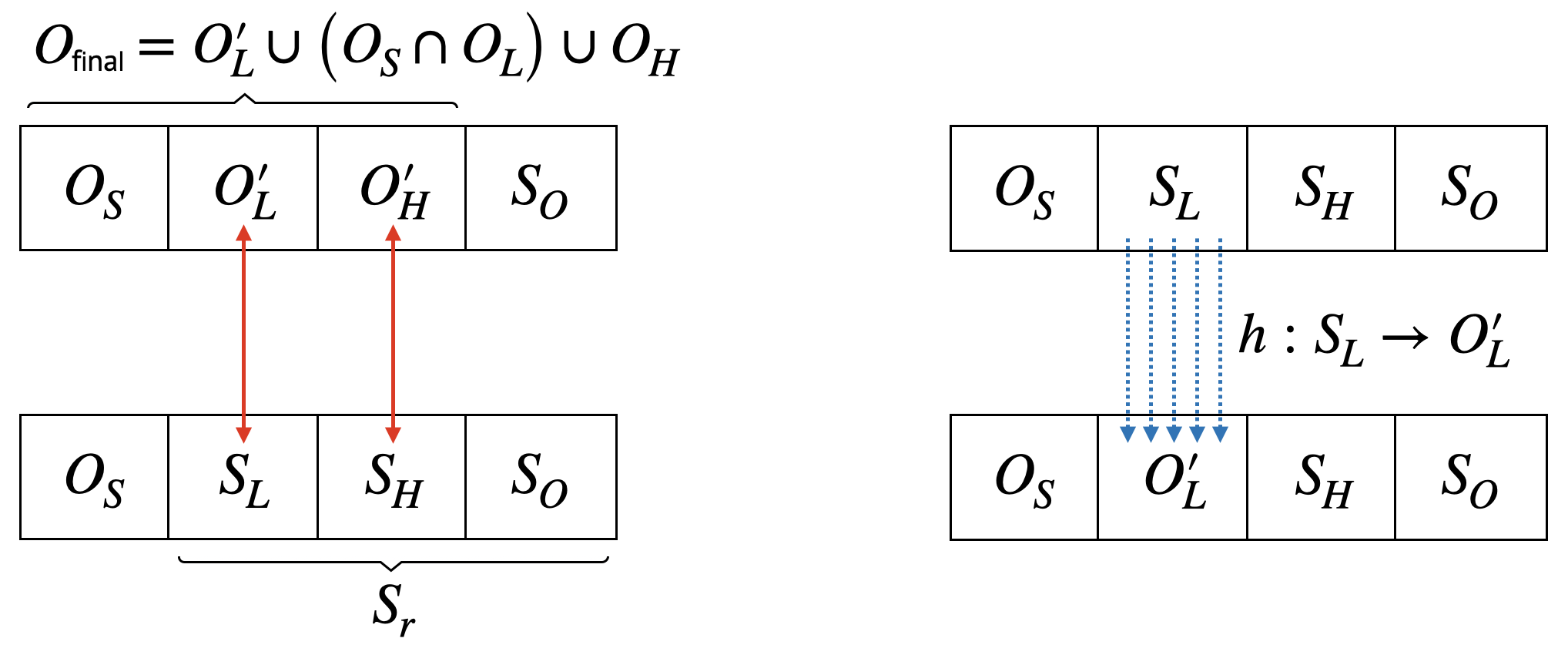}
    \caption{On the left, both $O_S\cupdot O_L'\cupdot O_H'\cupdot S_O$ and $O_S\cupdot S_L\cupdot S_H\cupdot S_O$ are bases of matroid $\M_{S_r\cup O_{\final}}$, and we can exchange $O_L'$ with $S_L$, or $O_H'$ with $S_H$, resulting in new bases. On the right, both $O_S\cupdot S_L\cupdot S_H\cupdot S_O$ and $O_S\cupdot O_L'\cupdot S_H\cupdot S_O$ are bases of matroid $\M_{S_r\cup O_{\final}}$, and there is a bijection $h:S_L\to O_L'$ such that for all $s\in S_L$, replacing element $s$ in $O_S\cupdot S_L\cupdot S_H\cupdot S_O$ with element $h(s)$ yields a new base. Finally, we remind the reader that $O_{\final}=O_S\cupdot O_L'\cupdot O_H'$ and $S_r=S_L\cupdot S_H\cupdot S_O$, and $O_{\final}$ can be also be decomposed as $O_L'\cup (O_S\cap O_L)\cup O_H$ by Eq.~\eqref{eq:O_final_decomposition}.}
    \label{fig:illustration_streaming}
\end{figure}

Furthermore, since we assume that $\max_{e\in [n]}f(\{e\})< \frac{r}{\eps}\cdot\min_{e\in T}f(\{e\})$ and condition on event $E_1$ in the lemma statement, we have that $w\ge\max_{e\in[n]}f(\{e\})$ by Lemma~\ref{lem:streaming_w_upper_bound_max_singleton_value}. Thus, we will treat $w$ as an upper bound on the singleton value of any element throughout the proof.

Next, we upper bound the marginal values of two subsets of the optimal solution, $O_S\cap O_L$ and $O_L'$, with respect to the solution set $S_r$.

\subsubsection*{Step 1: Upper bounding $f(O_S\cap O_L|S_r)$}
We consider two cases $O_S\cap O_L\neq\emptyset$ and $O_S\cap O_L=\emptyset$ and prove Ineq.~\eqref{eq:O_S_cap_O_L} for both cases.
\begin{equation}\label{eq:O_S_cap_O_L}
f(O_S\cap O_L|S_r)\le\frac{\eps^2 w}{(1-\eps)r}.
\end{equation}
\paragraph{Case 1: $O_S\cap O_L\neq\emptyset$.} Because $S_r\cupdot O_S$ is a base of $\M_{S_r\cup O_{\final}}$, we have that $S_r\cupdot O_S\in\M$. Because $O_S$ is non-empty and $\M$ has rank $r$, it follows that $|S_r|<r$. Given how $S_r$ is constructed in Algorithm~\ref{alg:greedy_filtering}, if $|S_r|<r$, then there must exist $i\in\{2,\dots,r\}$ for which the algorithm sets $s_i$ to $s_1$. It follows that $f(\{s_i\}|S_{i-1})=f(\{s_1\}|S_{i-1})=0$.

Moreover, for any element $o\in O_S\cap O_L$, since $S_r\cupdot O_S\in\M$, element $o$ can be added to $S_r$ without violating the matroid constraint, which implies that $S_{i-1}\cup\{o\}\in\M$. However, since $o\in O_L$ (and hence $o\notin H$), element $o$ was not selected by the selector $s_i$, despite $S_{i-1}\cup\{o\}\in\M$. Therefore, it must hold that $\floor{f(\{o\}|S_{i-1})}_I
\le\floor{f(\{s_i\}|S_{i-1})}_I=\floor{0}_I$. Since $f(\{o\}|S_{i-1})\le f(\{o\})\le w$ by submodularity, we can apply Claim~\ref{claim:streaming_rounding_error} (by setting $a$, $b$, $\gamma$ to be $0$, $f(\{o\}|S_{i-1})$, $\frac{\eps^2}{r^2}$, respectively), and we obtain that $f(\{o\}|S_{i-1})\le\frac{\eps^2 w}{(1-\eps)r^2}$. Hence, we can upper bound $f(O_S\cap O_L|S_r)$ as follows,
\begin{align*}
f(O_S\cap O_L|S_r)&\le\sum_{o\in O_S\cap O_L}f(\{o\}|S_r) &&\text{(By submodularity)}\nonumber\\
&\le\sum_{o\in O_S\cap O_L}f(\{o\}|S_{i-1}) &&\text{(By submodularity)}\nonumber\\
&\le|O_S\cap O_L|\cdot\frac{\eps^2 w}{(1-\eps)r^2}\le r\cdot\frac{\eps^2 w}{(1-\eps)r^2}=\frac{\eps^2 w}{(1-\eps)r}.
\end{align*}
\paragraph{Case 2: $O_S\cap O_L=\emptyset$.} We note that Ineq.~\eqref{eq:O_S_cap_O_L} holds trivially if $O_S\cap O_L=\emptyset$.

\subsubsection*{Step 2: Upper bounding $f(O_L'|S_r)$}
We denote $\ell:=|S_L|$. Since $S_L\subseteq S_r$ and given how $S_r$ is constructed in Algorithm~\ref{alg:greedy_filtering}, the elements in $S_L$ are $s_{i_1},\dots,s_{i_{\ell}}$ for some indices $i_1<i_2<\cdots<i_{\ell}$ in $[r]$. Then, we let $S_{L,\,j}:=\{s_{i_1},\dots,s_{i_j}\}$ for all $j\in[\ell]$, and let $S_{L,0}:=\emptyset$. We observe that for all $j\in[\ell]$, it holds that $S_{L,\,j-1}\subseteq S_{i_j-1}$ because $i_{j-1}\le i_j-1$. Recall that, as shown in Eq.~\eqref{eq:h_maps_S_L_to_O_L'}, there is a bijection $h:S_L\to O_L'$ such that $(S_r\setminus\{s_{i_j}\})\cup\{h(s_{i_j})\}\in\M$ for all $j\in[\ell]$, which implies that for all $j\in[\ell]$, $(S_{i_j}\setminus\{s_{i_j}\})\cup\{h(s_{i_j})\}\in\M$, since $S_{i_j}\subseteq S_r$. This is equivalent to $S_{i_j-1}\cup\{h(s_{i_j})\}\in\M$, since $S_{i_j-1}=S_{i_j}\setminus\{s_{i_j}\}$.

Furthermore, for any $j\in[\ell]$, since $h(s_{i_j})\in O_L'$ (and hence $h(s_{i_j})\notin H$), element $h(s_{i_j})$ was not selected by the selector $s_{i_j}$, despite $S_{i_j-1}\cup\{h(s_{i_j})\}\in\M$. Thus, it must hold that $\floor{f(\{h(s_{i_j})\}|S_{i_j-1})}_I\le\floor{f(\{s_{i_j}\}|S_{i_j-1})}_I$. Note that $f(\{s_{i_j}\}|S_{i_j-1})\ge0$, because element $s_{i_j}$ should have at least the same marginal value as element $s_1$, by Line~\ref{algline:streaming_greedy_selection} of Algorithm~\ref{alg:greedy_filtering}. Moreover, we have that $f(\{h(s_{i_j})\}|S_{i_j-1})\le f(\{h(s_{i_j})\})\le w$ by submodularity. Hence, we can apply Claim~\ref{claim:streaming_rounding_error} (by setting $a$, $b$, $\gamma$ to be $f(\{s_{i_j}\}|S_{i_j-1})$, $f(\{h(s_{i_j})\}|S_{i_j-1})$, $\frac{\eps^2}{r^2}$, respectively), and we obtain that
\begin{equation}\label{eq:s_i_j_vs_o_i_j}
f(\{s_{i_j}\}|S_{i_j-1})\ge(1-\eps)\cdot f(\{h(s_{i_j})|S_{i_j-1})-\frac{\eps^2w}{r^2}.
\end{equation}
Now we lower bound $f(S_L)$ as follows,
\begin{align*}
f(S_L)&=\sum_{j=1}^{\ell} f(\{s_{i_j}\}|S_{L,\,j-1}) &&\text{(By telescoping sum)}\nonumber\\
&\ge\sum_{j=1}^{\ell} f(\{s_{i_j}\}|S_{i_j-1}) &&\text{(By $S_{L,\,j-1}\subseteq S_{i_j-1}$ and submodularity)}\nonumber\\
&\ge\sum_{j=1}^{\ell} (1-\eps)\cdot f(\{h(s_{i_j})\}|S_{i_j-1})-\ell\cdot\frac{\eps^2w}{r^2} &&\text{(By Ineq.~\eqref{eq:s_i_j_vs_o_i_j})}\nonumber\\
&\ge\sum_{o\in O_L'} (1-\eps)\cdot f(\{o\}|S_{i_j-1})-\ell\cdot\frac{\eps^2w}{r^2} &&\text{(Since $h:S_L\to O_L'$ is a bijection)}\nonumber\\
&\ge (1-\eps)\cdot f(O_L'|S_{i_j-1})-\ell\cdot\frac{\eps^2w}{r^2} &&\text{(By submodularity)}\nonumber\\
&\ge (1-\eps)\cdot f(O_L'|S_r)-\ell\cdot\frac{\eps^2w}{r^2} &&\text{(By submodularity)}\nonumber\\
&\ge (1-\eps)\cdot f(O_L'|S_r)-\frac{\eps^2w}{r} &&\text{(Since $\ell\le r$)},
\end{align*}
which implies that
\begin{equation}\label{eq:S_L_vs_O_L'}
    f(O_L'|S_r)\le \frac{f(S_L)}{1-\eps}+\frac{\eps^2w}{(1-\eps)r}.
\end{equation}

\subsubsection*{Step 3: One of the three sets $S_L,S_r,S_L\cup O_H$ is a good solution}
Next, we upper bound $f(O_{\final}\cup S_r)$ using solution values of $S_L,S_r,S_L\cup O_H$, as follows,
\begin{align}\label{eq:upper_bound_O_final_cup_S_r}
f(O_{\final}\cup S_r)&=f(O_L'\cup (O_S\cap O_L) \cup O_H\cup S_r)\nonumber\\
&\qquad\text{(By Eq.~\eqref{eq:O_final_decomposition})}\nonumber\\
&=f(O_L'\cup (O_S\cap O_L)|S_r\cup O_H)+f(S_r\cup O_H)\nonumber\\
&=f(O_L'\cup (O_S\cap O_L)|S_r\cup O_H)+f(S_O\cup S_L\cup S_H\cup O_H)\nonumber\\
&\qquad\textrm{(Since $S_r=S_O\cupdot S_L\cupdot S_H$)}\nonumber\\
&=f(O_L'\cup (O_S\cap O_L)|S_r\cup O_H)+f(S_O\cup S_H|S_L\cup O_H)+f(S_L\cup O_H)\nonumber\\
&\le f(O_L'\cup (O_S\cap O_L)|S_r)+f(S_O\cup S_H|S_L)+f(S_L\cup O_H)\nonumber\\
&\qquad\text{(By submodularity)}\nonumber\\
&\le f(O_L'|S_r)+f(O_S\cap O_L|S_r)+f(S_O\cup S_H|S_L)+f(S_L\cup O_H)\nonumber\\
&\qquad\text{(By submodularity)}\nonumber\\
&\le \frac{f(S_L)}{1-\eps}+\frac{2\eps^2w}{(1-\eps)r}+f(S_O\cup S_H|S_L)+f(S_L\cup O_H)\nonumber\\
&\qquad\text{(By Ineq.~\eqref{eq:O_S_cap_O_L} and Ineq.~\eqref{eq:S_L_vs_O_L'})}\nonumber\\
&\le \frac{f(S_L)}{1-\eps}+\frac{4\eps^2w}{r}+f(S_O\cup S_H|S_L)+f(S_L\cup O_H)\nonumber\\
&\qquad\text{(Since $\textstyle\eps<\frac{1}{2}$)}\nonumber\\
&= \frac{\eps\cdot f(S_L)}{1-\eps}+f(S_O\cup S_H\cup S_L)+f(S_L\cup O_H)+\frac{4\eps^2w}{r}\nonumber\\
&=\frac{\eps\cdot f(S_L)}{1-\eps}+f(S_r)+f(S_L\cup O_H)+\frac{4\eps^2w}{r}\nonumber\\
&\qquad\textrm{(Since $S_r=S_O\cupdot S_L\cupdot S_H$)}\nonumber\\
&=\frac{2-\eps}{1-\eps}\cdot\left(\frac{\eps}{2-\eps}f(S_L)+\frac{1-\eps}{2-\eps}\cdot f(S_r)+\frac{1-\eps}{2-\eps}\cdot f(S_L\cup O_H)\right)+\frac{4\eps^2w}{r}.
\end{align}
Moreover, because $S_r$ is a subset of elements constructed by Algorithm~\ref{alg:greedy_filtering} in the $r$ windows $\bigcup_{i=1}^r V_i$, conditioned on event $E_3$ in the lemma statement, we have that $f(S_r|O_{\final})\ge f(S_r|O)\ge-\eta\cdot f(O)$ (where the first inequality is by submodularity), and thus, $f(O_{\final}\cup S_r)\ge f(O_{\final})-\eta\cdot f(O)$. Combining this with Ineq.~\eqref{eq:upper_bound_O_final_cup_S_r}, we obtain that
\[
\frac{\eps}{2-\eps}f(S_L)+\frac{1-\eps}{2-\eps}\cdot f(S_r)+\frac{1-\eps}{2-\eps}\cdot f(S_L\cup O_H)\ge\frac{1-\eps}{2-\eps}\cdot\left(f(O_{\final})-\eta\cdot f(O)-\frac{4\eps^2w}{r}\right).
\]
Recall that $w=\frac{r}{\eps}\cdot \max_{e\in V_0}f(\{e\})$ in Algorithm~\ref{alg:greedy_filtering}, and hence, $w\le\frac{r}{\eps}\cdot f(O)$. It follows that
\begin{align*}
&\frac{\eps}{2-\eps}f(S_L)+\frac{1-\eps}{2-\eps}\cdot f(S_r)+\frac{1-\eps}{2-\eps}\cdot f(S_L\cup O_H)\\
\ge&\,\frac{1-\eps}{2-\eps}\cdot(f(O_{\final})-\eta\cdot f(O)-4\eps\cdot f(O))\\
\ge&\left(\frac{1}{2}-\eps\right)\cdot(f(O_{\final})-\eta\cdot f(O)-4\eps\cdot f(O)) &&\text{(Since $\textstyle0<\eps<\frac{1}{2}$)}\\
\ge&\left(\frac{1}{2}-\eps\right)\cdot (f(O_{\final})-5\eta\cdot f(O)) &&\text{(Since $\textstyle\eta=\sqrt{\frac{2\eps n+2r}{n}}\ge\eps$)}.
\end{align*}
Therefore, one of the sets $S_L,S_r,S_L\cup O_H$ must have a value at least $\left(\frac{1}{2}-\eps\right)\cdot (f(O_{\final})-5\eta\cdot f(O))$.
We finish the proof by noticing that $S_L$, $S_r$, and $S_L\cup O_H$ are all feasible solutions and are considered in the final exhaustive search of Algorithm~\ref{alg:greedy_filtering}.
\end{proof}

\subsection{Proof of Theorem~\ref{thm:streaming}}\label{sec:proof_of_thm_streaming}
\begin{proof}[Proof of Theorem~\ref{thm:streaming}]
We assume w.l.o.g.~that $\max_{e\in [n]}f(\{e\})< \frac{r}{\eps}\cdot\min_{e\in T}f(\{e\})$, since otherwise Theorem~\ref{thm:streaming} follows immediately from Lemma~\ref{lem:streaming_easy_case}. Recall that $V_{\final}=\{\pi(i) \mid i=\ceil{\eps n}+r\cdot\ceil{\frac{\eps n}{r}}+1,\dots,n\}$, and hence, $|V_{\final}|=n-\ceil{\eps n}-r\cdot\ceil{\frac{\eps n}{r}}\ge n-(\eps n+1)-r(\frac{\eps n}{r}+1)\ge n-2\eps n-2r$. Since $V_{\final}$ contains exactly $\frac{|V_{\final}|}{n}$ fraction of elements in the random stream $\pi$, every element in $[n]$ appears in $V_{\final}$ with probability exactly $\frac{|V_{\final}|}{n}\ge1-2\eps-\frac{2r}{n}$. It follows by Lemma~\ref{lem:subsample_exactly} that
\begin{equation}\label{eq:O_cap_V_final_vs_O}
\E[f(O\cap V_{\final})]\ge\frac{|V_{\final}|}{n}\cdot f(O)\ge\left(1-2\eps-\frac{2r}{n}\right)\cdot f(O).
\end{equation}

Now let $E_1$ and $E_2$ be the events defined in Lemma~\ref{lem:streaming_w_upper_bound_max_singleton_value} and Lemma~\ref{lem:H_i_setminus_H_i-1}. By Lemma~\ref{lem:streaming_w_upper_bound_max_singleton_value} and Lemma~\ref{lem:H_i_setminus_H_i-1}, we have that $\Pr[E_1]\ge1-\eps$ and $\Pr[E_2]\ge1-\eps$. Moreover, let $E_3$ be the event that $f(X|O)\ge-\eta \cdot f(O)$ for all $X\subseteq\bigcup_{i=1}^r V_i$, with $\eta=\sqrt{\frac{2\eps n+2r}{n}}$, as defined in Lemma~\ref{lem:streaming_regular_case}. Notice that $|\bigcup_{i=1}^r V_i|=\ceil{\eps n}+r\cdot\ceil{\frac{\eps n}{r}}\le2\eps n+2r$ in Algorithm~\ref{alg:greedy_filtering}, and thus, in the random stream $\pi$, the probability that an element $e\in[n]$ appears in $\bigcup_{i=1}^r V_i$ is at most $\frac{2\eps n+2r}{n}$. It follows from Claim~\ref{claim:eps_fraction_does_not_hurt} that $\Pr[E_3]\ge1-\frac{2\eps n+2r}{n\cdot \eta}=\eta$. By a union bound, we have that
\begin{equation}\label{eq:union_bound_E1_E2_E3}
\Pr[E_1\wedge E_2\wedge E_3]\ge 1-2\eps-\eta,
\end{equation}
and we derive that
\begin{align}\label{eq:O_cap_V_final_conditioned}
\E[f(O\cap V_{\final})]&=\Pr[E_1\wedge E_2\wedge E_3]\times \E[f(O\cap V_{\final})\mid E_1\wedge E_2\wedge E_3]+\Pr[\overline{E_1\wedge E_2\wedge E_3}]\nonumber\\
&\quad\times \E[f(O\cap V_{\final})\mid \overline{E_1\wedge E_2\wedge E_3}]\nonumber\\
&\le\Pr[E_1\wedge E_2\wedge E_3]\times \E[f(O\cap V_{\final})\mid E_1\wedge E_2\wedge E_3]\nonumber\\
&\quad+\left(2\eps+\eta\right)\cdot f(O)\qquad\qquad\qquad\qquad\qquad\text{(By $f(O\cap V_{\final})\le f(O)$ and Ineq.~\eqref{eq:union_bound_E1_E2_E3})}\nonumber\\
&\le\E[f(O\cap V_{\final})\mid E_1\wedge E_2\wedge E_3]+\left(2\eps+\eta\right)\cdot f(O).
\end{align}
Combining Ineq.~\eqref{eq:O_cap_V_final_vs_O} with Ineq.~\eqref{eq:O_cap_V_final_conditioned}, we obtain that
\begin{align}\label{eq:O_cap_V_final_vs_O_conditioned}
\E[f(O\cap V_{\final})\mid E_1\wedge E_2\wedge E_3]&\ge\left(1-4\eps-\frac{2r}{n}-\eta\right)\cdot f(O)\nonumber\\
&\ge\left(1-5\eta\right)\cdot f(O),
\end{align}
where the last inequality is because $\textstyle2\eps$ and $\frac{r}{n}$ (both of which are at most $1$) are upper bounded by $\sqrt{\frac{2\eps n+2r}{n}}=\eta$. Then, by Lemma~\ref{lem:streaming_regular_case}, we have that
\begin{align}\label{eq:X_alg_conditioned_lower_bound}
\E[f(X_{\alg})\mid E_1\wedge E_2\wedge E_3]\ge&\left(\frac{1}{2}-\eps\right)\cdot(\E[f(O\cap V_{\final})\mid E_1\wedge E_2\wedge E_3]- 5\eta\cdot f(O))\nonumber\\
\ge&\left(\left(\frac{1}{2}-\eps\right)\cdot(1-5\eta-5\eta)\right)\cdot f(O)\quad\quad\text{(By Ineq.~\eqref{eq:O_cap_V_final_vs_O_conditioned})}\nonumber\\
\ge&\left(\frac{1}{2}-\eps-5\eta\right)\cdot f(O)\nonumber\\
\ge&\left(\frac{1}{2}-6\eta\right)\cdot f(O)\qquad\qquad\qquad\qquad\quad\,\,\,\text{(Since $\textstyle\eta=\sqrt{\frac{2\eps n+2r}{n}}\ge\eps$)}.
\end{align}
Finally, we derive that
\begin{align*}
\E[f(X_{\alg})]=&\Pr[E_1\wedge E_2\wedge E_3]\times\E[f(X_{\alg})\mid E_1\wedge E_2\wedge E_3]\\
&\quad+\Pr[\overline{E_1\wedge E_2\wedge E_3}]\times\E[f(X_{\alg})\mid \overline{E_1\wedge E_2\wedge E_3}]\\
\ge&\Pr[E_1\wedge E_2\wedge E_3]\times\E[f(X_{\alg})\mid E_1\wedge E_2\wedge E_3] &&\text{(Since $f(X_{\alg})\ge0$)}\\
\ge&\textstyle\left(1-2\eps-\eta\right)\cdot\left(\frac{1}{2}-6\eta\right)\cdot f(O) &&\text{(By Ineq.~\eqref{eq:union_bound_E1_E2_E3} and~\eqref{eq:X_alg_conditioned_lower_bound})}\\
\ge&\textstyle\left(1-3\eta\right)\cdot\left(\frac{1}{2}-6\eta\right)\cdot f(O) &&\text{(Since $\textstyle\eta=\sqrt{\frac{2\eps n+2r}{n}}\ge\eps$)}\\
\ge&\textstyle\left(\frac{1}{2}-8\eta\right)\cdot f(O),
\end{align*}
which finishes the proof of Theorem~\ref{thm:streaming}.
\end{proof}

\subsection{Proof of Lemma~\ref{lem:offline_break_condition}}\label{sec:proof_of_lem_offline_break_condition}
We start by introducing important concepts that will be used throughout the proof. Each vector $\Delta_i^t=\mathbf{1}_{S_i^t}-\mathbf{1}_{S_{i-1}^t}$ constructed in the second phase of Subroutine~\ref{sub:continuous_greedy_filtering} acts as a \emph{selector} in the third phase: for each element $e\in[n]$, Subroutine~\ref{sub:continuous_greedy_filtering} will include $e$ in set $H$ if it satisfies the selection condition posed by $\Delta_i^t$, i.e., if $S_{i-1}^t\cup\{e\}\in\M$ and $\smallfloor{F(\eps\cdot\mathbf{1}_e|x^{t-1}+\eps\cdot\mathbf{1}_{S_{i-1}^t})}_I>\smallfloor{F(\eps\cdot\Delta_{i}^t|x^{t-1}+\eps\cdot\mathbf{1}_{S_{i-1}^t})}_I$ (we say that an element $e\in[n]$ is \emph{selected} by the selector $\Delta_i^t$ if it satisfies this condition). For each $t\in\left[\frac{1}{\eps}\right]$ and $i\in[r]$, we let $G_i^t$ denote the set of elements in $[n]$ that are selected by $\Delta_i^t$, i.e.,
\begin{equation}\label{eq:G_i_t}
\textstyle G_i^t:=\{e\in[n]\mid S_{i-1}^t\cup\{e\}\in\M \textrm{ and } \smallfloor{F(\eps\cdot\mathbf{1}_e|x^{t-1}+\eps\cdot\mathbf{1}_{S_{i-1}^t})}_I>\smallfloor{F(\eps\cdot\Delta_{i}^t|x^{t-1}+\eps\cdot\mathbf{1}_{S_{i-1}^t})}_I\}.
\end{equation}

Within each epoch $t\in\left[\frac{1}{\eps}\right]$, we order the selectors according to their indices: $\Delta_1^t,\dots,\Delta_r^t$, and for each $i\in[r]$, we define
\begin{equation}\label{eq:H_i_t}\textstyle
H_i^t:=\bigcup_{j=1}^{i} G_j^t \textrm{ (and we let $H_0^t:=\emptyset$ for completeness)}.
\end{equation}
We say that a selector $\Delta_i^t$ is \emph{ineffective} if there is an earlier selector $\Delta_j^t$ for some $j<i$ in epoch $t$ such that $\smallfloor{F(\eps\cdot\Delta_{i}^t|x^{t-1}+\eps\cdot\mathbf{1}_{S_{i-1}^t})}_I\ge\smallfloor{F(\eps\cdot\Delta_{j}^t|x^{t-1}+\eps\cdot\mathbf{1}_{S_{j-1}^t})}_I$ (otherwise we call $\Delta_i^t$ an \emph{effective} selector). In Claim~\ref{claim:subroutine_ineffective_selector}, we show that, as the name suggests, any element that is selected by an ineffective selector $\Delta_i^t$ would have already been selected by an earlier selector $\Delta_j^t$ for some $j<i$, which implies that $H_i^t=H_{i-1}^t$.
\begin{claim}\label{claim:subroutine_ineffective_selector}
Suppose that $\Delta_i^t$ is an ineffective selector, i.e., there exists some $j<i$ such that $\smallfloor{F(\eps\Delta_{i}^t|x^{t-1}+\eps\mathbf{1}_{S_{i-1}^t})}_I\ge\smallfloor{F(\eps\Delta_{j}^t|x^{t-1}+\eps\mathbf{1}_{S_{j-1}^t})}_I$. Then, any element $e\in[n]$ that is selected by $\Delta_i^t$ is also selected by $\Delta_j^t$. In particular, this implies that $H_i^t=H_{i-1}^t$.
\end{claim}
\begin{proof}
We prove the first part of the claim by showing that $S_{i-1}^t\cup\{e\}\in\M$ implies $S_{j-1}^t\cup\{e\}\in\M$, and that $\smallfloor{F(\eps\mathbf{1}_e|x^{t-1}+\eps\mathbf{1}_{S_{i-1}^t})}_I>\smallfloor{F(\eps\Delta_{i}^t|x^{t-1}+\eps\mathbf{1}_{S_{i-1}^t})}_I$ implies $\smallfloor{F(\eps\mathbf{1}_e|x^{t-1}+\eps\mathbf{1}_{S_{j-1}^t})}_I>\smallfloor{F(\eps\Delta_{j}^t|x^{t-1}+\eps\mathbf{1}_{S_{j-1}^t})}_I$.

First, since $\M$ is a matroid and $S_{j-1}^t\cup\{e\}\subseteq S_{i-1}^t\cup\{e\}$, the condition $S_{i-1}^t\cup\{e\}\in\M$ implies the condition $S_{j-1}^t\cup\{e\}\in\M$.

Moreover, because $F$ is the multi-linear extension of submodular function $f$ and $S_{j-1}^t\subseteq S_{i-1}^t$, it holds that $F(\eps\mathbf{1}_e|x^{t-1}+\eps\mathbf{1}_{S_{j-1}^t})\ge F(\eps\mathbf{1}_e|x^{t-1}+\eps\mathbf{1}_{S_{i-1}^t})$, which, by monotonicity of the rounding operator $\floor{\cdot}_I$ (Definition~\ref{def:rounding_op}), implies that $\smallfloor{F(\eps\mathbf{1}_e|x^{t-1}+\eps\mathbf{1}_{S_{j-1}^t})}_I\ge\smallfloor{F(\eps\mathbf{1}_e|x^{t-1}+\eps\mathbf{1}_{S_{i-1}^t})}_I$. Thus, the condition $\smallfloor{F(\eps\mathbf{1}_e|x^{t-1}+\eps\mathbf{1}_{S_{i-1}^t})}_I>\smallfloor{F(\eps\Delta_{i}^t|x^{t-1}+\eps\mathbf{1}_{S_{i-1}^t})}_I$ implies $\smallfloor{F(\eps\mathbf{1}_e|x^{t-1}+\eps\mathbf{1}_{S_{j-1}^t})}_I>\smallfloor{F(\eps\Delta_{i}^t|x^{t-1}+\eps\mathbf{1}_{S_{i-1}^t})}_I$. Since we assume that $\smallfloor{F(\eps\Delta_{i}^t|x^{t-1}+\eps\mathbf{1}_{S_{i-1}^t})}_I\ge\smallfloor{F(\eps\Delta_{j}^t|x^{t-1}+\eps\mathbf{1}_{S_{j-1}^t})}_I$, it follows that $\smallfloor{F(\eps\mathbf{1}_e|x^{t-1}+\eps\mathbf{1}_{S_{j-1}^t})}_I>\smallfloor{F(\eps\Delta_{j}^t|x^{t-1}+\eps\mathbf{1}_{S_{j-1}^t})}_I$.

Now we prove the second part of the claim: $H_i^t=H_{i-1}^t$. Recall that every element $e\in G_i^t$ satisfies the selection condition of $\Delta_{i}^t$, which implies that element $e$ also satisfies the selection condition of $\Delta_{j}^t$ (by the first part of the claim). It follows that every element $e\in G_i^t$ belongs to $G_j^t$. Hence, we have that $G_i^t\subseteq G_j^t\subseteq H_{i-1}^t$, since $H_{i-1}^t=\bigcup_{\ell=1}^{i-1} G_{\ell}^t$ and $j\le i-1$. Finally, $H_i^t=H_{i-1}^t$ follows from $H_i^t=G_i^t\cup H_{i-1}^t$ and $G_i^t\subseteq H_{i-1}^t$.
\end{proof}

Next, we establish Lemma~\ref{lem:H_i_t_setminus_H_i-1_t} before proving Lemma~\ref{lem:offline_break_condition}. Lemma~\ref{lem:H_i_t_setminus_H_i-1_t} is an analogue of Lemma~\ref{lem:H_i_setminus_H_i-1} and shares a similar intuition and proof structure.
\begin{lemma}\label{lem:H_i_t_setminus_H_i-1_t}
For each $t\in\left[\frac{1}{\eps}\right]$ and $i\in[r]$, let $k_i^t$ denote the number of effective selectors among the first $i$ selectors $\Delta_{1}^t,\dots,\Delta_{i}^t$ in epoch $t$ of Subroutine~\ref{sub:continuous_greedy_filtering}, and let $k_0^t=0$. Then, for all $t\in\left[\frac{1}{\eps}\right]$ and $i\in[r]$, we have that
\begin{equation}\label{eq:H_i_t_setminus_H_i-1_t}
\Pr\left[|H_i^t| > k_i^t\cdot \frac{r\ln(r/\eps^2)}{\eps^3} \,\middle\vert\, \forall\,j\in\{0,\dots,i-1\},\,|H_j^t| \le k_j^t\cdot \frac{r\ln(r/\eps^2)}{\eps^3}\right]\le\frac{\eps^2}{r}.
\end{equation}
\end{lemma}
\begin{proof}[Proof of Lemma~\ref{lem:H_i_t_setminus_H_i-1_t}]
To simplify the proof language, for any element $e\in[n]$ and vectors $x,y\in[0,1]^n$, we will refer to the values $F(\eps\mathbf{1}_e|y)$ and $F(\eps x|y)$ as the ``$(\eps,y)$-marginal values'' of element $e$ and vector $x$, respectively. We prove the statement for an arbitrary epoch $t\in\left[\frac{1}{\eps}\right]$ of Subroutine~\ref{sub:continuous_greedy_filtering}.
\subsubsection*{Case 1: $i=1$}
We first prove the statement for epoch $t$ and $i=1$, in which case Ineq.~\eqref{eq:H_i_t_setminus_H_i-1_t} reduces to
\begin{equation}
\Pr\left[|H_1^t| > \frac{r\ln(r/\eps^2)}{\eps^3}\right]\le\frac{\eps^2}{r},
\end{equation}
because, by definition, we have that $H_0^t=\emptyset$, $k_0^t=0$ and $k_1^t=1$ (the first selector in epoch $t$ is always effective). Recall that $H_1^t=G_1^t$ by definition of $H_1^t$, and $$G_1^t=\{e\in[n]\mid \smallfloor{F(\eps\mathbf{1}_e|x^{t-1})}_I>\smallfloor{F(\eps\Delta_{1}^t|x^{t-1})}_I\},$$
by definition of $G_1^t$ in Eq.~\eqref{eq:G_i_t}. Thus, $H_1^t$ consists of elements in $[n]$ that have a strictly higher (rounded) $(\eps,x^{t-1})$-marginal value than $\Delta_1^t$. We observe that given how $\Delta_1^t$ is constructed in Subroutine~\ref{sub:continuous_greedy_filtering}, its (rounded) $(\eps,x^{t-1})$-marginal value is no less than that of any element in the subsampled set $V_1^t$ (note that this is trivially true if $V_1^t$ happens to be empty). Thus, it follows that elements in $H_1^t$ have a strictly higher (rounded) $(\eps,x^{t-1})$-marginal value than all elements in $V_1^t$. Therefore, $|H_1^t|>\frac{r\ln(r/\eps^2)}{\eps^3}$ holds only if none of the top $\ceil{\frac{r\ln(r/\eps^2)}{\eps^3}}$ elements in $[n]$ with the highest $(\eps,x^{t-1})$-marginal values appear in $V_1^t$. Because each of these top elements appears in $V_1^t$ independently with probability $\frac{\eps^3}{r}$, the probability that this happens is at most $(1-\frac{\eps^3}{r})^{\ceil{\frac{r\ln(r/\eps^2)}{\eps^3}}}\le\frac{\eps^2}{r}$. Hence, we have that $\Pr\big[|H_1^t|>\frac{r\ln(r/\eps^2)}{\eps^3}\big]\le\frac{\eps^2}{r}$.

\subsubsection*{Case 2: $i\in\{2,\dots,r\}$}
Now we prove the statement for epoch $t$ and $i\in\{2,\dots,r\}$. We let $T_{t,i}^{(1)}$ be 
the set of elements in $[n]$ that can be added to set $S_{i-1}^t$ without violating the matroid constraint, whose (rounded) $(\eps,x_{t-1}+\eps\mathbf{1}_{S_{i-1}^t})$-marginal values are no less than $\min_{j\in[i-1]}\smallfloor{F(\eps\Delta_{j}^t|x^{t-1}+\eps\mathbf{1}_{S_{j-1}^t})}_I$, i.e.,
\[\textstyle
    T_{t,i}^{(1)}:=\{e\in[n]\mid S_{i-1}^t\cup\{e\}\in\M \textrm{ and } \smallfloor{F(\eps\mathbf{1}_{e}|x^{t-1}+\eps\mathbf{1}_{S_{i-1}^t})}_I\ge\min\limits_{j\in[i-1]}\smallfloor{F(\eps\Delta_{j}^t|x^{t-1}+\eps\mathbf{1}_{S_{j-1}^t})}_I\}.
\]
Then, we let $T_{t,i}^{(2)}$ be the set of elements in $[n]$ that can be added to set $S_{i-1}^t$ without violating the matroid constraint, excluding $T_{t,i}^{(1)}$, i.e.,
\[\textstyle
T_{t,i}^{(2)}:=\{e\in [n]\mid S_{i-1}^t\cup\{e\}\in\M\}\setminus T_{t,i}^{(1)}.
\]
Moreover, we define two events $A_{t,i}^{(1)}$ and $A_{t,i}^{(2)}$ as follows:
\begin{itemize}
    \item Event $A_{t,i}^{(1)}$: None of the elements in $T_{t,i}^{(1)}$ appear in the subsampled set $V_i^t$.
    \item Event $A_{t,i}^{(2)}$: None of the top $\ceil{\frac{r\ln(r/\eps^2)}{\eps^3}}-|T_{t,i}^{(1)}|$ elements in $T_{t,i}^{(2)}$, with the highest $(\eps,x_{t-1}+\eps\mathbf{1}_{S_{i-1}^t})$-marginal values, appear in $V_i^t$. (If $\ceil{\frac{r\ln(r/\eps^2)}{\eps^3}}-|T_{t,i}^{(1)}|\le 0$ or $\ceil{\frac{r\ln(r/\eps^2)}{\eps^3}}-|T_{t,i}^{(1)}|>|T_{t,i}^{(2)}|$, we assume that this event holds trivially.)
\end{itemize}
Next, we show that conditioned on $|H_{i-1}^t| \le k_{i-1}^t\cdot \frac{r\ln(r/\eps^2)}{\eps^3}$, the event $|H_i^t| > k_i^t\cdot \frac{r\ln(r/\eps^2)}{\eps^3}$ implies events $A_{t,i}^{(1)}$ and $A_{t,i}^{(2)}$.

\subsubsection*{Step 1: Conditioned on $|H_{i-1}^t| \le k_{i-1}^t\cdot \frac{r\ln(r/\eps^2)}{\eps^3}$, the event $|H_i^t| > k_i^t\cdot \frac{r\ln(r/\eps^2)}{\eps^3}$ implies $A_{t,i}^{(1)}$}
We first prove that conditioned on $|H_{i-1}^t| \le k_{i-1}^t\cdot \frac{r\ln(r/\eps^2)}{\eps^3}$, if event $A_{t,i}^{(1)}$ does not occur, then the event $|H_i^t| > k_i^t\cdot \frac{r\ln(r/\eps^2)}{\eps^3}$ cannot occur either. Specifically, if $A_{t,i}^{(1)}$ does not occur, then there is an element $e_1\in T_{t,i}^{(1)}$ that appears in the subsampled set $V_i^t$. Given how $\Delta_i^t$ is constructed in Subroutine~\ref{sub:continuous_greedy_filtering}, $\Delta_i^t$ must satisfy that $\smallfloor{F(\eps\Delta_i^t|x^{t-1}+\eps\mathbf{1}_{S_{i-1}^t})}_I\ge\smallfloor{F(\eps\mathbf{1}_{e_1}|x^{t-1}+\eps\mathbf{1}_{S_{i-1}^t})}_I$. Moreover, since $e_1\in T_{t,i}^{(1)}$, it follows by definition of $T_{t,i}^{(1)}$ that $\smallfloor{F(\eps\mathbf{1}_{e_1}|x^{t-1}+\eps\mathbf{1}_{S_{i-1}^t})}_I\ge\min_{j\in[i-1]}\smallfloor{F(\eps\Delta_{j}^t|x^{t-1}+\eps\mathbf{1}_{S_{j-1}^t})}_I$. Therefore, we have that $\smallfloor{F(\eps\Delta_i^t|x^{t-1}+\eps\mathbf{1}_{S_{i-1}^t})}_I\ge\min_{j\in[i-1]}\smallfloor{F(\eps\Delta_{j}^t|x^{t-1}+\eps\mathbf{1}_{S_{j-1}^t})}_I$, which implies that $\Delta_i^t$ is an ineffective selector.
Because $\Delta_i^t$ is an ineffective selector, we have that $k_i^t=k_{i-1}^t$ by definition of $k_i^t$ and $k_{i-1}^t$, and that $H_i^t=H_{i-1}^t$ by Claim~\ref{claim:subroutine_ineffective_selector}. Conditioned on $|H_{i-1}^t| \le k_{i-1}^t\cdot \frac{r\ln(r/\eps^2)}{\eps^3}$, this implies that $|H_i^t| \le k_i^t\cdot \frac{r\ln(r/\eps^2)}{\eps^3}$, and hence, the event $|H_i^t| > k_i^t\cdot \frac{r\ln(r/\eps^2)}{\eps^3}$ cannot occur.

\subsubsection*{Step 2: Conditioned on $|H_{i-1}^t| \le k_{i-1}^t\cdot \frac{r\ln(r/\eps^2)}{\eps^3}$, the event $|H_i^t| > k_i^t\cdot \frac{r\ln(r/\eps^2)}{\eps^3}$ implies $A_{t,i}^{(2)}$}
Now we show that conditioned on $|H_{i-1}^t| \le k_{i-1}^t\cdot \frac{r\ln(r/\eps^2)}{\eps^3}$, if the event $|H_i^t| > k_i^t\cdot \frac{r\ln(r/\eps^2)}{\eps^3}$ occurs, then event $A_{t,i}^{(2)}$ also occurs (we assume w.l.o.g.~that $\ceil{\frac{r\ln(r/\eps^2)}{\eps^3}}-|T_{t,i}^{(1)}|\ge1$, since otherwise $A_{t,i}^{(2)}$ holds trivially by definition). Specifically, conditioned on $|H_{i-1}^t| \le k_{i-1}^t\cdot \frac{r\ln(r/\eps^2)}{\eps^3}$, if the event $|H_i^t| > k_i^t\cdot \frac{r\ln(r/\eps^2)}{\eps^3}$ occurs, then it must hold that $k_i^t=k_{i-1}^t+1$ (because otherwise $k_i^t=k_{i-1}^t$, and $\Delta_i^t$ would be an ineffective selector, which would imply that $|H_i^t|=|H_{i-1}^t|\le k_{i-1}^t\cdot \frac{r\ln(r/\eps^2)}{\eps^3}=k_i^t\cdot \frac{r\ln(r/\eps^2)}{\eps^3}$ by Claim~\ref{claim:subroutine_ineffective_selector}), and hence, $|H_i^t\setminus H_{i-1}^t|=|H_i^t|-|H_{i-1}^t|>\frac{r\ln(r/\eps^2)}{\eps^3}$. By definition of $H_i^t$ and $H_{i-1}^t$ in Eq.~\eqref{eq:H_i}, we have that $H_i^t\setminus H_{i-1}^t=G_i^t\setminus(\bigcup_{j=1}^{i-1} G_j^t)$, which implies that $|G_i^t|\ge|H_i^t\setminus H_{i-1}^t|>\frac{r\ln(r/\eps^2)}{\eps^3}$. This in turn implies that $|G_i^t\setminus T_{t,i}^{(1)}|\ge|G_i^t|-|T_{t,i}^{(1)}|>\frac{r\ln(r/\eps^2)}{\eps^3}-|T_{t,i}^{(1)}|$.

Note that by definition of $G_i^t$ (Eq.~\eqref{eq:G_i_t}) and $T_{t,i}^{(1)},T_{t,i}^{(2)}$, the set $G_i^t\setminus T_{t,i}^{(1)}$ is a subset of elements in $T_{t,i}^{(2)}$ that have higher (rounded) $(\eps,x_{t-1}+\eps\mathbf{1}_{S_{i-1}^t})$-marginal values than $\Delta_i^t$. Since $|G_i^t\setminus T_{t,i}^{(1)}|>\frac{r\ln(r/\eps^2)}{\eps^3}-|T_{t,i}^{(1)}|$, there must be \emph{strictly} more than $\frac{r\ln(r/\eps^2)}{\eps^3}-|T_{t,i}^{(1)}|$ elements in $T_{t,i}^{(2)}$ that have higher (rounded) $(\eps,x_{t-1}+\eps\mathbf{1}_{S_{i-1}^t})$-marginal values than $\Delta_i^t$. Next, we derive a contradiction to this, assuming event $A_{t,i}^{(2)}$ does not occur.

Suppose for contradiction that $A_{t,i}^{(2)}$ does not occur. That is, one of the top $\ceil{\frac{r\ln(r/\eps^2)}{\eps^3}}-|T_{t,i}^{(1)}|$ elements in $T_{t,i}^{(2)}$ with the highest $(\eps,x_{t-1}+\eps\mathbf{1}_{S_{i-1}^t})$-marginal values, which we denote by $e_2\in T_{t,i}^{(2)}$, appears in the subsampled set $V_i^t$. Given how $\Delta_i^t$ is constructed in Subroutine~\ref{sub:continuous_greedy_filtering}, $\Delta_i^t$ must have at least the same (rounded) $(\eps,x_{t-1}+\eps\mathbf{1}_{S_{i-1}^t})$-marginal value as element $e_2$. Hence, the (rounded) marginal value of $\Delta_i^t$ is no less than that of any element in $T_{t,i}^{(2)}$, except for the other top $\ceil{\frac{r\ln(r/\eps^2)}{\eps^3}}-|T_{t,i}^{(1)}|-1$ elements (excluding $e_2$), which is a contradiction.

\subsubsection*{Step 3: Proving Ineq.~\eqref{eq:H_i_t_setminus_H_i-1_t} under the additional condition $\ceil{\frac{r\ln(r/\eps^2)}{\eps^3}}>|T_{t,i}^{(1)}|+|T_{t,i}^{(2)}|$}
Next, we prove that Ineq.~\eqref{eq:H_i_t_setminus_H_i-1_t} holds trivially if we additionally condition on the event $\ceil{\frac{r\ln(r/\eps^2)}{\eps^3}}>|T_{t,i}^{(1)}|+|T_{t,i}^{(2)}|$. Specifically, we show that
\begin{equation}\label{eq:H_i_t_setminus_H_i-1_t_not_enough_feasible_elements}
\Pr\left[|H_i^t| \le k_i^t\cdot \frac{r\ln(r/\eps^2)}{\eps^3} \,\middle\vert\, |H_{i-1}^t| \le k_{i-1}^t\cdot \frac{r\ln(r/\eps^2)}{\eps^3},\, \ceil{\frac{r\ln(r/\eps^2)}{\eps^3}}>|T_{t,i}^{(1)}|+|T_{t,i}^{(2)}|\right]=1.
\end{equation}

To this end, we notice that conditioned on $\ceil{\frac{r\ln(r/\eps^2)}{\eps^3}}>|T_{t,i}^{(1)}|+|T_{t,i}^{(2)}|$, the set $T_{t,i}^{(1)}\cup T_{t,i}^{(2)}$, which consists of all elements in $[n]$ that can be added to $S_{i-1}^t$ without violating the matroid constraint, has size at most $\ceil{\frac{r\ln(r/\eps^2)}{\eps^3}}-1\le\frac{r\ln(r/\eps^2)}{\eps^3}$. By definition of $G_i^t$ in Eq.~\eqref{eq:G_i_t}, we have that $G_i^t\subseteq T_{t,i}^{(1)}\cup T_{t,i}^{(2)}$, and hence, $|G_i^t|\le\frac{r\ln(r/\eps^2)}{\eps^3}$. This implies that $|H_i^t\setminus H_{i-1}^t|\le\frac{r\ln(r/\eps^2)}{\eps^3}$ since $H_i^t\setminus H_{i-1}^t\subseteq G_i^t$. Conditioned on $|H_{i-1}^t| \le k_{i-1}^t\cdot \frac{r\ln(r/\eps^2)}{\eps^3}$, it follows that $|H_i^t|=|H_i^t\setminus H_{i-1}^t|+|H_{i-1}^t|\le (k_{i-1}^t+1)\cdot \frac{r\ln(r/\eps^2)}{\eps^3}$, which implies the event $|H_i^t|\le k_i^t\cdot \frac{r\ln(r/\eps^2)}{\eps^3}$ if $k_i^t=k_{i-1}^t+1$. If instead $k_i^t=k_{i-1}^t$, then $\Delta_i^t$ is an ineffective selector, and hence, by Claim~\ref{claim:subroutine_ineffective_selector}, we have that $|H_i^t|=|H_{i-1}^t|$, which also implies the event $|H_i^t|\le k_i^t\cdot \frac{r\ln(r/\eps^2)}{\eps^3}$, conditioned on $|H_{i-1}^t| \le k_{i-1}^t\cdot \frac{r\ln(r/\eps^2)}{\eps^3}$. Thus, Ineq.~\eqref{eq:H_i_t_setminus_H_i-1_t_not_enough_feasible_elements} follows.

\subsubsection*{Step 4: Proving Ineq.~\eqref{eq:H_i_t_setminus_H_i-1_t} under the additional condition $\ceil{\frac{r\ln(r/\eps^2)}{\eps^3}}\le|T_{t,i}^{(1)}|+|T_{t,i}^{(2)}|$}
Finally, we prove that Ineq.~\eqref{eq:H_i_t_setminus_H_i-1_t} also holds if we additionally condition on $\ceil{\frac{r\ln(r/\eps^2)}{\eps^3}}\le|T_{t,i}^{(1)}|+|T_{t,i}^{(2)}|$.
Since we have shown that conditioned on $|H_{i-1}^t| \le k_{i-1}^t\cdot \frac{r\ln(r/\eps^2)}{\eps^3}$, the event $|H_i^t| > k_i^t\cdot \frac{r\ln(r/\eps^2)}{\eps^3}$ implies events $A_{t,i}^{(1)}$ and $A_{t,i}^{(2)}$, it suffices to show that $$\textstyle\Pr\left[A_{t,i}^{(1)}\wedge A_{t,i}^{(2)}\,\middle\vert\, \left(\forall\,j\in\{0,\dots,i-1\},\,|H_j^t|\le k_j^t\cdot \frac{r\ln(r/\eps^2)}{\eps^3}\right),\,\ceil{\frac{r\ln(r/\eps^2)}{\eps^3}}\le|T_{t,i}^{(1)}|+|T_{t,i}^{(2)}|\right]\le\frac{\eps^2}{r}.$$
Moreover, observe that the numbers $k_1^t,\dots,k_{i-1}^t$, the sets $H_1^t,\dots,H_{i-1}^t$, $T_{t,i}^{(1)}$ and $T_{t,i}^{(2)}$, and the events $A_{t,i}^{(1)}$ and $A_{t,i}^{(2)}$ are fully determined by the subsampled sets $V_1^t,\dots,V_{i-1}^t$ and the fractional solution $x^{t-1}$ constructed in epoch $t-1$ in Subroutine~\ref{sub:continuous_greedy_filtering}. Therefore, it suffices to prove that $\Pr[A_{t,i}^{(1)}\wedge A_{t,i}^{(2)} \mid V_1^t,\dots,V_{i-1}^t,x^{t-1}]\le\frac{\eps^2}{r}$, for any fixed $V_1^t,\dots,V_{i-1}^t,x^{t-1}$ that are consistent with the condition $\ceil{\frac{r\ln(r/\eps^2)}{\eps^3}}\le|T_{t,i}^{(1)}|+|T_{t,i}^{(2)}|$.

To this end, we consider any $V_1^t,\dots,V_{i-1}^t,x^{t-1}$ that satisfy the condition $\ceil{\frac{r\ln(r/\eps^2)}{\eps^3}}-|T_{t,i}^{(1)}|\le|T_{t,i}^{(2)}|$ (which is equivalent to $\ceil{\frac{r\ln(r/\eps^2)}{\eps^3}}\le|T_{t,i}^{(1)}|+|T_{t,i}^{(2)}|$). Since each element of $[n]$ appears in the $i$-th subsampled set $V_i^t$ in epoch $t$ independently with probability $\frac{\eps^3}{r}$, the probability that none of the elements in $T_{t,i}^{(1)}$ appear in $V_i^t$ is at most $(1-\frac{\eps^3}{r})^{|T_{t,i}^{(1)}|}$, namely,
\begin{equation}\label{eq:A_i_t_1_probability}
\Pr\left[A_{t,i}^{(1)} \,\middle\vert\, V_1^t,\dots,V_{i-1}^t,x^{t-1}\right]\le\left(1-\frac{\eps^3}{r}\right)^{|T_{t,i}^{(1)}|}.
\end{equation}
If $V_1^t,\dots,V_{i-1}^t,x^{t-1}$ further satisfy the condition $\ceil{\frac{r\ln(r/\eps^2)}{\eps^3}}-|T_{t,i}^{(1)}|\ge1$, then it follows that $1\le\ceil{\frac{r\ln(r/\eps^2)}{\eps^3}}-|T_{t,i}^{(1)}|\le|T_{t,i}^{(2)}|$, and hence, the top $\ceil{\frac{r\ln(r/\eps^2)}{\eps^3}}-|T_{t,i}^{(1)}|$ elements in $T_{t,i}^{(2)}$, with the highest $(\eps,x^{t-1}+\eps\mathbf{1}_{S_{i-1}^t})$-marginal values, are well-defined. The probability that none of these top $\ceil{\frac{r\ln(r/\eps^2)}{\eps^3}}-|T_{t,i}^{(1)}|$ elements of $T_{t,i}^{(2)}$ appear in $V_i^t$ is at most $(1-\frac{\eps^3}{r})^{\ceil{\frac{r\ln(r/\eps^2)}{\eps^3}}-|T_{t,i}^{(1)}|}$, and this probability can only decrease if we additionally condition on event $A_{t,i}^{(1)}$. Hence, when $V_1^t,\dots,V_{i-1}^t,x^{t-1}$ further satisfy  $\ceil{\frac{r\ln(r/\eps^2)}{\eps^3}}-|T_{t,i}^{(1)}|\ge1$, we have that 
\begin{equation}\label{eq:A_i_t_2_probability}
\Pr\left[A_{t,i}^{(2)} \,\middle\vert\, V_1^t,\dots,V_{i-1}^t,x^{t-1} \textrm{ and } A_{t,i}^{(1)}\right]\le\left(1-\frac{\eps^3}{r}\right)^{\ceil{\frac{r\ln(r/\eps^2)}{\eps^3}}-|T_{t,i}^{(1)}|}.
\end{equation}
On the other hand, if $V_1^t,\dots,V_{i-1}^t,x^{t-1}$ do not satisfy $\ceil{\frac{r\ln(r/\eps^2)}{\eps^3}}-|T_{t,i}^{(1)}|\ge1$, which implies that $\ceil{\frac{r\ln(r/\eps^2)}{\eps^3}}-|T_{t,i}^{(1)}|\le0$, then Ineq.~\eqref{eq:A_i_t_2_probability} holds trivially. Thus, Ineq.~\eqref{eq:A_i_t_2_probability} holds regardless of whether $V_1^t,\dots,V_{i-1}^t,x^{t-1}$ satisfy $\ceil{\frac{r\ln(r/\eps^2)}{\eps^3}}-|T_{t,i}^{(1)}|\ge1$. Finally, we derive that
\begin{align*}
    &\Pr\left[A_{t,i}^{(1)}\wedge A_{t,i}^{(2)} \,\middle\vert\, V_1^t,\dots,V_{i-1}^t,x^{t-1}\right]\\
    =&\Pr\left[A_{t,i}^{(1)}\,\middle\vert\, V_1^t,\dots,V_{i-1}^t,x^{t-1}\right]\times
    \Pr\left[A_{t,i}^{(2)}\,\middle\vert\, V_1^t,\dots,V_{i-1}^t,x^{t-1} \textrm{ and } A_{t,i}^{(1)}\right]\\
    \le& \left(1-\frac{\eps^3}{r}\right)^{|T_{t,i}^{(1)}|}\times
    \Pr\left[A_{t,i}^{(2)}\,\middle\vert\, V_1^t,\dots,V_{i-1}^t,x^{t-1} \textrm{ and } A_{t,i}^{(1)}\right] &&\text{(By Ineq.~\eqref{eq:A_i_t_1_probability})}\\
    \le& \left(1-\frac{\eps^3}{r}\right)^{|T_{t,i}^{(1)}|}\times\left(1-\frac{\eps^3}{r}\right)^{\ceil{\frac{r\ln(r/\eps^2)}{\eps^3}}-|T_{t,i}^{(1)}|} &&\text{(By Ineq.~\eqref{eq:A_i_t_2_probability})}\\
    =& \left(1-\frac{\eps^3}{r}\right)^{\ceil{\frac{r\ln(r/\eps^2)}{\eps^3}}}\le\frac{\eps^2}{r},
\end{align*}
which finishes the proof of Lemma~\ref{lem:H_i_t_setminus_H_i-1_t}.
\end{proof}

We are ready to complete the proof of Lemma~\ref{lem:offline_break_condition}.
\begin{proof}[Proof of Lemma~\ref{lem:offline_break_condition}]
First, following the notations in Lemma~\ref{lem:H_i_t_setminus_H_i-1_t}, we upper bound the probability $\Pr\left[|H_r^t|\le k_r^t\cdot\frac{r\ln(r/\eps^2)}{\eps^3}\right]$ for each $t\in\left[\frac{1}{\eps}\right]$ as follows,
\begin{align*}
    &\Pr\left[|H_r^t|\le k_r^t\cdot\frac{r\ln(r/\eps^2)}{\eps^3}\right]\\
    \ge& \Pr\left[\forall\,i\in[r],\,|H_i^t| \le k_i^t\cdot\frac{r\ln(r/\eps^2)}{\eps^3}\right]\\
    =&\prod_{i\in[r]}\Pr\left[|H_i^t| \le k_i^t\cdot \frac{r\ln(r/\eps^2)}{\eps^3} \,\middle\vert\, \forall\,j\in\{0,\dots,i-1\},\,|H_j^t| \le k_j^t\cdot \frac{r\ln(r/\eps^2)}{\eps^3}\right]\\
    \ge&\left(1-\frac{\eps^2}{r}\right)^r \qquad\qquad\qquad\qquad\qquad\qquad\qquad\qquad\qquad\qquad\qquad\qquad\qquad\qquad\text{(By Lemma~\ref{lem:H_i_t_setminus_H_i-1_t})}\\
    \ge&\,1-\eps^2.
\end{align*}
Then, by a union bound, we have that
\begin{equation}\label{eq:Pr_H_r_t_le_k_r_t}
\Pr\left[\forall\,t\in\left[\frac{1}{\eps}\right],\,|H_r^t|\le k_r^t\cdot\frac{r\ln(r/\eps^2)}{\eps^3}\right]\ge1-\eps.
\end{equation}
Moreover, we notice that the total number of effective selectors $k_r^t$ in each epoch is at most $|I|$, because any two distinct effective selectors $\Delta_i^t$ and $\Delta_j^t$ must satisfy that $\floor{F(\eps\Delta_i^t|x^{t-1}+\eps\mathbf{1}_{S_{i-1}^t)}}_I\neq\floor{F(\eps\Delta_j^t|x^{t-1}+\eps\mathbf{1}_{S_{j-1}^t)}}_I$ (and there are only $|I|$ distinct values in the range of $\floor{\cdot}_I$). Thus, Ineq.~\eqref{eq:Pr_H_r_t_le_k_r_t} implies $\Pr\left[\forall\,t\in\left[\frac{1}{\eps}\right],\,|H_r^t|\le \frac{r\ln(r/\eps^2)\cdot |I|}{\eps^3}\right]\ge1-\eps$. The proof finishes by observing that $$\Pr\Bigg[\sum_{t\in\left[\frac{1}{\eps}\right]}|H_r^t|\le \frac{r\ln(r/\eps^2)\cdot |I|}{\eps^4}\Bigg]\ge\Pr\left[\forall\,t\in\left[\frac{1}{\eps}\right],\,|H_r^t|\le \frac{r\ln(r/\eps^2)\cdot |I|}{\eps^3}\right],$$ and that the set $H$ in Subroutine~\ref{sub:continuous_greedy_filtering} is a subset of $\bigcup_{t\in\left[\frac{1}{\eps}\right]}H_r^t$.
\end{proof}

\subsection{Proof of Lemma~\ref{lem:offline_regular_case}}\label{sec:proof_of_lem_offline_regular_case}
\begin{proof}[Proof of Lemma~\ref{lem:offline_regular_case}]
\subsubsection*{The proof setup}
We start by introducing the notations that will be used throughout the analysis and deriving their relations using basic structural properties of matroids. For each epoch $t\in\left[\frac{1}{\eps}\right]$ in the second phase of Subroutine~\ref{sub:continuous_greedy_filtering}, we consider the restriction $\M_{S_r^t\cup O}$ of matroid $\M$ to set $S_r^t\cup O$. Because the restriction $\M_{S_r^t\cup O}$ is a matroid, we can augment $S_r^t\in\M_{S_r^t\cup O}$ with some subset $O_S^t\subseteq O$ such that $S_r^t\cupdot O_S^t$ is a base of $\M_{S_r^t\cup O}$, and similarly, we can augment $O\in\M_{S_r^t\cup O}$ with some subset $S_O^t\subseteq S_r^t$ such that $O\cupdot S_O^t$ is a base of $\M_{S_r^t\cup O}$.

Then, we partition the optimal solution $O$ into two disjoint subsets: $O_L:=O\setminus H$ and $O_H:=O\cap H$. That is, $O_H$ is the subset of elements from $O$ that are included in $H$ by Subroutine~\ref{sub:continuous_greedy_filtering}, and $O_L$ is the subset of elements from $O$ that are not included in $H$ (importantly, conditioned on event $E_1$ defined in Lemma~\ref{lem:offline_break_condition}, these elements were not added to $H$ because they did not satisfy the filter condition at Line~\ref{algline:offline_filter_condition} of Subroutine~\ref{sub:continuous_greedy_filtering}, not because of the break condition at Line~\ref{algline:offline_break_condition}). Moreover, for each $t\in\left[\frac{1}{\eps}\right]$, we denote $O_L^t:=O_L\setminus O_S^t$ and $O_H^t:=O_H\setminus O_S^t$. Hence, the optimal solution $O$ can be decomposed as $O=O_S^t\cupdot O_L^t\cupdot O_H^t$, and $O_L$ can be decomposed as
\begin{equation}\label{eq:O_L_decomposition}
O_L=(O_S^t\cap O_L)\cupdot O_L^t.
\end{equation}

Now we consider the contraction $\M_{S_r^t\cup O}/(O_S^t\cupdot S_O^t)$ of matroid $\M_{S_r^t\cup O}$ for each $t\in\left[\frac{1}{\eps}\right]$. Notice that both $O_L^t\cupdot O_H^t$ and $S_r^t\setminus S_O^t$ are bases of matroid $\M_{S_r^t\cup O}/(O_S^t\cupdot S_O^t)$. Hence, by Lemma~\ref{lem:base_exchange}, there exists a partition $S_r^t\setminus S_O^t=S_L^t\cupdot S_H^t$ such that both $S_L^t\cupdot O_H^t$ and $O_L^t\cupdot S_H^t$ are bases of $\M_{S_r^t\cup O}/(O_S^t\cupdot S_O^t)$, which implies that both $O_S^t\cupdot S_L^t\cupdot O_H^t\cupdot S_O^t$ and $O_S^t\cupdot O_L^t\cupdot S_H^t\cupdot S_O^t$ are bases of $\M_{S_r^t\cup O}$.

In particular, since both $O_S^t\cupdot O_L^t\cupdot S_H^t\cupdot S_O^t$ and $O_S^t\cupdot S_L^t\cupdot S_H^t\cupdot S_O^t=O_S^t\cupdot S_r^t$ are bases of $\M_{S_r^t\cup O}$, by Lemma~\ref{lem:base_matching}, there exists a bijection $h_t:S_L^t\to O_L^t$ such that for all $s\in S_L^t$,
\[
\textrm{$((O_S^t\cupdot S_L^t\cupdot S_H^t\cupdot S_O^t)\setminus\{s\})\cup\{h_t(s)\}$ is a base of $\M_{S_r^t\cup O}$.}
\]
Since $S_L^t\cupdot S_H^t\cupdot S_O^t=S_r^t$, it follows that $(S_r^t\setminus\{s\})\cup\{h_t(s)\}\in\M_{S_r^t\cup O}$ for all $s\in S_L^t$, which implies that for all $s\in S_L^t$,
\begin{equation}\label{eq:h_t_maps_S_L_t_to_O_L_t}
(S_r^t\setminus\{s\})\cup\{h_t(s)\}\in\M.
\end{equation}
(To help the reader keep track of the notations, we show the relations between $O_S^t\cupdot O_L^t\cupdot O_H^t\cupdot S_O^t$ and $O_S^t\cupdot S_L^t\cupdot S_H^t\cupdot S_O^t$, and between $O_S^t\cupdot S_L^t\cupdot S_H^t\cupdot S_O^t$ and $O_S^t\cupdot O_L^t\cupdot S_H^t\cupdot S_O^t$ in Figure~\ref{fig:illustration_offline}.)

\begin{figure}[ht]
    \centering
    \includegraphics[scale=0.3]{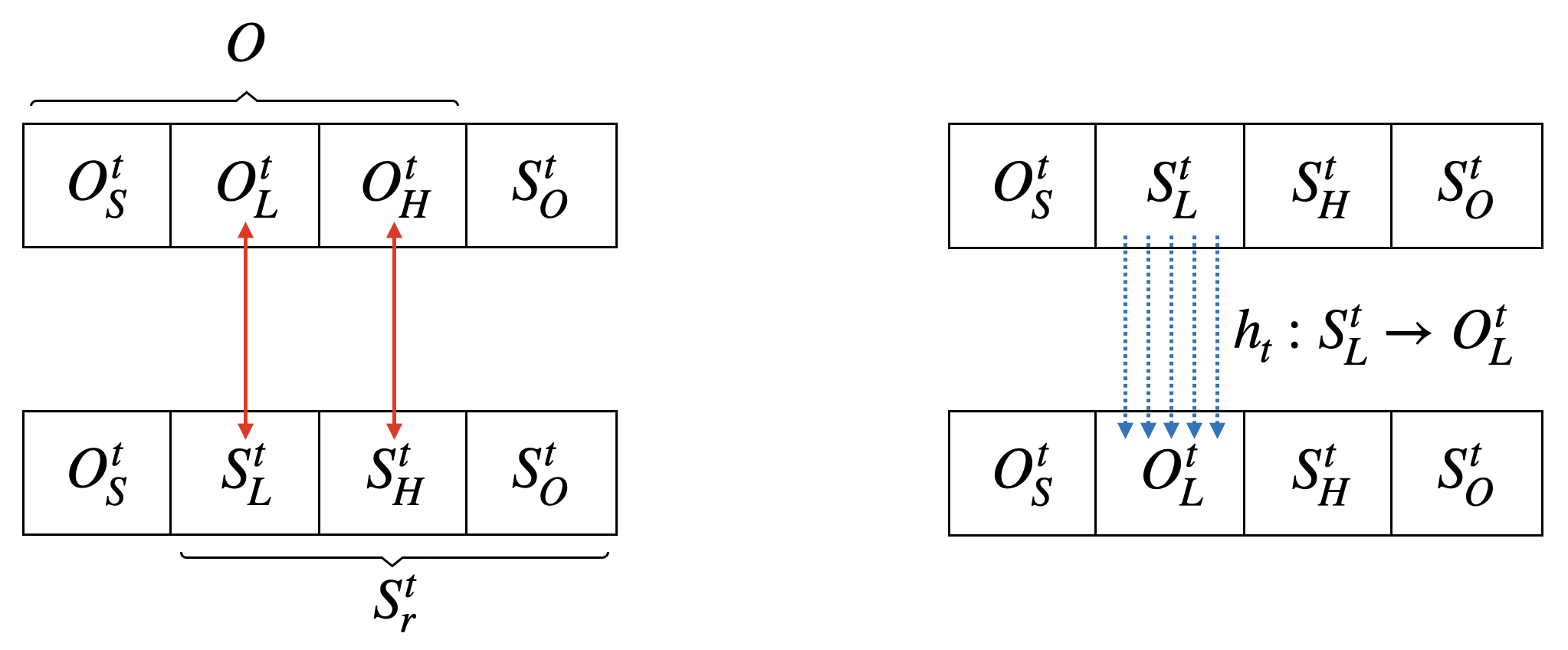}
    \caption{On the left, both $O_S^t\cupdot O_L^t\cupdot O_H^t\cupdot S_O^t$ and $O_S^t\cupdot S_L^t\cupdot S_H^t\cupdot S_O^t$ are bases of matroid $\M_{S_r^t\cup O}$, and we can exchange $O_L^t$ with $S_L^t$, or $O_H^t$ with $S_H^t$, resulting in new bases. On the right, both $O_S^t\cupdot S_L^t\cupdot S_H^t\cupdot S_O^t$ and $O_S^t\cupdot O_L^t\cupdot S_H^t\cupdot S_O^t$ are bases of matroid $\M_{S_r^t\cup O}$, and there is a bijection $h_t:S_L^t\to O_L^t$ such that for all $s\in S_L^t$, replacing element $s$ in $O_S^t\cupdot S_L^t\cupdot S_H^t\cupdot S_O^t$ with element $h_t(s)$ yields a new base. Finally, we remind the reader that $O=O_S^t\cupdot O_L^t\cupdot O_H^t$, $S_r^t=S_L^t\cupdot S_H^t\cupdot S_O^t$, and that $O_L=(O_S^t\cap O_L)\cupdot O_L^t$, as stated in Eq.~\eqref{eq:O_L_decomposition}.}
    \label{fig:illustration_offline}
\end{figure}

Next, for arbitrary $t\in\left[\frac{1}{\eps}\right]$, we upper bound $\E[f(O_S^t\cap O_L|\cR(x^t))]$ and $\E[f(O_L^t|\cR(x^t))]$ (recall that $x^t=x^{t-1}+\eps\mathbf{1}_{S_r^t}$ is the final fractional solution constructed in epoch $t$ in the second phase of Subroutine~\ref{sub:continuous_greedy_filtering}).

\subsubsection*{Step 1: Upper bounding $\E[f(O_S^t\cap O_L|\cR(x^t))]$}
We consider two cases $O_S^t\cap O_L\neq\emptyset$ and $O_S^t\cap O_L=\emptyset$ and prove Ineq.~\eqref{eq:O_S_t_cap_O_L_t} for both cases.
\begin{equation}\label{eq:O_S_t_cap_O_L_t}
\eps\cdot\E[f(O_S^t\cap O_L|\cR(x^t))]\le\frac{\eps^2 v}{1-\eps}.
\end{equation}

\paragraph{Case 1: $O_S^t\cap O_L\neq\emptyset$.} Because $S_r^t\cupdot O_S^t$ is a base of $\M_{S_r^t\cup O}$, we have that $S_r^t\cupdot O_S^t\in\M$. Because $O_S^t$ is non-empty and $\M$ has rank $r$, it follows that $|S_r^t|<r$. Given how $S_r^t$ is constructed in Subroutine~\ref{sub:continuous_greedy_filtering}, if $|S_r^t|<r$, then there must exist $i\in[r]$ for which the subroutine sets $\Delta_i^t$ to $\mathbf{0}$. It follows that $F(\eps\Delta_i^t|x^{t-1}+\eps\mathbf{1}_{S_{i-1}^t})=0$.

Moreover, for any element $o\in O_S^t\cap O_L$, since $S_r^t\cupdot O_S^t\in\M$, element $o$ can be added to $S_r^t$ without violating the matroid constraint, which implies that $S_{i-1}^t\cup\{o\}\in\M$. However, since $o\in O_L$ (and hence $o\notin H$), element $o$ did not satisfy the selection condition posed by $\Delta_i^t$ at Line~\ref{algline:offline_filter_condition} of Subroutine~\ref{sub:continuous_greedy_filtering}, despite $S_{i-1}^t\cup\{o\}\in\M$. Therefore, it must hold that $\smallfloor{F(\eps\mathbf{1}_o|x^{t-1}+\eps\mathbf{1}_{S_{i-1}^t})}_I
\le\smallfloor{F(\eps\Delta_i^t|x^{t-1}+\eps\mathbf{1}_{S_{i-1}^t})}_I=\smallfloor{0}_I$. Because $F(\eps\mathbf{1}_o|x^{t-1}+\eps\mathbf{1}_{S_{i-1}^t})\le F(\eps\mathbf{1}_o)$ (by submodularity) and $F(\eps\mathbf{1}_o)=\eps\cdot f(\{o\})\le \eps v$ (recall that $v$ is the value of the most valuable element in $[n]$, which is computed in the first phase of Subroutine~\ref{sub:continuous_greedy_filtering}), we have that $F(\eps\mathbf{1}_o|x^{t-1}+\eps\mathbf{1}_{S_{i-1}^t})\le\eps v$. Thus, we can apply Claim~\ref{claim:streaming_rounding_error} (by setting $a$, $b$, $\gamma$, $w$ to be $0$, $F(\eps\mathbf{1}_o|x^{t-1}+\eps\mathbf{1}_{S_{i-1}^t})$, $\frac{\eps}{r}$, $\eps v$, respectively), and we obtain that 
\begin{equation}\label{eq:o_in_O_S_t}
F(\eps\mathbf{1}_o|x^{t-1}+\eps\mathbf{1}_{S_{i-1}^t})\le\frac{\eps^2 v}{(1-\eps)r}.
\end{equation}
Then, we derive that
\begin{align*}
\eps\cdot\E[f(O_S^t\cap O_L|\cR(x^t))]&\le\sum_{o\in O_S^t\cap O_L}\eps\cdot \E[f(\{o\}|\cR(x^t))] &&\text{(By submodularity)}\nonumber\\
&\le\sum_{o\in O_S^t\cap O_L}\eps\cdot \E[f(\{o\}|\cR(x^t)\setminus\{o\})] &&\text{(By submodularity)}\nonumber\\
&=\sum_{o\in O_S^t\cap O_L}F(\eps\mathbf{1}_{o}|x^t) &&\text{(By Definition~\ref{def:multi-linear})}\nonumber\\
&=\sum_{o\in O_S^t\cap O_L}F(\eps\mathbf{1}_{o}|x^{t-1}+\eps\mathbf{1}_{S_r^t}) &&\text{(Since $x^t=x^{t-1}+\eps\cdot\mathbf{1}_{S_r^t}$)}\nonumber\\
&\le\sum_{o\in O_S^t\cap O_L}F(\eps\mathbf{1}_{o}|x^{t-1}+\eps\mathbf{1}_{S_{i-1}^t}) &&\text{(By submodularity)}\nonumber\\
&\le|O_S^t\cap O_L|\cdot\frac{\eps^2 v}{(1-\eps)r} &&\text{(By Eq.~\eqref{eq:o_in_O_S_t})}\nonumber\\
&\le r\cdot\frac{\eps^2 v}{(1-\eps)r}=\frac{\eps^2 v}{1-\eps}.
\end{align*}
\paragraph{Case 2: $O_S^t\cap O_L=\emptyset$.} We note that Ineq.~\eqref{eq:O_S_t_cap_O_L_t} holds trivially if $O_S^t\cap O_L=\emptyset$.

\subsubsection*{Step 2: Upper bounding $\E[f(O_L^t|\cR(x^t))]$}
We denote $\ell:=|S_L^t|$. Since $S_L^t\subseteq S_r^t$ and given how $S_r^t$ is constructed in Subroutine~\ref{sub:continuous_greedy_filtering}, the elements in $S_L^t$ are $s_{i_1}^t,\dots,s_{i_{\ell}}^t$ for some indices $i_1<i_2<\cdots<i_{\ell}$ in $[r]$. Recall that, as shown in Eq.~\eqref{eq:h_t_maps_S_L_t_to_O_L_t}, there is a bijection $h_t:S_L^t\to O_L^t$ such that $(S_r^t\setminus\{s_{i_j}^t\})\cup\{h_t(s_{i_j}^t)\}\in\M$ for all $j\in[\ell]$, which implies that for all $j\in[\ell]$, $(S_{i_j}^t\setminus\{s_{i_j}^t\})\cup\{h_t(s_{i_j}^t)\}\in\M$, since $S_{i_j}^t\subseteq S_r^t$. This is equivalent to $S_{i_j-1}^t\cup\{h_t(s_{i_j}^t)\}\in\M$, since $S_{i_j-1}^t=S_{i_j}^t\setminus\{s_{i_j}^t\}$.

Furthermore, for any $j\in[\ell]$, since $h_t(s_{i_j}^t)\in O_L^t$ (and hence $h_t(s_{i_j}^t)\notin H$), element $h_t(s_{i_j}^t)$ did not satisfy the selection condition posed by $\Delta_{i_j}^t$ at Line~\ref{algline:offline_filter_condition} of Subroutine~\ref{sub:continuous_greedy_filtering}, despite $S_{i_j-1}^t\cup\{h_t(s_{i_j}^t)\}\in\M$. Thus, it must hold that $\smallfloor{F(\eps\mathbf{1}_{h_t(s_{i_j}^t)}|x^{t-1}+\eps\mathbf{1}_{S_{i_j-1}^t})}_I
\le\smallfloor{F(\eps\Delta_{i_j}^t|x^{t-1}+\eps\mathbf{1}_{S_{i_j-1}^t})}_I$, which implies that $\smallfloor{F(\eps\mathbf{1}_{h_t(s_{i_j}^t)}|x^{t-1}+\eps\mathbf{1}_{S_{i_j-1}^t})}_I\le\smallfloor{F(\eps\mathbf{1}_{s_{i_j}^t}|x^{t-1}+\eps\mathbf{1}_{S_{i_j-1}^t})}_I$, since $\Delta_{i_j}^t=\mathbf{1}_{S_{i_j}^t}-\mathbf{1}_{S_{i_j-1}^t}=\mathbf{1}_{s_{i_j}^t}$. Note that $F(\eps\mathbf{1}_{s_{i_j}^t}|x^{t-1}+\eps\mathbf{1}_{S_{i_j-1}^t})\ge0$, because of the if condition at Line~\ref{algline:offline_greedy_selection_if_condition} of Subroutine~\ref{sub:continuous_greedy_filtering}. Moreover, we have that $F(\eps\mathbf{1}_{h_t(s_{i_j}^t)}|x^{t-1}+\eps\mathbf{1}_{S_{i_j-1}^t})\le F(\eps\mathbf{1}_{h_t(s_{i_j}^t)})$ by submodularity, and that $F(\eps\mathbf{1}_{h_t(s_{i_j}^t)})=\eps\cdot f(\{h_t(s_{i_j}^t)\})\le\eps v$ (since $v$ is the value of the most valuable element in $[n]$), which together imply that $F(\eps\mathbf{1}_{h_t(s_{i_j}^t)}|x^{t-1}+\eps\mathbf{1}_{S_{i_j-1}^t})\le\eps v$. Hence, we can apply Claim~\ref{claim:streaming_rounding_error} (by setting $a$, $b$, $\gamma$, $w$ to be $F(\eps\mathbf{1}_{s_{i_j}^t}|x^{t-1}+\eps\mathbf{1}_{S_{i_j-1}^t})$, $F(\eps\mathbf{1}_{h_t(s_{i_j}^t)}|x^{t-1}+\eps\mathbf{1}_{S_{i_j-1}^t})$, $\frac{\eps}{r}$, $\eps v$), and we obtain that
\begin{equation}\label{eq:s_i_j_t_vs_o_i_j_t}
F(\eps\mathbf{1}_{s_{i_j}^t}|x^{t-1}+\eps\mathbf{1}_{S_{i_j-1}^t})\ge(1-\eps)\cdot F(\eps\mathbf{1}_{h_t(s_{i_j}^t)}|x^{t-1}+\eps\mathbf{1}_{S_{i_j-1}^t})-\frac{\eps^2v}{r}.
\end{equation}
Now we lower bound $F(\eps\mathbf{1}_{S_r^t}|x^{t-1})$ as follows,
\begin{align*}
&\,F(\eps\mathbf{1}_{S_r^t}|x^{t-1})\nonumber\\
=&\,\sum_{i=1}^{r} F(\eps\Delta_i^t|x^{t-1}+\eps\mathbf{1}_{S_{i-1}^t}) &&\text{(By telescoping sum)}\nonumber\\
\ge&\sum_{j=1}^{\ell} F(\eps\mathbf{1}_{s_{i_j}^t}|x^{t-1}+\eps\mathbf{1}_{S_{i_j-1}^t}) &&\text{(By the if condition at Line~\ref{algline:offline_greedy_selection_if_condition})}\nonumber\\
\ge&\sum_{j=1}^{\ell} (1-\eps)\cdot F(\eps\mathbf{1}_{h_t(s_{i_j}^t)}|x^{t-1}+\eps\mathbf{1}_{S_{i_j-1}^t})-\ell\cdot\frac{\eps^2v}{r} &&\text{(By Ineq.~\eqref{eq:s_i_j_t_vs_o_i_j_t})}\nonumber\\
\ge&\sum_{o\in O_L^t} (1-\eps)\cdot F(\eps\mathbf{1}_{o}|x^{t-1}+\eps\mathbf{1}_{S_{i_j-1}^t})-\ell\cdot\frac{\eps^2v}{r} &&\text{(Since $h_t:S_L^t\to O_L^t$ is a bijection)}\nonumber\\
\ge&\sum_{o\in O_L^t} (1-\eps)\cdot F(\eps\mathbf{1}_{o}|x^{t-1}+\eps\mathbf{1}_{S_r^t})-\ell\cdot\frac{\eps^2v}{r} &&\text{(By submodularity)}\nonumber\\
=&\sum_{o\in O_L^t} (1-\eps)\cdot F(\eps\mathbf{1}_{o}|x^t)-\ell\cdot\frac{\eps^2v}{r} &&\text{(Since $x^t=x^{t-1}+\eps\cdot\mathbf{1}_{S_r^t}$)}\nonumber\\
=&\sum_{o\in O_L^t} (1-\eps)\cdot \eps\cdot\E[f(\{o\}|\cR(x^t)\setminus\{o\})]-\ell\cdot\frac{\eps^2v}{r} &&\text{(By Definition~\ref{def:multi-linear})}\nonumber\\
\ge&\sum_{o\in O_L^t} (1-\eps)\cdot \eps\cdot\E[f(\{o\}|\cR(x^t))]-\ell\cdot\frac{\eps^2v}{r} &&\text{(By submodularity)}\nonumber\\
\ge&\,(1-\eps)\cdot \eps\cdot\E[f(O_L^t|\cR(x^t))]-\ell\cdot\frac{\eps^2v}{r} &&\text{(By submodularity)}\nonumber\\
\ge&\,(1-\eps)\cdot \eps\cdot\E[f(O_L^t|\cR(x^t))]-\eps^2v &&\text{(Since $\ell\le r$)},
\end{align*}
which implies that
\begin{equation}\label{eq:S_r_t_vs_O_L_t}
    \eps\cdot\E[f(O_L^t|\cR(x^t))]\le\frac{F(\eps\mathbf{1}_{S_r^t}|x^{t-1})+\eps^2v}{1-\eps}.
\end{equation}

\subsubsection*{Step 3: Combining the upper bounds}
Given the upper bounds which we have derived for $\E[f(O_S^t\cap O_L|\cR(x^t))]$ and $\E[f(O_L^t|\cR(x^t))]$, we can now bound $\E[f(O_L|\cR(x^t))]$ as follows,
\begin{align*}
&\,\eps\cdot\E[f(O_L|\cR(x^t))]\\
\le&\,\eps\cdot\E[f(O_S^t\cap O_L|\cR(x^t))]+\eps\cdot\E[f(O_L^t|\cR(x^t))]&&\text{(By Eq.~\eqref{eq:O_L_decomposition} and submodularity)}\\
\le&\,\frac{F(\eps\mathbf{1}_{S_r^t}|x^{t-1})+2\eps^2v}{1-\eps} &&\text{(By Ineq.~\eqref{eq:O_S_t_cap_O_L_t} and Ineq.~\eqref{eq:S_r_t_vs_O_L_t})}\\
=&\,\frac{F(x^{t})-F(x^{t-1})+2\eps^2v}{1-\eps} &&\text{(Since $x^t=x^{t-1}+\eps\cdot\mathbf{1}_{S_r^t}$)}\\
\le&\,\frac{F(x^{t})-F(x^{t-1})+2\eps^2\cdot f(O)}{1-\eps} &&\text{(Since $\textstyle v=\max_{e\in[n]}f(\{e\})\le f(O)$)}.
\end{align*}
After rearranging, we derive that
\begin{align}\label{eq:one_epoch_of_continuous_greedy}
    F(x^{t})-F(x^{t-1})&\ge(1-\eps)\cdot \eps\cdot\E[f(O_L|\cR(x^t))]-2\eps^2\cdot f(O) \nonumber\\
    &=(1-\eps)\cdot \eps\cdot(\E[f(O_L\cup\cR(x^t))]-\E[f(\cR(x^t))])-2\eps^2\cdot f(O) \nonumber\\
    &=(1-\eps)\cdot \eps\cdot(\E[f(O_L\cup\cR(x^t))]-F(x^t))-2\eps^2\cdot f(O) &&\text{(By Definition~\ref{def:multi-linear})}\nonumber\\
    &\ge(1-\eps)\cdot \eps\cdot(f(O_L)-\eps\cdot f(O)-F(x^t))-2\eps^2\cdot f(O) &&\text{(By event $E_2$)}\nonumber\\
    &\ge (1-\eps)\cdot \eps\cdot(f(O_L)-4\eps\cdot f(O)-F(x^t)) &&\text{(Since $\textstyle\eps<\frac{1}{4}$)},
\end{align}
where the penultimate inequality is because we condition on event $E_2$, as stated in Lemma~\ref{lem:offline_regular_case}. Specifically, recall that $E_2$ is the event that $f(X|O)\ge-\eps\cdot f(O)$ for all $X\subseteq\bigcup_{t\in\left[\frac{1}{\eps}\right],\,i\in[r]} V_i^t$. Notice that given how $x^t$ is constructed in Subroutine~\ref{sub:continuous_greedy_filtering}, the support of $x^t$ (i.e., the set of non-zero coordinates of $x^t$) is a subset of $\bigcup_{t\in\left[\frac{1}{\eps}\right],\,i\in[r]} V_i^t$. Hence, $\cR(x^t)$ is always a subset of $\bigcup_{t\in\left[\frac{1}{\eps}\right],\,i\in[r]} V_i^t$, which implies that $f(\cR(x^t)|O_L)\ge-\eps\cdot f(O)$, conditioned on event $E_2$.
\end{proof}
\subsubsection*{Step 4: Finishing the proof}
By rearranging Ineq.~\eqref{eq:one_epoch_of_continuous_greedy}, we get that
\[
    (1+\eps(1-\eps))\cdot(f(O_L)-4\eps\cdot f(O)-F(x^t))\le f(O_L)-4\eps\cdot f(O)-F(x^{t-1}).
\]
Then, by iteratively applying the above inequality from $t=1$ to $t=\frac{1}{\eps}$, we obtain that
\begin{align}\label{eq:1/eps_epoches_of_continuous_greedy}
&\,(1+\eps(1-\eps))^{\frac{1}{\eps}}\cdot(f(O_L)-4\eps\cdot f(O)-F(x^{1/\eps})) \nonumber\\
\le&\,f(O_L)-4\eps\cdot f(O)-F(x^0) \nonumber\\
=&\,f(O_L)-4\eps\cdot f(O) &&\text{(Since $x^0=0$)}.
\end{align}
Because $1+\alpha\ge e^{\alpha-\alpha^2}$ holds for all $\alpha\in[0,1]$, we have that $(1+\eps(1-\eps))^{\frac{1}{\eps}}\ge e^{(1-\eps)-\eps(1-\eps)^2}=e\cdot e^{-\eps(1+(1-\eps)^2)}\ge e\cdot(1-\eps(1+(1-\eps)^2))\ge e\cdot(1-2\eps)$. Putting this together with Ineq.~\eqref{eq:1/eps_epoches_of_continuous_greedy}, we get
\begin{align*}
    F(x^{1/\eps})&\ge\left(1-\frac{1}{e\cdot(1-2\eps)}\right)\cdot(f(O_L)-4\eps\cdot f(O))\\
    &=\left(1-\frac{1}{e}-\frac{2\eps}{e\cdot (1-2\eps)}\right)\cdot(f(O_L)-4\eps\cdot f(O))\\
    &\ge\left(1-\frac{1}{e}-\frac{\eps}{1-2\eps}\right)\cdot(f(O_L)-4\eps\cdot f(O))\\
    &\ge\left(1-\frac{1}{e}-2\eps\right)\cdot(f(O_L)-4\eps\cdot f(O)) &&\text{(Since $\textstyle\eps<\frac{1}{4}$)}\\
    &\ge\left(1-\frac{1}{e}\right)\cdot(f(O_L)-4\eps\cdot f(O))-2\eps\cdot f(O) &&\text{(Since $f(O)\ge f(O_L)$)}\\
    &\ge\left(1-\frac{1}{e}\right)\cdot (f(O_L)-8\eps\cdot f(O)),
\end{align*}
which completes the proof of Lemma~\ref{lem:offline_regular_case}.

\end{document}